\documentclass[11pt]{article}
\usepackage{graphicx,color,amsfonts,amsmath,amssymb,enumerate}
 \usepackage{url,fancyhdr,indentfirst}
   \usepackage{amsthm,natbib,comment,bm}
    \usepackage[colorlinks=true,citecolor=blue]{hyperref}
\usepackage{dsfont}
\usepackage{mathrsfs}
\usepackage{tikz-qtree}
\usepackage{xr}
\usepackage{threeparttable}
\usepackage[title]{appendix}

\usepackage{multirow} 
\usepackage{booktabs} 
\usepackage{caption}

\def\d{\mathrm{d}}
\def\laweq{\buildrel \d \over =}

\def\lawcn{\buildrel \d \over \rightarrow}

\newcommand{\VaR}{\mathrm{VaR}}

\newcommand{\ES}{\mathrm{ES}}
\newcommand{\ex}{\mathrm{ex}}
\newcommand{\E}{\mathbb{E}}
\newcommand{\R}{\mathbb{R}}

\newcommand{\M}{\mathcal{M}}

\newcommand{\N}{\mathbb{N}}
\newcommand{\p}{\mathbb{P}}

\newcommand{\id}{\mathds{1}}

\newcommand{\X}{\mathcal X}

\renewcommand{\ge}{\geqslant}
\renewcommand{\le}{\leqslant}

\renewcommand{\epsilon}{\varepsilon}

\theoremstyle{plain}
\newtheorem{theorem}{Theorem}

\newtheorem{assumption}{Assumption}

\newtheorem{lemma}{Lemma}
\newtheorem{proposition}{Proposition}

\theoremstyle{definition}
\newtheorem{definition}{Definition}
\newtheorem{example}{Example}

\theoremstyle{remark}
\newtheorem{remark}{Remark}

\theoremstyle{definition}

\renewcommand{\cite}{\citet}

\newcommand{\tbl}{\textcolor{black}}

\usepackage{tikz}
\usetikzlibrary{arrows.meta,positioning}

\usepackage[onehalfspacing]{setspace}

\begin{document}

\title{Model Aggregation for Risk Evaluation  and Robust Optimization}
\author{Tiantian Mao\thanks{Department of Statistics and Finance,
University of Science and Technology of China, China. E-mail: tmao@ustc.edu.cn}
\and
Ruodu Wang\thanks{Department of Statistics and Actuarial Science, University of Waterloo, Canada. E-mail: wang@uwaterloo.ca}
\and
Qinyu Wu\thanks{Department of Statistics and Actuarial Science, University of Waterloo, Canada. E-mail: q35wu@uwaterloo.ca}
}

\date{\today}

\maketitle

\begin{abstract}
	We introduce a new approach for prudent risk evaluation based on stochastic dominance, which will be called the model aggregation (MA) approach. In contrast to the classic worst-case risk (WR) approach, the MA approach produces not only  a robust value of  risk evaluation but also  a robust distributional model, independent of any specific risk measure. The MA risk evaluation can be computed through explicit formulas in the lattice theory of stochastic dominance, and 
 under some standard assumptions, the MA robust optimization admits a convex-program  reformulation. The MA  approach for  Wasserstein and  mean-variance uncertainty sets admits  explicit formulas for the obtained robust models.
 Via an equivalence property between the MA and the WR approaches, new axiomatic characterizations are obtained for the Value-at-Risk (VaR) and the Expected Shortfall (ES, also known as CVaR).  The new approach is illustrated with various risk measures and  examples from portfolio optimization.

\end{abstract}

\section{Introduction}
Modern risk management often requires the evaluation of risks under multiple scenarios.
The risk evaluation obtained under various scenarios needs to be aggregated properly, and a prudent approach is often implemented in practice.
As a prominent example,  in the Fundamental Review of the Trading Book (FRTB) of Basel IV (\cite{BASEL19}), banks need to evaluate the market risk of their portfolio losses under stressed scenarios, in particular including a model generated from  data during the financial crisis of  2007, and the obtained risk models are aggregated via a worst-case approach; see \citet[Section 1]{WZ18} for a description of the stressed scenarios and the model aggregation in the FRTB, and \cite{CF17} for techniques to generate regulatory scenarios. 
In the literature of portfolio risk assessment and optimization, the worst-case approach is popular; we refer to \cite{G03},  \cite{NPS08}, \cite{ZF09} and \cite{GX14} for robust portfolio optimization, and  \cite{EPR13} and \cite{WPY13}  for robust risk aggregation.
This idea further leads to many studies in distributionally robust optimization; see 
\cite{DY10},  \cite{GK16}, \cite{EK18} and  \cite{BM19}.
 In this paper, we will work in the context where a prudent risk evaluation under multiple models, which is our main focus, is desirable.\footnote{This assumption is natural in a regulatory setting such as the FRTB, where risk measures are heavily used; see also the above mentioned references. Other ways to aggregate risk models, such as averaging, max-min, smooth aggregation (\cite{KMM05}) and anti-conservative (e.g., best-case) approaches, may be suitable in different contexts, and they are outside the scope of the current paper.}

A natural question for risk management in the presence of model uncertainty is how to generate a robust model from a collection of models resulting from statistical and machine learning procedures,  operational considerations, or expert's opinion.
 Such a robust model can be used for risk evaluation, simulation, optimization, and decision making. 

Our main ideas to address this question are described below.
Let $\mathcal M$ be a set of distributions on $\R$, representing possible risk models; for illustrative purposes, we focus on one-dimensional financial losses for which the theory of risk measures is rich.
Suppose that a risk analyst evaluates a random loss using different   methodologies, scenarios or data sets,  and obtains a collection $\mathcal F\subseteq \mathcal M$ of distributional models. Here, the number of models in $\mathcal F$  may be finite or infinite.
For instance,  $\mathcal F$ may contain distributions of the random loss under different probability measures (economic scenarios), estimation methods, or values of statistical parameters; alternatively,   $\mathcal F$ may represent   distributions from   losses which may occur from different  possible decisions from a business competitor.
The set $\mathcal F$ will be called an \emph{uncertainty set}.
The distributions  in $\mathcal F$ will be used to assess the risk, together with a  risk measure $\rho:\mathcal M\to \R$, such as a Value-at-Risk (VaR) or an Expected Shortfall (ES, also known as CVaR); see Section \ref{sec:22} for their definitions.
Prudent regulation and risk management require  a conservative approach which aggregates the above information.
There are two conceptually intuitive ways to generate a robust assessment of the risk:
\begin{enumerate}[(i)]
	\item Directly calculate the maximum (or supremum) of  $\rho(F)$ over $F\in \mathcal F$.
	\item Calibrate  a robust (conservative) distributional model $F^*$   based on  $\mathcal F$, and calculate $\rho(F^*)$.
\end{enumerate}

Arguably, each of (i) and (ii) is a reasonable approach to take, but they may yield different risk evaluations.
We shall call (i) the  \emph{worst-case risk (WR) approach},
and (ii) the  \emph{model aggregation (MA) approach}.
There are two obvious advantages of the MA approach: We obtain a robust model which is useful for analysis and simulation, thus answering the motivating question above, and the procedure applies for generic risk measures, not only a specific one.
Other less obvious, but important, advantages of the MA approach will be revealed through this paper.   
{\color{black} The model  $F^*$ is robust in two senses: First, it is  more conservative than any models in $\mathcal F$; second, it applies to a wide range of risk measures or decision criteria. }

At this point, we have not yet specified how the robust distributional model $F^*$ may be obtained in the MA approach (ii). For this purpose, we need an order relation,  often consistent with the risk measure $\rho$ used by the risk analyst. We will describe some natural choices of partial orders, in particular, first- and second-order stochastic dominance, in Section \ref{sec:21}.

Our main objective   is a comprehensive theory on the  two approaches of robust risk evaluation, with a focus on the newly introduced MA approach.
The following questions naturally arise.
\begin{enumerate}[Q1.]
	\item What are the advantages and disadvantages of  the MA approach in contrast to the WR approach, in addition to the points mentioned above?\label{Q1}
	\item \textcolor{black}{What are theoretical and computational properties  of the MA approach  in distributionally robust optimization?}\label{Q2}
	\item  \textcolor{black}{Which risk measures yield equivalent robust risk evaluation results via the MA and WR approaches and how are they used in regulatory practice}?\label{Q3}
	\item How is the MA approach implemented in common settings of uncertainty, optimization, and real-data applications?\label{Q4}
\end{enumerate}
We will answer the four questions above  by means of  several novel theoretical results.
Our main contributions can be explained {as} follows.
After introducing partial orders and risk measures in Section \ref{sec2},
we present a rigorous formulation of the MA and WR approaches for risk evaluation  and distributionally robust  optimization in Section \ref{sec:2}. These optimization problems will be called WR or  MA optimization. 
Our new method is related to stochastic optimization of risk measures (e.g., \cite{DR04}, \cite{RS06}, \cite{SDR21}); see our discussion in Section \ref{sec:2}.

 \textcolor{black}{Features of MA risk evaluation will be discussed in Section \ref{sec:4} and 
 MA robust optimization will be studied in Section \ref{sec:34}.
We show convenient  properties of the MA approach in risk evaluation and optimization.
In particular, the MA risk evaluation remains convex when the risk measure is convex,
and the MA robust optimization admits a convex program reformulation under suitable conditions. 
This answers Q\ref{Q2}, and also Q\ref{Q1} partially.}

We establish in Section \ref{sec:charac}   that the property of  equivalence in model aggregation characterizes VaR and ES 
among very general classes of risk measures.
The equivalence property identifies for which risk measures  the two approaches can be converted to each other. 
Through these   results, which require long technical proofs, the rich literature of robust risk evaluation and optimization, popular in operations research,\footnote{In addition to the literature on portfolio optimization, robust risk evaluation and optimization also broadly exist in other applications of operations research; see \cite{WKS14}, \cite{EK18}, \cite{BKM19} and \cite{ESW21} for a small specimen.} is connected to that of the axiomatic theory of risk preferences, popular in decision theory,\footnote{For developments on axiomatic studies in decision theory, see e.g., \cite{KMM05}, \cite{MMR06} and \cite{CHMM21}. Axiomatic theory of risk measures have also been an active topic in quantitative finance since  the seminal work of \cite{ADEH99}; see \cite{FS16} for a comprehensive treatment.} for the first time.
Our results contribute to the latter literature by offering new axiomatizations of both VaR and ES which are important issues in risk management in themselves.\footnote{In particular,  \cite{C09} obtained an axiomatization of VaR and  \cite{WZ20} obtained an axiomatization of ES; see also Remarks \ref{rem:characterizeVaR} and \ref{rem:characterizeES} for other axiomatizations of VaR and ES.} 
These results answer question Q\ref{Q3} above.

We address two settings of uncertainty, those generated by  Wasserstein metrics  and those generated by moment information in Section \ref{sec:US}.
We illustrate that the MA approach leads to closed-form robust distributional models in these settings, being easy to apply and computationally feasible.  \textcolor{black}{Based on a new result on dimension reduction for Wasserstein balls (Theorem \ref{th:7}), we show that the MA approach can conveniently handle multivariate Wasserstein uncertainty in the setting of portfolio selection.}
Section \ref{sec:Numerical} contains two applications of worst-case risk evaluation and portfolio selection under uncertainty using real financial data.  These two sections answer Q\ref{Q4}.

Finally,
advantages and limitations  of the MA approach, as well as directions for future work, are summarized and discussed in Section \ref{sec:conclude}, which also contains  a preliminary discussion  on aggregating multivariate risk models, in contrast to the univariate risk models  treated throughout the paper. These discussions address Q\ref{Q1} at a high level.

In the main text of the paper, we focus on the set of distributions with finite mean to make our analysis concise and managerial insights clear. More general choices of the space of distributions are treated in the   appendices, which also contain  technical proofs of all  results.

\section{Preliminaries and standing notation} \label{sec2}
We first {\color{black} introduce} some notation.
Let $(\Omega,\mathcal B,\p)$ be a nonatomic probability space. 
{\color{black}
For $d\in\N$, let $\mathfrak B(\R^d)$ be the Borel $\sigma$-field on $\R^d$. A random vector $\mathbf X$ is a measurable mapping from $(\Omega,\mathcal B)$ to $(\R^d,\mathfrak B(\R^d))$. 
Denote by $\mu_{\mathbf X}:=\p\circ \mathbf X^{-1}$ the probability distribution induced by a random vector $\mathbf X$ under $\p$, where $\mathbf X^{-1}$ is the inverse image of $\mathbf X$. 
Denote by $F_{\mathbf X}$ the cumulative distribution function (cdf) of $\mathbf X$ under $\p$, i.e., $F_{\mathbf X}(\mathbf x)=\p(\mathbf X\le \mathbf x)$ for $\mathbf x\in\R^d$, where the inequality is component-wise. 
We will take cdfs $F_{\mathbf X}$,   identified with distributions $\mu_{\mathbf X}$, as the main research object of the entire paper. 
We use $\delta_{\mathbf t}$ to represent the point-mass at $\mathbf t\in \R^d$.
Let $L^1$ be the space of 
all integrable random variables on $(\Omega,\mathcal B,\p)$, 
where almost surely equal random variables are treated as identical.
Denote by $\mathcal M_1$
the set of cdfs of all random variables in $L^1$, i.e., $\mathcal M_1$ is the set of all cdfs $F$ satisfying $\int_{\R}|x|\d F(x)<\infty$.
The mean of a random variable $X$  under $\p$ is written as $\E[X]=\mathfrak{m}(F_X)$, where $\E$ is defined on $L^1$ and $\mathfrak m$ is defined on $\mathcal M_1$. 
We denote  by  $\Delta_d$ the standard simplex $\{\bm\lambda\in [0,1]^d: \sum_{i=1}^d \lambda_i=1\} $ in $\R^d$ and by $[d]=\{1,\dots,d\}$.
}


\subsection{Stochastic orders and lattices}\label{sec:21}

For any set of cdfs $\mathcal M$, 
	let $\preceq$ be a partial order on $\mathcal M$, and $(\M,\preceq)$ is called a partially ordered  set.
The relevant tool is the lattice theory which we collect in Appendix \ref{app:lattice}, and here we only present  a basic result needed to understand our main ideas.  The most commonly used partial orders in finance and economics are the first-order stochastic dominance $\preceq_1$ and the increasing convex order $\preceq_2$, defined as, for $F,G \in \mathcal M$,\footnote{
	Note that we treat  $F$ and $G$ as loss cdfs instead of wealth cdfs, and hence a larger element in $\preceq_1$ or $\preceq_2$ means higher risk, see e.g., \cite{SS07}.  
	For a cdf $F$, denote by $\widetilde F$ the cdf of $-X$ where $X$ follows $F$.
	Up to a sign flip,
	increasing convex order (for loss distributions) is economically equivalent to second-order stochastic dominance (for gain distributions) in the sense that  $F \succeq_2 G$ 
	if and only if $\widetilde F \preceq_{\rm ssd} \widetilde G$, 
	where $\preceq_{\rm ssd}$ is defined via (b) by changing increasing convex functions to increasing concave functions.
	This is because $u$ is increasing convex if and only if $x\mapsto -u(-x)$ is increasing concave. 
}
\begin{enumerate}[(a)]
\item  $F\preceq_1 G$  if   $\int u\d F \le \int u \d G$ for all increasing functions $u$;
\item
$F\preceq_2 G$  if   $\int u\d F \le \int u \d G$ for all increasing convex functions $u$.
\end{enumerate}
Other useful equivalent definitions of $\preceq_1$ and $\preceq_2$ are put in Appendix
\ref{app:lattice}.
To build a robust distributional model, we need to define the supremum of a set $\mathcal F \subseteq \mathcal M$. Throughout, $\mathcal F$ is  assumed nonempty.
For a partial ordered set $(\M,\preceq)$ and $\mathcal F\subseteq \mathcal M$, the supremum of $\mathcal F$, denoted by $\bigvee \mathcal F$, is defined as the distribution $H$ satisfying
$H\in\mathcal M$ and
$F \preceq H\preceq G$ for all $F\in \mathcal F$ and all $G\in \M$ which dominates every element of $\mathcal F$ (uniqueness is guaranteed by definition).  If there exists $G\in \M$ such that $G$ dominates every element of $\mathcal F$, we say that $\mathcal F$ is \emph{bounded from above with respect to $\preceq$},  denoted by $\preceq$-bounded. The supremum does not always exist, but for the two choices of ordered sets $(\mathcal M_1,\preceq_1)$ and $(\mathcal M_1,\preceq_2)$ that we consider in the main paper, this does not  create any problem; see e.g., \cite{KR00} and \cite{MS06} for the lattice structure of cdfs with $\preceq_1$ and $\preceq_2$. 
The two cases of $\bigvee \mathcal F$ for  $\preceq_1$ and $\preceq_2$  admit  explicit formulas,  given in Proposition \ref{th-sharp} in Section \ref{sec:4}.  

\subsection{Risk measures}\label{sec:22}

In  the classic framework of \cite{ADEH99} and \cite{FS16},
a risk measure is traditionally defined as a mapping from a set $\X$ of random losses to ${\R}$.
Denote by $\mathcal M $ the set of cdfs of random variables in $\mathcal X$. {\color{black}  For a partial order $\preceq$ on $\mathcal M$, a natural interpretation of $F\preceq G$ is that $G$ is riskier than $F$ according to  $\preceq$. A risk measure $\rho:\M\to \R$  is {\it $\preceq$-consistent} if $\rho(F)\le \rho(G)$ for all $F,G\in \M$ with $F\preceq G$.}

\begin{definition}\label{def-risk}
	A \emph{distribution based risk measure} is a mapping $\rho:\M \to \R$ {\color{black} satisfying $\preceq_1$-consistency}. 
	For such $\rho$, its associated \emph{random-variable based risk measure} is  $\widetilde\rho:\X\to \R$ given by $\widetilde \rho(X)=\rho(F_X)$. 
	Both $\rho$ and $\widetilde\rho$ will be called risk measures in this paper.
\end{definition}

Note that $\rho$ is  $\preceq_1$-consistent if and only if $\widetilde \rho$ is monotone (i.e., $\widetilde \rho(X)\le \widetilde \rho(Y)$ when $X\le Y$).
The random-variable based risk measure in Definition \ref{def-risk} satisfies 
\emph{law-invariance} (i.e., $\widetilde\rho(X)=\widetilde\rho(Y)$ whenever $F_X=F_Y$).
There exists a one-to-one correspondence between $\rho:\M\to \R$ 
and ${\color{red}\tilde{\rho}}:\X\to \R$ satisfying 
law-invariance;
see e.g., Proposition 1 of \cite{DKW19}.
We will choose $\X=L^1$ in the main part of the paper, so that the two partial orders $\preceq_1$ and $\preceq_2$ both behave well.\footnote{In particular, it is well known that $\preceq_2$ is closely related to mean-preserving spreads of \cite{RS70}, and a finite mean is essential for such a connection. On the other hand, $\preceq_1$ fits well in any space of random variables or cdfs.} 
For a better exposition of distributional uncertainty, we will  present ideas and results mainly using $\rho$ instead of $\widetilde \rho$.


The two most popular and important risk measures in financial practice,  VaR and ES,
are both law-invariant. The risk measure VaR at level $\alpha\in (0,1)$ is the functional $\VaR_\alpha:\mathcal M_1 \to \mathbb{R}$ defined by
$$
\VaR_\alpha (F)= \inf \{x\in \R: F(x) \ge \alpha\},
$$
which is
the left $\alpha$-quantile of  a cdf.  The risk measure ES at level $\alpha\in[0,1)$ is the functional $\ES_\alpha:\mathcal M_1 \to \mathbb{R}$ defined by
$$
\ES_{\alpha}(F)=\frac{1}{1-\alpha}\int_\alpha^1 {\rm VaR}_{s}(F)\d s,
$$
and in particular, $\ES_0(F)=\mathfrak{m}(F)$.
We can also define $\VaR_0$, $\VaR_1$ and $\ES_1$, which are not finite-valued on $\M_1$; see Appendix \ref{app:A}.
Note that in addition to the $\preceq_{1}$-consistency of 
VaR and ES,  ES also satisfies  the  $\preceq_{2}$-consistency. 

\section{Introducing the MA approach}\label{sec:2}
\label{sec:31}
We   describe the two approaches for robust risk evaluation,   the primary objects of this paper.
For a  risk measure   $\rho:\mathcal M\to \R$ and an uncertainty set $\mathcal F\subseteq \mathcal M$, a common way to obtain a  robust risk evaluation is to calculate the following worst-case risk measure 
\begin{align}
	{\rm WR:} ~~~~~\rho^{\rm WR}(\mathcal F) = \sup_{F\in\mathcal F}\rho(F). \label{eq-WC}
\end{align}
The value in  \eqref{eq-WC} is called the \emph{WR robust $\rho$ value}, and it has been widely studied in the literature; some references are mentioned in the introduction.
Next, we  propose a new method of robust risk evaluation, that is, assuming that the supremum $\bigvee \mathcal F$ with respect to $\preceq$ exists,
\begin{align}\label{eq-MA}
	{\rm MA:} ~~~~~\rho^{\rm MA}(\mathcal F) = \rho\left(\bigvee \mathcal F\right),
\end{align}
and $\rho^{\rm MA}(\mathcal F)=\infty$ if $\mathcal F$ is not bounded from above.  
The value in \eqref{eq-MA} is called the \emph{$\preceq$-MA robust $\rho$ value}
(``$\preceq$-" will be omitted if the  order is clear from the context). 
In the main text of the paper,
$\bigvee \mathcal F$ exists for all $\mathcal F$ bounded from above, and
hence,
$\rho^{\rm MA}$ is always well-defined.
In   case that $\bigvee \mathcal F$ may not exist, \eqref{eq-MA} needs to be modified as   in Appendix \ref{app:lattice}.

{\color{black}The idea of the MA approach can be described in two steps}: First, take the supremum  $\bigvee \mathcal F$ of the uncertainty set $\mathcal F$ as the robust distribution, and second, calculate the value of the risk measure of the robust distribution. 
The robust distribution $\bigvee \mathcal F$ obtained in the first step can be used for any risk measure.
If, in addition, the risk measure $\rho$ is $\preceq$-consistent, then the MA approach produces a larger robust risk value than the WR approach,
that is, for  any $\mathcal F\subseteq \M$,
\begin{align}\label{eq-gap}
	\rho^{\rm WR}(\mathcal F) \le  \rho^{\rm MA}(\mathcal F),
\end{align}
since $\preceq$-consistency implies   $\rho(\bigvee \mathcal F)\ge \rho(F)$ for all $F\in \mathcal F$.
The MA approach can be implemented even in case no risk measure is involved (thus skipping the second step above), as the model $\bigvee \mathcal F$ is ready to use without a specification of any specific objective.


In the sequel, we will focus on $\preceq_1$ and $\preceq_2$. {\color{black} For a simpler notation, we write ${\rm MA}_1$
when $\preceq$ is specified as $\preceq_1$, and ${\rm MA}_2$ is similar. 
}
For these two stochastic orders,
the explicit forms of $\bigvee\mathcal F$  are  obtained in Section \ref{sec:4}. It is also worth noting that if $\rho$ is consistent with more than one partial orders, then the MA approach with a stronger partial order leads to a higher risk evaluation.
For instance, if $\rho$ is both $\preceq_1$-consistent and $\preceq_2$-consistent, then
$
\rho^{\rm WR}(\mathcal F) \le 
\rho^{{\rm MA}_2}(\mathcal F) \le \rho^{{\rm MA}_1}(\mathcal F)
$
because any $(\M,\preceq_{1})$-upper bound on $\mathcal F$ is also an $(\M,\preceq_{2})$-upper bound  on $\mathcal F$.

{\color{black} Using the WR and MA approaches,   two types of distributionally robust optimization problems arise: 
\begin{equation}
	\label{eq:robustopt0}
	\min_{\mathbf a\in A}~~ \rho^{\rm WR} (\mathcal F_{\mathbf a,f} ) \mbox{~~~~~and~~~~~}\min_{\mathbf a\in A} ~~\rho^{\rm MA} (\mathcal F_{\mathbf a,f}),
\end{equation}
where $A$ is a set of possible actions, $f:A\times \R^d  \to \R$ is a loss function, $\mathcal F$ is a set of cdfs on $\R^d$ and
\begin{equation}\label{eq-mulUS}
\mathcal F_{\mathbf a,f}=\left\{F_{f(\mathbf a,\mathbf X)}:  F_{\mathbf X}\in \mathcal F\right\}.
\end{equation}
The set $\mathcal F_{\mathbf a,f}$ consists all univariate cdfs of $f(\mathbf a,\mathbf X)$ where $\mathbf X$ has cdf in $\mathcal F$. 
For instance, by choosing $A\subseteq \R^d$ and $f(\mathbf a,\mathbf x)= \mathbf a^\top \mathbf x$, one arrives at the setting of robust portfolio selection, where $\mathbf a$ represents the vector of portfolio weights and $\mathbf x$ represents the vector of losses from individual assets.  
 The WR robust optimization problem in \eqref{eq:robustopt0} is equivalent to a minimax problem:
\begin{equation}\label{prob:minmax}
\min_{\mathbf a\in A}~ \sup_{F\in\mathcal F}\widetilde{\rho}^F(f(\mathbf a,\mathbf X)),
\end{equation}
where $\widetilde{\rho}^F(f(\mathbf a,\mathbf X))$ represents the value of $\widetilde{\rho}(f(\mathbf a,\mathbf X))$ when $\mathbf X$ has the cdf $F$.
If $\rho$ is $\preceq$-consistent, then the MA robust optimization problem in \eqref{eq:robustopt0} can be converted to a stochastic program
with    partial order $\preceq$ constraints: 
\begin{align}\label{prob:MA}
\min_{\mathbf a\in A} ~\inf_{H }\rho(H)
~~~~~~{\rm s.t.}~G\preceq H,~~\forall G\in\mathcal F_{\mathbf a,f}.
\end{align}
 The above problem with $\preceq$ being $\preceq_2$ will be studied  in Section \ref{sec:34}. Section \ref{sec:US} is dedicated to the portfolio selection problem, where two specific settings of uncertainty will be considered.

\textbf{Comparison: Optimization with stochastic dominance.}
\cite{DR04} introduced an optimization
problem with stochastic dominance constraint. 
Adapting to our notation, and focusing on $\preceq_2$, their model can be described as  
\begin{align}\label{prob:DR04}
\min_{\mathbf a\in A}~~\E[f(\mathbf a, \mathbf X)]~~~~~~{\rm s.t.}~F_{g_i(\mathbf a, \mathbf X)}\preceq_2 F_i,~~ i\in [m],
\end{align}
where $\mathbf X$ has a fixed cdf,  $f,g_1,\dots,g_m$ are fixed functions,  and $F_1,\dots,F_m\in \mathcal M_1$. 
Note that if we replace $\E$ by $\widetilde \rho$ 
and set $g_1=\dots=g_m=f$, then the problem is 
\begin{align}\label{prob:DR04-2}
\min_{\mathbf a\in A}~~\widetilde \rho( f(\mathbf a, \mathbf X)) ~~~~~~{\rm s.t.}~F_{f(\mathbf a, \mathbf X)}\preceq_2 F_i,~~ i\in [m].
\end{align}
The   problem \eqref{prob:DR04-2} is similar to our MA problem \eqref{prob:MA}, but there are a few essential differences. First, for fixed $\mathbf a\in A$, the cdf of  ${f(\mathbf a, \mathbf X)}$ is fixed in \eqref{prob:DR04-2},
whereas   \eqref{prob:MA} searches for a robust model for $ f(\mathbf a, \mathbf X)$ over the uncertainty set $\mathcal F_{\mathbf a,f}$.
Second, the direction of stochastic dominance is flipped, as our $H$ dominates every $G$ in $\mathcal F_{\mathbf a,f}$ 
and their $F_{f(\mathbf a, \mathbf X)}$ is dominated by every $F_i$.  
Note that the interpretation of stochastic dominance is very different here: 
\eqref{prob:MA} looks at risks larger in $\preceq_2$ (riskier) because our objective is robust optimization,
whereas \eqref{prob:DR04-2} looks at risks  smaller in $\preceq_2$ (safer) because of risk constraint. 
In \cite{DR03}, the following problem has been considered:
 $
\min_{X\in \mathcal C} \widetilde \rho(X)~{\rm s.t.}~F_{X}\preceq_2 F_i,~ i\in[m],
$ 
where $\mathcal C$ is a set of random variables. 
When $\rho$ is $\preceq_2$-consistent,  our MA risk evaluation problem  can be written as
 $ \min_{X\in L^1} \widetilde \rho(X)~{\rm s.t.}~F_{i}\preceq_2  F_X,~ i\in[m],$
which is similar to the model of \cite{DR03}; see the recent work of \cite{D23} for a related two-stage optimization with stochastic dominance constraints. 

\begin{remark}
In this paper,  we focus on risk measures taking real values.
In some applications, 
risk measures may be multi-valued or set-valued (e.g., \cite{EP06}, \cite{HH10}, \cite{HHR11}). For such risk measures, the  MA approach, together with multivariate stochastic order, can also be applied.
\end{remark}

}

\section{MA approach in  risk evaluation}
\label{sec:4}
In this section, we study properties of the MA risk evaluation by focusing on $\preceq_1$ and $\preceq_2$.

\subsection{Computing the robust model}
\label{sec:32}

We use $\bigvee_1 \mathcal F$ and $\bigvee_2 \mathcal F$ to represent the supremum of the uncertainty set $\mathcal F$ on the ordered set $(\mathcal M_1,\preceq_1)$ and $(\mathcal M_1,\preceq_2)$, respectively, and  $\pi_{F_X}$ represents the \emph{{\color{black} integrated} survival function} of $X\in L^1$, defined as
\begin{align}\label{eq-pi}
	\pi_{F_X}(x)=\int_x^{\infty}(1-F_X(t)){\d}t=\E[(X-x)_+],~~~x\in\R,
\end{align}
{where $x_+=\max\{0,x\}$ for $x\in\R$. It is straightforward from \eqref{eq-pi} that a simple relationship between the integrated survival function  $\pi_F$ and the cdf $F$ is $F=1+(\pi_F)_+'$, where $(\pi_F)_+'$ is the right derivative of $\pi_F$. 
\textcolor{black}{The left quantile function of $F\in\mathcal M_1$ is defined by $F^{-1}(\alpha)=\inf\{x\in \R: F(x)\ge \alpha\}$ for $\alpha\in(0,1]$, which is   $\VaR_\alpha (F)$ when $\alpha\in(0,1)$. 
The functions $\pi_F$ and $F^{-1}$ will be used throughout the paper.}

\begin{proposition} \label{th-sharp}
	\begin{itemize} \item [(a)] For a set $\mathcal F\subseteq \mathcal M_1$ that is $\preceq_1$-bounded, we have $\bigvee_1 \mathcal F=\inf_{F\in\mathcal F} F$\footnote{Note that the infimum of upper semicontinuous functions  $F\in\mathcal F$ is again upper semicontinuous and  thus a valid cdf when $\mathcal F$ is $\preceq_1$-bounded.} and the left quantile function of $\bigvee_1\mathcal F$ is $\sup_{F\in\mathcal F} F^{-1}$.
		\item [(b)] For a set $\mathcal F\subseteq\mathcal M_1$ that is $\preceq_2$-bounded, we have
		\begin{equation}
			\pi_{\bigvee_2\mathcal F}=\sup_{F\in\mathcal F}\pi_F,  \mbox{~and thus~}  \bigvee_2 \mathcal F=1+\left(\sup_{F\in\mathcal F}\pi_F\right)_+',\label{eq:computeV2}
		\end{equation}
		{\color{black}where $g'_+$ denotes the right derivative of $g$.}
	\end{itemize}
\end{proposition}
{\color{black}\begin{remark}Under the condition of Proposition \ref{th-sharp} (b), the function $  \bigvee_2 \mathcal F=1+\left(\sup_{F\in\mathcal F}\pi_F\right)_+'$ is a well-defined cdf  by noting that $\pi_F$ is a decreasing  convex function, and thus, its right derivative  is well-defined, nonnegative, and right continuous. More details can be found in the proof of Proposition EC.1.
\end{remark}}

\tbl{By Proposition \ref{th-sharp}, for a risk measure $\rho$, 
	the evaluation of $\mathrm{MA}_1$ 
	is equivalent to applying $\rho$ to the distribution with a worst-case quantile function. 
	Similarly,  the evaluation of $\mathrm{MA}_2$ 
	is equivalent to applying $\rho$ to  the distribution with a worst-case upper partial moment given by \eqref{eq:computeV2}. 
	Worst-case quanitle and worst-case 
	upper partial moment functions
	have wide range of applications in optimization; see e.g., \cite{G03}, \cite{L87}, \cite{N10} and \cite{C11}.} 
To compute $\pi_{\bigvee_2 \mathcal F}$ numerically, as  a decreasing convex function,  $\pi_{\bigvee_2 \mathcal F}$   can be well approximated by a piece-wise linear function with finitely many pieces (which requires computing   it    at finitely many points). 
The next result  concerns convexity of the uncertainty set. 
\begin{proposition}\label{prop-EMAM}
	Suppose that $i\in\{1,2\}$. For $\mathcal F\subseteq\mathcal M_1$, we have $\bigvee_i {\rm conv}{\mathcal F} =\bigvee_i\mathcal F$, where ${\rm conv}{\mathcal F}$ is the convex hull of $\mathcal F$.
\end{proposition}

Proposition \ref{prop-EMAM} illustrates that for the ${\rm MA}_1$ and ${\rm MA}_2$ approaches,
one can convert freely between any uncertainty set and its convex hull.
For WR risk evaluation, this does not hold in general.

\subsection{VaR and ES}
\label{sec:33}
Next we discuss the MA and WR approaches applied to   VaR and ES, and this will help us understand the inequality \eqref{eq-gap}.
The case of VaR, coupled with the partial order $\preceq_1$, is simple. By Proposition \ref{th-sharp}, for $\alpha\in (0,1)$ and any $\mathcal F$ that is $\preceq_1$-bounded,  $\VaR_\alpha ^{\rm WR}(\mathcal F) =  \VaR_\alpha^{{\rm MA}_1} (\mathcal F) $, and thus \eqref{eq-gap} holds as an equality in this specific setting; this   result will be collected in Theorem \ref{prop-cxES} below.

The case of ES is more illuminating.  Note that ES is consistent with respect to both $\preceq_1$ and $\preceq_2$.
First, we consider the MA approach with $\preceq_1$.
Since $$\ES_\alpha  ^{\rm WR}(\mathcal F) = \frac 1{1-\alpha}\sup_{F\in \mathcal F}  \int_\alpha ^1 F^{-1}(s)\d s
\le \frac 1{1-\alpha} \int_\alpha ^1 \sup_{F\in \mathcal F}  F^{-1}(s )\d s = \ES_\alpha^{{\rm MA}_1} (\mathcal F),$$
for \eqref{eq-gap} to hold as an equality, one needs to exchange the order of a supremum and an integral. Such an exchange, if legitimate, means that there exists $F\in \mathcal F$ such that $F^{-1}(s)\ge G^{-1}(s)$ for all $G\in \mathcal F$ and $s\in (\alpha,1)$, which is  a quite strong assumption unlikely to hold in applications.

Next, we consider the MA approach for ES with $\preceq_2$.
Recall a representation of $\ES_\alpha$ for $\alpha \in(0,1)$ in the celebrated work of \cite{RU02} and \cite{P00}, 
that is,
\begin{align} \ES_{\alpha}(F)
	=  \min_{x\in\R}\left\{x+\frac{1}{1-\alpha}\pi_{F}(x)\right\},~~~F\in \M_1.\label{eq:RU02}
\end{align}
Using \eqref{eq:RU02},  we obtain the WR robust ES value,  that is,
\begin{align}\label{eq:illustrateES}
	\ES_\alpha  ^{\rm WR}(\mathcal F)=\sup_{F\in \mathcal F} \ES_{\alpha}(F)
	&=\sup_{F\in \mathcal F}\min_{x\in\R}\left\{x+\frac{1}{1-\alpha}\pi_{F}(x)\right\}.
\end{align}
On the other hand, the $\preceq_2$-MA robust ES value can also be calculated using \eqref{eq:RU02} and \eqref{eq:computeV2} in Proposition \ref{th-sharp}, that is,
\begin{align}\label{eq:illustrateES2}
	\ES_\alpha ^{{\rm MA}_2} (\mathcal F)=\ES_\alpha \left(\bigvee_2 \mathcal F\right)
	&=\min_{x\in\R}\sup_{F\in \mathcal F} \left\{x+\frac{1}{1-\alpha}\pi_{F}(x)\right\},
\end{align}
where the second equality follows from  \eqref{eq:RU02} and  $\pi_{\bigvee_2\mathcal F} (x) = \sup_{F\in\mathcal F}   \pi_{ F} (x)$ by Proposition \ref{th-sharp} (ii).
The formulas \eqref{eq:illustrateES} and \eqref{eq:illustrateES2} imply that the WR and MA robust ES values can be seen as, respectively, the maximin and the minimax of the same bivariate objective function.
This observation immediately leads to
\begin{align}\label{eq:illustrateES3}  \ES_\alpha  ^{\rm WR}(\mathcal F) \le \ES_\alpha^{{\rm MA}_2}(\mathcal F),\mbox{ and equality holds if a minimax theorem holds}.
\end{align}
Therefore, although \eqref{eq-gap} is generally not an equality,   it may be an equality for $\ES_\alpha$  and $\preceq_2$ under certain conditions on $\mathcal F$.
In particular, as shown by \cite {ZF09}, if $\mathcal F$ is a convex polytope (see
Section \ref{sec:EMAM} for a definition) or a compact convex set of discrete cdfs, then \eqref{eq:illustrateES3} becomes an equality.
In the following theorem, we establish a more general sufficient condition to make \eqref{eq:illustrateES3} an equality, where $\ES_0=\mathfrak{m}$ and $\ES_\alpha$ for $\alpha \in (0,1)$ are treated separately.   We also collect the corresponding result for $\VaR_\alpha$ discussed above.
\tbl{Recall that $\rho^{\rm WR}$ and $\rho^{\rm MA}$ are defined in \eqref{eq-WC} and \eqref{eq-MA} for any risk measure $\rho$.}
\begin{theorem}\label{prop-cxES}
	Suppose that $\mathcal F\subseteq\mathcal M_1$.
	\begin{itemize}
		\item[(a)] If  $\sup_{F\in\mathcal F}\int_{\R}(x-y)_+\d F(y)
		\to 0 $ as $x\to -\infty$,  then $\mathfrak{m}^{\rm WR}(\mathcal F)=\mathfrak{m}^{{\rm MA}_2}(\mathcal F)$.
		
		\item[(b)] For $\alpha\in(0,1)$, if $\mathcal F $ is convex and $\preceq_2$-bounded,
		then $\ES_\alpha^{\rm WR}(\mathcal F)=\ES_\alpha^{{\rm MA}_2}(\mathcal F)$.\item[(c)] For $\alpha\in(0,1)$, if $\mathcal F $ is $\preceq_1$-bounded,
		then $\VaR_\alpha^{\rm WR}(\mathcal F)=\VaR_\alpha^{{\rm MA}_1}(\mathcal F)$.
	\end{itemize}
\end{theorem}

The most useful part of Theorem \ref{prop-cxES} is (b), which offers a simple condition under which the WR robust ES value can be obtained by implementing the MA approach. This result generalizes Theorems 1 and  2 of \cite{ZF09} where the set $\mathcal F$ is a convex polytope and a compact convex set of discrete cdfs, respectively.
Without convexity of $\mathcal F$, for $\alpha\in(0,1)$, $\ES_\alpha^{\rm WR}(\mathcal F) = \ES_\alpha ^{{\rm MA}_2}(\mathcal F)$ may not hold, as illustrated by the following example.

\begin{example} \label{ex-1es}
	Let $\alpha\in (0,1)$. Let  $\epsilon=(1-\alpha)/2 $, $F_1=\delta_0$ and $F_2=(1-\epsilon)\delta_{-1/(1-\epsilon)-1}+\epsilon\delta_{1/\epsilon}$, where $\delta_{t}$ represents the point-mass at $t\in\R$.
	By computing $\max \{\pi_{F_1},\pi_{F_2}\}$, we get $\bigvee_2\{F_1,F_2\}=(1-\epsilon)\delta_{-1/(1-\epsilon)}+\epsilon\delta_{1/\epsilon}$ and $$\ES_{\alpha}\left(\bigvee_2\{F_1,F_2\}\right) = \frac 12 \left(\frac{1}{ \epsilon}  -\frac{1}{1- \epsilon}\right) > \frac{1}{2}\left(\frac{1}{\epsilon} - \frac{2-\epsilon}{1- \epsilon}  \right)_+=\max\left\{ \ES_\alpha(F_1) ,\ES_\alpha(F_2)\right\}.$$
	Hence, $\ES_\alpha^{\rm WR}( \{F_1,F_2\} )  < \ES_\alpha^{{\rm MA}_2} (\{F_1,F_2\})$.
\end{example}

The conditions on $\mathcal F$ in (a) and (b) of Theorem \ref{prop-cxES} do not imply each other. The following example shows that $\mathfrak{m}^{\rm WR}(\mathcal F)=\mathfrak{m}^{{\rm MA}_2}(\mathcal F)$ may not hold in case $\mathcal F$ does not satisfy the condition in (a) and satisfy the condition in (b).
\begin{example} For $n\in\N$, let  $F_n=(1/n)\delta_{-n}+(1-1/n)\delta_{0}$, and denote by $\mathcal F$ the convex hull of $\{F_n\}_{n\in\N}$. By computing $\max \{\pi_{F_n},n\in\N\}$, we have $\bigvee_2\mathcal F=\bigvee_2\{F_n\}_{n\in\N}=\delta_0$. Note that  $\mathfrak{m}(F)=-1$ for any $F\in\mathcal F$. Hence, $\mathfrak{m}(\bigvee_2\mathcal F)=0>-1=\sup_{F\in\mathcal F}\mathfrak{m}(F)$, that is, $\mathfrak{m}^{\rm WR}( \mathcal F)  <  \mathfrak{m} ^{{\rm MA}_2}(\mathcal F)$.
\end{example}

\subsection{Convexity and other properties} \label{sec:5-3}

We first introduce some standard properties of a risk measure $\rho:\mathcal M_1\to\R$ and its associated $\widetilde \rho:L^1\to \R$.  
{\it Translation invariance}:
$\widetilde \rho( X+c)=\widetilde \rho(X)+c$ for any $c\in\R$ and $X\in L^1$.
{\it Positive homogeneity}:
$\widetilde \rho({\lambda X} )=\lambda\widetilde \rho(X)$ for any $\lambda>0$ and $X\in L^1$.
{\it Convexity}:
$\widetilde \rho({\lambda X+(1-\lambda)Y}) \le \lambda \widetilde \rho(X) + (1-\lambda)\widetilde \rho(Y)$ for any $\lambda \in [0,1]$ and $X,Y\in L^1$.
{\it Lower semicontinuity}: $\liminf_{n\to\infty} \widetilde \rho(X_{n})\ge \widetilde \rho(X)$ if $X_n,X\in L^1$ for all $n$ and 
$X_n\stackrel{\rm d}\to X$ as $n\to\infty$, where $\stackrel{\rm d}\to$ denotes convergence in distribution.\footnote{Convergence in distribution corresponds to weak convergence on $\M_1$. Note that this lower semicontinuity is different from $L^1$-lower semicontinuity commonly used in the literature of risk measures (e.g., \cite{FS16}).}
\emph{Comonotonic additivity}: $\widetilde \rho({X+Y})=\widetilde \rho(X)+\widetilde \rho(Y)$ for any $X,Y\in L^1$ that are comonotonic.\footnote{\color{black}Random variables $X$ and $Y$ are \emph{comonotonic} if there exists $\Omega_0\in\mathcal F$ with $\p(\Omega_0)=1$ such that for all $\omega,\omega'\in\Omega_0$,
	\begin{align*}
		(X(\omega)-X(\omega'))(Y(\omega)-Y (\omega'))\ge 0.
	\end{align*}
}
A risk measure  is \emph{coherent}, as defined by \cite{ADEH99}, if it satisfies translation invariance, positive homogeneity, and convexity \tbl{(also monotonicity, which is assumed in Definition \ref{def-risk}).} 
All the properties are defined for both $\rho$ and $\widetilde \rho$.

It is well known that $\VaR_\alpha$ and $\ES_\alpha$, $\alpha\in(0,1)$ satisfy translation invariance, positive homogeneity, lower semicontinuity, comonotonic additivity, and
$\ES_\alpha$ further satisfies convexity.
Translation invariance,  positive homogeneity  and convexity are standard  properties with interpretations extensively discussed by \cite{ADEH99} and \cite{FS16}.
Lower semicontinuity,  called the prudence axiom by \cite{WZ20}, means that if the loss cdf is modeled using a
truthful approximation, then the approximated risk model should not underreport the capital requirement. Comonotonic additivity is popular in both the literature of decision theory (e.g., \cite{Y87} and \cite{S89}) and that of risk measures (e.g., \cite{K01}).
A coherent risk measure on $\M_1$, including ES, is automatically consistent with both  $\preceq_{1}$ and $\preceq_{2}$; see e.g., \cite{L05},
	\cite{FS16} and \cite{SDR21}.

Next we formulate uncertainty for $\widetilde \rho$ with different probabilities based on $\rho:\M_1\to \R$.
Let $\mathcal P$ be the set of all  probability measures on $(\Omega,\mathcal B)$ absolutely continuous with respect to $\p$.
Define $\widetilde \rho^Q(X)=\rho(F_X^Q)$ where $F_X^Q$ is denoted by the cdf of $X$ under $Q \in \mathcal P$, assumed to be in $\mathcal M_1$. In particular, we have $\widetilde\rho^{\p}(X)=\widetilde\rho(X)=\rho(F_X)$ for all $X\in L^1$.
For  $\mathcal Q\subseteq \mathcal P$, we use $\mathcal F_{X|\mathcal Q}$ to represent the set of all possible cdfs of $X$ under the uncertainty set, i.e., $\mathcal F_{X|\mathcal Q}=\{F_X^Q: Q\in\mathcal Q\}$. 

In this setting, we study the properties of risk measures via WR approach 
\begin{align*}
	{\rm WR:} ~~~~~\widetilde\rho^{\rm WR}(X) =\rho^{\rm WR}\left(\mathcal F_{X|\mathcal Q}\right)=\sup_{F\in\mathcal F_{X|\mathcal Q}}\rho(F)
	=\sup_{Q\in\mathcal Q}\widetilde \rho^Q(X),
\end{align*}
and via ${\rm MA}_i$ approach for $i=1,2$,
\begin{align*}
	{\rm MA:} ~~~~~\widetilde\rho^{{\rm MA}_i}(X) =\rho^{{\rm MA}_i}\left(\mathcal F_{X|\mathcal Q}\right)=  {\color{black}\left\{\begin{array}{ll}\rho\left(\bigvee_i{\mathcal F_{X|\mathcal Q}}\right),~ &\textrm{if } \mathcal F_{X|\mathcal Q}  \textrm{ is $\preceq_i$-bounded;}\\
			\infty,&\textrm{otherwise}. \end{array}\right.}
\end{align*}

The two mappings  $\widetilde\rho^{{\rm WR}}$  and  $\widetilde\rho^{{\rm MA}_i}$ are well defined on the set $L=\bigcap_{Q\in \mathcal Q}L^1(\Omega,\mathcal B,Q)$.
This set always contains, for instance, all bounded random variables in $L^1 $. 
The next result gives properties of $\widetilde\rho^{{\rm WR}}$  and  $\widetilde\rho^{{\rm MA}_i}$ based on those of $\widetilde \rho$. 

\begin{theorem}\label{th-cxca}
	Let $\widetilde \rho:L^1\to\R$ be a risk measure, $\mathcal Q\subseteq \mathcal P$,
	and $L=\bigcap_{Q\in \mathcal Q}L^1(\Omega,\mathcal B,Q)$.
	The following statements hold.
	\begin{itemize}
		\item[(a)] If $\widetilde \rho$ is convex, then $\widetilde\rho^{\rm WR}$  and $\widetilde\rho^{{\rm MA}_2}$ are convex on $L$.
		\item[(b)] 
		If $\widetilde \rho$ satisfies
		comonotonic additivity, then $ \widetilde\rho^{{\rm MA}_1}$ satisfies comonotonic additivity on $L$.
		\item[(c)] If $\widetilde \rho$ satisfies translation invariance (positive homogeneity), then $\widetilde\rho^{\rm WR}$, $\widetilde\rho^{{\rm MA}_1}$  and $\widetilde\rho^{{\rm MA}_2}$ satisfy translation invariance (positive homogeneity) on $L$.
	\end{itemize}
\end{theorem}

A direct result from Theorem \ref{th-cxca} is that $\widetilde\rho^{\rm WR}$ and $\widetilde\rho^{{\rm MA}_2}$ are coherent whenever $\widetilde \rho$ is a coherent risk measure.
Note that $\widetilde\rho^{\rm WR}$ may not be comonotonic additive even if $\widetilde \rho$ is. For instance, the mapping $\widetilde\rho^{\rm WR}:X\mapsto \max\{\E^{\p_1}[X], \E^{\p_2}[X]\}$ for $\p_1\ne \p_2$ is not comonotonic additive in general.

\begin{remark}
	Although all the risk measures considered in the paper are law-invariant with respect to $\p$, $\widetilde\rho^{\rm WR}$, $\widetilde\rho^{{\rm MA}_1}$ and $\widetilde\rho^{{\rm MA}_2}$ may not be law-invariant with respect to any one probability measure.
	Nevertheless, $\widetilde\rho^{\rm WR}$, $\widetilde\rho^{{\rm MA}_1}$ and $\widetilde\rho^{{\rm MA}_2}$ are all law-invariant with respect to the probability set $\mathcal Q$ according to the definitions of \cite{DKW19} and \cite{WZ18}.
\end{remark}

\section{MA approach in {\color{black}distributionally} robust optimization}\label{sec:34}

In this section, we consider optimization problems
in which uncertainty is addressed by the WR   and ${\rm MA}_2$ approaches, i.e., $\preceq_2$ is chosen as the partial order.
Let $A$ be a set of possible actions, $f:A\times \R^d  \to \R$ be a loss function, $\mathcal F$ be a set of cdfs on $\R^d$, and $\mathcal F_{\mathbf a,f}$ is defined by \eqref{eq-mulUS}.
We consider the following two optimization problems 
\begin{equation}
	\label{eq:robustopt}
	\min_{\mathbf a\in A}~~ \rho^{\rm WR} (\mathcal F_{\mathbf a,f} ) \mbox{~~~~~and~~~~~}\min_{\mathbf a\in A} ~~\rho^{{\rm MA}_2} (\mathcal F_{\mathbf a,f}).
\end{equation}
Recall the equivalency between the WR  optimization and \eqref{prob:minmax}, i.e., $$
\min_{\mathbf a\in A}\rho^{\rm WR} (\mathcal F_{\mathbf a,f})=\min_{\mathbf a\in A}\sup_{F\in\mathcal F}\widetilde{\rho}^F(f(\mathbf a,\mathbf X)).
$$ 
It is straightforward to see that the WR optimization is a convex problem if $A$, $\widetilde  \rho$  and $f$ (in its first argument) are convex \tbl{because $\widetilde \rho$ is monotone and the inner problem is the supremum of a set of convex functionals}. We demonstrate in the following result that the ${\rm MA}_2$  optimization is convex under the same conditions.

\begin{proposition}\label{prop:cxMA}
	Let $A$ be a convex set.  
	Suppose that a risk measure $\rho:\mathcal M_1\to \R$ is convex, and $f: A\times \R^d\to \R$ is convex in its first argument. Then, the mapping $\mathbf a\mapsto \rho^{\rm MA_2}(\mathcal F_{\mathbf a,f})$ is convex. As a consequence, the ${\rm MA}_2$  optimization problem in \eqref{eq:robustopt} is a convex one.
\end{proposition}

We note that the convexity of both WR and ${\rm MA}_2$  optimization problems may not always coincide; see Section \ref{app:ex}.

Next, we detail the application of the ${\rm MA}_2$ approach to a broad class of coherent risk measures.
We introduce an extra assumption before diving in.

\begin{assumption}\label{assump:Babove}
	For any $\mathbf a\in A$, the uncertainty set $\mathcal F_{\mathbf a,f}$ is $\preceq_2$-bounded.
\end{assumption}

Assumption \ref{assump:Babove} guarantees that $\bigvee_2\mathcal F_{\mathbf a,f}$ exists for any $\mathbf a\in A$.  
This assumption is mild.  
For example,  it holds if we consider a common compact support of all involved cdfs.


A law-invariant coherent risk measure has a Kusuoka representation (\cite{K01} and \cite{S13}), which can be approximated by the following  form\footnote{ 
	While our primary focus is on law-invariant coherent risk measures, we note that
	there is no extra difficulty to deal with a law invariant convex risk measure which has the form $\rho (F) = \sup_{w\in \mathbb W}\{\sum_{j=1}^{n^w} p_j^w {\rm ES}_{\alpha_j^w} (F)-\beta(w)\}$ where $\beta:\mathbb W\to\R$ is a penalty function.}
\begin{align} 
	\label{eq:0818-3}
	\rho (F) = \sup_{w\in \mathbb W}~\sum_{j=1}^{n^w} p_j^w {\rm ES}_{\alpha_j^w} (F),
\end{align}
where $\mathbb W$ is a finite index set, and  for each $w\in \mathbb W$, $n^w\in \N$,  $(p_1^w,\dots,p_{n^w}^w) \in\Delta_{n^w}$ and $(\alpha_1^w,\dots,\alpha_{n^w}^w)\in [0,1]^{n^w}$.  
{We formally show in Proposition \ref{pro:A4} in Section \ref{app:optMA-1-1} that \eqref{eq:0818-3} can approximate any law-invariant coherent risk measure as the size of $\mathbb W$ goes infinity.} 
For the risk measure in \eqref{eq:0818-3}, we obtain a reformulation of 
the ${\rm MA}_2$  optimization.


\begin{theorem}\label{th-MAcr}
	Suppose that Assumption \ref{assump:Babove} holds, and
	let $\rho$ be given by \eqref{eq:0818-3}. The ${\rm MA}_2$  optimization problem in \eqref{eq:robustopt} can be reformulated as the following problem 
	\begin{align}\label{prob-MAcr}
		\min_{\mathbf a\in A,\,x_j^{w}, h,h_{j}^{w}\in\R} ~ &  h\\
		{\rm s.t.} ~~ &  \sum_{j=1}^{n^w} p_j^w h_j^w \le h,~~ w\in\mathbb{W}\notag\\
		&   x_j^{w}+ \frac1{1-\alpha_j^w} \sup_{F\in\mathcal F}\E^F[( f({\bf a},\mathbf X)-x_j^{w})_+]   \le h_j^w,~~ j\in[n^w],~w\in\mathbb{W}.\notag
	\end{align}
\end{theorem}

Consider a finite uncertainty set $\mathcal F=\{F_1,\dots,F_n\}$, where $n\in\N$. Under this setting, we reformulate
the ${\rm MA}_2$  optimization, which is a direct consequence of Theorem \ref{th-MAcr}. 
The two optimization problems in \eqref{eq:robustopt} can be respectively reformed as
\begin{align} \label{eq:0820-1}
	\mbox{$\mathrm{WR}$}: \min_{\mathbf a\in A,\,x_{i,j}^{w}, h,h_{i,j}^{w}\in\R}~&  h \\
	{\rm s.t.}~~&  \sum_{j=1}^{n^w} p_j^w h_{i,j}^{w} \le h,~~ i\in[n],~ w\in\mathbb{W} \nonumber\\
	& x_{i,j}^{w}+\frac{1}{1-\alpha_j^w}\E^{F_i}[(f(\mathbf a,\mathbf X)-x_{i,j}^{w})_+] \le h_{i,j}^{w},~~  i\in[n],~j\in[n^w],~w\in\mathbb{W}; \nonumber
\end{align} 
\begin{align}\label{eq:0820-2}
	\mbox{$\mathrm{MA}_2$}:  \min_{\mathbf a\in A,\,x_j^{w}, h,h_{j}^{w}\in\R}~&  h \\
	{\rm s.t.}~~&    \sum_{j=1}^{n^w} p_j^w h_j^w \le h,~~ w\in\mathbb{W} \nonumber\\
	&x_j^{w}+\frac{1}{1-\alpha_j^w}\E^{F_i}[(f(\mathbf a,\mathbf X)-x_j^{w})_+] \le h_j^w,~~  i\in[n],~j\in[n^w],~ w\in\mathbb{W}. \nonumber
\end{align} 
If $\rho=\ES_\alpha$ for $\alpha\in(0,1)$, then the two optimization problems in \eqref{eq:robustopt} can be reformed as
\begin{align*}
	&\mbox{$\mathrm{WR}$}: \min_{\mathbf a\in A,\,x_i,h\in\R}~h~~~~  
	{\rm s.t.}~~ x_i+\frac{1}{1-\alpha}\E^{F_i}[(f(\mathbf a,\mathbf X)-x_i)_+] \le h,~ i\in[n]; \\
	&\mbox{$\mathrm{MA}_2$}: \min_{\mathbf a\in A,\,x,   h \in\R}~h~~~~  
	{\rm s.t.}~~  x+ \frac{1}{1-\alpha}\E^{F_i}[(f(\mathbf a,\mathbf X)-x )_+] \le h ,~  i\in[n].
\end{align*}

\tbl{A comparison between the computation time of   \eqref{eq:0820-1} and \eqref{eq:0820-2} based on a specific example is given in Section \ref{app:ex}.} 
\begin{remark}
	By Theorem \ref{th-MAcr} (or using Proposition \ref{prop-EMAM}), the ${\rm MA}_2$   optimization retains the form \eqref{eq:0820-2}   if we substitute the uncertainty set $\mathcal F=\{F_1,\dots,F_n\}$ with any set $\widetilde{\mathcal F}$, whose convex hull has extreme points  $F_1,\ldots,F_n$. 
\end{remark}

For other settings of  Problem \eqref{prob-MAcr}  that are solvable by convex program, see Section \ref{app:ex}.
The two optimization problems in \eqref{eq:robustopt} are equivalent under a convex uncertainty set $\mathcal F$ when $\rho={\rm ES}_\alpha$ for all $\alpha\in(0,1)$. This equivalence stems from the fact that $\mathcal F_{\mathbf a,f}$ is convex for all $\mathbf a\in A$ which combining with Theorem \ref{prop-cxES} imply
$\ES_{\alpha}^{\rm WR}(\mathcal F_{\mathbf a,f})=\ES_{\alpha}^{\rm MA_2}(\mathcal F_{\mathbf a,f})$ for all $\mathbf a\in A$. Many results in the literature rely on converting between $\ES_{\alpha}^{\rm WR}(\mathcal F_{\mathbf a,f})$ and $\ES_{\alpha}^{\rm MA_2}(\mathcal F_{\mathbf a,f})$;
see e.g., \cite{ZF09}, \cite{N10} and \cite{C11}.

\section{Wasserstein   and  mean-variance uncertainty sets}\label{sec:US}
In this section, we focus on three specific and popular uncertainty sets:
				(a) univariate Wasserstein uncertainty, (b) multivariate Wasserstein uncertainty, and (c) mean-variance uncertainty.
				We obtain explicit formulas for
				the  robust models as well as the WR and MA robust risk evaluation. 
				Furthermore, the  portfolio selection problem will be explored based on these two robust approaches.
				
				For results in this section and Section \ref{sec:Numerical},  we define a few classes of risk measures other than $\VaR$ and $\ES$.
				The Range Value-at-Risk (RVaR), proposed by \cite{C10}, is defined as
				$$
				{\rm RVaR}_{\alpha,\beta}(F)=\frac{1}{\beta-\alpha}\int_\alpha^\beta \VaR_s(F)\d s,~~0\le\alpha<\beta\le 1.
				$$
				Special and limiting cases of ${\rm RVaR}_{\alpha,\beta}$ include $\ES_\alpha$ with $\beta=1$  and $\VaR_\beta$  with $\alpha\uparrow\beta\in (0,1) $.   If $\beta<1$, then ${\rm RVaR}_{\alpha,\beta}$ is not $\preceq_2$-consistent by e.g., \citet[Theorem 3]{WWW20}. 
				The power-distorted (PD) risk measure (\cite{W95, CM09}) is defined as
				$$
				{\rm PD}_k(F)=\int_0^1 ks^{k-1}\VaR_s(F)\d s,~~k\ge 1.
				$$
				The PD risk measure is coherent.
				The expectile,  proposed by \cite{NP87} and denoted by $\ex_\alpha$, is defined as the unique solution $t = \ex_\alpha(F_X) \in \R$ to the following equation,
				$$
				\alpha\E[(X-t)_+]=(1-\alpha)\E[(X-t)_-],~~X\in L^1.
				$$
				The risk measure $\ex_\alpha$ is   coherent  (and thus $\preceq_2$-consistent) if and only if   $\alpha\in [1/2,1)$; we will use this specification.

				\subsection{Uncertainty induced by the univariate Wasserstein metric}\label{sec:62}
				We first focus on an uncertainty set induced by the Wasserstein metric.
				Let $\M_p$ be the set of cdfs on $\R$ with finite  $p$th moment and $F_0\in \M_p$ be a pre-specified cdf used as benchmark.
				For $p\ge 1$, the {\color{black}$p$-Wasserstein} metric between $F$ and $F_0$ is defined as
				\begin{equation}\label{eq:defW}{W_p}(F,F_0)=\left(\int_0^1 |F^{-1}(s)-F_0^{-1}(s)|^p\d s\right)^{1/p}.\end{equation}
				The corresponding uncertainty set is, for a parameter $\epsilon\ge0$,  \begin{align}\label{eq-WU}
					\mathcal F_{p,\epsilon}(F_0) =\left\{F\in \M_p: W_p(F,F_0) \le\epsilon\right\},
				\end{align}
				which is a convex set.
				The parameter $\epsilon$ represents the magnitude of uncertainty.
				Denote by $$F_{p,\epsilon|F_0}^1=\bigvee_1 \mathcal F_{p,\epsilon}(F_0)~~~{\rm and }~~~ F_{p,\epsilon|F_0}^2=\bigvee_2 \mathcal F_{p,\epsilon}(F_0) $$ the supremum of $\mathcal F_{p,\epsilon}(F_0)$ with respect to $\preceq_1$ and $\preceq_2$, respectively.
				In the following result,
				we will identify an explicit form of the suprema $F_{p,\epsilon|F_0}^1$ and $F_{p,\epsilon|F_0}^2$
				in terms of left quantile functions.
				
				
				
				\begin{theorem}\label{prop-WU}
					Suppose that $\epsilon> 0$, $p\ge 1$ and $ F_0 \in \M_p$.
					\begin{itemize}
						\item[(a)] The left quantile function of $F_{p,\epsilon|F_0}^1$ is   uniquely determined by 
						\begin{align}\label{eq-WU1}
							\int_\alpha^1 \left((F_{p,\epsilon|F_0}^1)^{-1}(\alpha)-F_0^{-1}(s)\right)_+^p\d s =\epsilon^p,~~~\alpha\in(0,1).
						\end{align}
						\item[(b)] The set $\mathcal F_{1,\epsilon}(F_0)$ is not $\preceq_2$-bounded. For $p>1$, the left quantile function of $F_{p,\epsilon|F_0}^2$ is given by
						\begin{align}\label{eq-WU2}
							(F_{p,\epsilon|F_0}^2)^{-1}(\alpha)=F_0^{-1}(\alpha)
							+\left(1-\frac1p\right)(1-\alpha)^{-1/p}\epsilon,~~~\alpha\in(0,1).
						\end{align}
					\end{itemize}
				\end{theorem}

	Since the cdfs $F_{p,\epsilon|F_0}^1$ and $F_{p,\epsilon|F_0}^2$, as well as their quantile functions, are obtained explicitly in  Theorem \ref{prop-WU}, the robust risk values
	$\rho^{{\rm MA}_1} (\mathcal F_{p,\epsilon}(F_0)) $ and  $\rho^{{\rm MA}_2} (\mathcal F_{p,\epsilon}(F_0)) $ can be computed in a straightforward manner. On the other hand, $\rho^{\rm WR}(\mathcal F_{p,\epsilon}(F_0))$ is often difficult to compute if the risk measure is complicated,
	although there are some results in the literature {\color{black} that} considered the WR approach for special choices of risk measures. {\color{black}\cite{PHM16} presented  combinations of risk measures and uncertainty sets that allow for computationally tractable reformulations.} 
				
				As a feature of the robust model,  both $F_{p,\epsilon|F_0}^1$  and $F_{p,\epsilon|F_0}^2$ are  heavy-tailed  even if the benchmark distribution $F_0$ is light-tailed.
				Heavy-tailed distributions are common for modeling financial data; see e.g., \cite{MFE15}.
				Indeed, $(F_{p,\epsilon|F_0}^1)^{-1}\ge (F_{p,\epsilon|F_0}^2)^{-1}$, and $(F_{p,\epsilon|F_0}^2)^{-1}$ is the sum of the quantile $F_0^{-1}$ and a Pareto quantile with tail index $p>1$.
				Some other observations on the supremum distributions in Theorem \ref{prop-WU} are made
				in Remark \ref{rem:omit-WD}.
				
				Noting that the Wasserstein uncertainty set  $\mathcal F_{p,\epsilon}(F_0)$ is
				convex, we have
				$\ES_\alpha^{\rm WR}(\mathcal F_{p,\epsilon}(F_0))
				=\ES_\alpha^{{\rm MA}_2}(\mathcal F_{p,\epsilon}(F_0))$ by Theorem \ref{prop-cxES}. A simulation result in case of $p=2$, $\epsilon=0.1$ and  a  standard normal benchmark distribution  is reported in Section \ref{app:G}.

				

				\subsection{Multivariate Wasserstein uncertainty}\label{sec:63}
				For $p\ge 1$ and $a\ge 1$, let $ \M_p(\R^d)$ be the set of all cdfs on $\R^d$ with finite $p$th moment in each component.
				The {\color{black}$p$-Wasserstein} metric on $\R^d$ between $F,G\in \mathcal M_p(\R^d)$ is defined as
				$${W^d_{a,p}}(F,G)= \inf_{F_{\mathbf X}=F, F_{\mathbf Y}=G}\left( \E[\Vert \mathbf X-\mathbf Y\Vert ^p_a]\right)^{1/p},$$
				where $\Vert \cdot \Vert_a$ is the $L^a$ norm on $\R^d$; see e.g., \cite{BCZ22}.
				If $d=1$, then $W_{a,p}^d$ is $W_p$ in \eqref{eq:defW} where the infimum   is attained by comonotonicity via the Fr\'echet-Hoeffding inequality.
				Define the   Wasserstein uncertainty set for a benchmark distribution $F_0\in \mathcal M_p(\R^d)$ as, similar to \eqref{eq-WU},
				\begin{align}\label{eq-WU-multi}
					\mathcal F_{a,p,\epsilon}^d(F_0) =\left\{F\in \M_p(\R^d): W^d_{a,p}(F,F_0) \le\epsilon\right\}, ~~\epsilon\ge0.
				\end{align}

				We focus on a portfolio selection problem, {\color{black} i.e., the loss function is chosen as  the linear function $f(\mathbf w,\mathbf x)=\mathbf w^\top \mathbf x$.} The  
				portfolio risk is $\rho(F_{\mathbf w^\top \mathbf X})$ for some weight vector $\mathbf w\in \R^d$ and risk   vector $\mathbf X$ with unknown cdf in the multi-dimensional Wasserstein ball $\mathcal F_{a,p,\epsilon}^d(F_0)$.
				The univariate uncertainty set for the cdf of $ \mathbf w^\top \mathbf X$ is denoted
				by
				\begin{align}\label{eq-WU-multi2}
					\mathcal F_{\mathbf w, a,p,\epsilon} (F_0) =\left\{F_{\mathbf w^\top \mathbf X}:  F_{\mathbf X}\in
				    \mathcal F_{a,p,\epsilon}^d(F_0) \right\},~~~F_0\in \mathcal M_p(\R^d).
				\end{align}
In the following theorem, we show that the problem of multivariate Wasserstein uncertainty can be conveniently converted to a univariate setting.
				\begin{theorem}\label{th:7}
					For $\epsilon\ge 0$, $a\ge 1$ and $p \ge 1$,   random vector $\mathbf X$ with $F_{\mathbf X}\in \M_p(\R^d)$ and $\mathbf w\in  \R^d$, we have
{\color{black}$$\mathcal F _{\mathbf w, a,p,\epsilon }(F_{\mathbf X}) = \mathcal F_{p,\epsilon\Vert\mathbf w \Vert_b}(F_{\mathbf w^\top \mathbf X}) ,$$}
 where $b$ satisfies $1/a+1/b=1$.
					As a consequence,  for any  $\rho:\M_1\to \R$ and $i\in\{1,2\}$,
					$$\rho^{\rm WR}  (\mathcal F _{\mathbf w,a,p,\epsilon }(F_{\mathbf X}))
					= \rho^{\rm WR} (\mathcal F_{p,\epsilon\Vert \mathbf w \Vert_b    }(F_{\mathbf w^\top \mathbf X}))
					\mbox{~~~and~~~} \rho^{{\rm MA}_i}  (\mathcal F _{\mathbf w, a,p,\epsilon }(F_{\mathbf X}))
					= \rho^{{\rm MA}_i} (\mathcal F_{p,\epsilon \Vert \mathbf w \Vert_b   }(F_{\mathbf w^\top \mathbf X})).
					$$
				\end{theorem}

				Intuitively, by
				Theorem \ref{th:7},
				the multi-dimensional Wasserstein ball has the simple property of a usual Euclidean ball, that its affine projection is a lower-dimensional ball (this intuitive observation is not completely trivial because of the infimum   in the Wasserstein metric).
				This result allows us
				to solve the MA robust portfolio optimization by applying Theorem \ref{prop-WU}. 

{\color{black}				
We illustrate Theorem \ref{th:7}  with the   setting of an elliptical benchmark distribution. 
An elliptical distribution with characteristic generator $\psi$ is denoted by ${\rm E}(\bm\mu,\Sigma,\psi)$, which has normal and t-distributions as special cases; see \citet[Chapter 6]{MFE15}  for a precise definition.
Let the benchmark distribution $F_0 = {\rm E}(\bm\mu,\Sigma,\psi)$
and denote by $F_{\psi} ={\rm E}(0,1,\psi)$. 
Define a Pareto distribution $G_{p}$ with $G_{p}^{-1}(\alpha)=(1-\alpha)^{-1/p}$ for $\alpha\in(0,1)$.
By Theorems \ref{prop-WU} and \ref{th:7}, it holds that
\begin{align}\label{prob-MAellipsoid}
	\min_{{\mathbf w}\in\mathcal W}:~~ \rho^{{\rm MA}_2}\left( \mathcal F _{\mathbf w,a,p,\epsilon }(F_{0})\right)
	=\widetilde\rho\left(\mathbf w^\top\bm\mu + \sqrt{\mathbf w^\top\Sigma \mathbf w} F_{\psi}^{-1}(U)
	+\left(1-\frac{1}{p}\right)\epsilon\|\mathbf w\|_b G_p^{-1}(U)\right),
\end{align}
where $U$ is a uniform random variable on $[0,1]$. 
The WR approach does not admit an explicit formula like \eqref{prob-MAellipsoid}, 
unless $\rho$ is a coherent distortion risk measure; see \cite{W14} and \cite{PHM16}.}

				Consider a coherent distortion risk measure defined by $\rho_h(F)=\int_0^1\VaR_s(F)\d{h(s)}$, where $h:[0,1]\to[0,1]$ is increasing and convex with $h(0)=1-h(1)=0$. In this case, by applying Proposition 4 of \cite{LMWW21} and Theorems \ref{prop-WU} and \ref{th:7}, we obtain the following reformulations \begin{align}
					&\min_{{\mathbf w}\in\mathcal W}: ~~\rho_h^{\rm WR} \left( \mathcal F _{\mathbf w,p,\epsilon }(F_{0})\right)= {\mathbf w}^\top\bm\mu+\rho_h(F_{\psi})\sqrt{{\mathbf w}^\top\Sigma {\mathbf w}}+\zeta(p,h)\epsilon\Vert \mathbf w\Vert_b, \label{eq:WassWR}
				\end{align} 
				(this is also obtained by \cite{W14}) and
				\begin{align}
					\min_{{\mathbf w}\in\mathcal W}:~~ \rho_h^{\rm MA_2}\left( \mathcal F _{\mathbf w,p,\epsilon }(F_{0})\right)= {\mathbf w}^\top\bm\mu+\rho_h(F_{\psi})\sqrt{{\mathbf w}^\top\Sigma {\mathbf w}}+\xi(p,h)\epsilon\Vert \mathbf w\Vert_b, \label{eq:WassMA}
				\end{align}
				where 
				$$
				\zeta(p,h)=\left(\int_0^1(h'_+(s))^{p/(p-1)}\d{s}\right)^{(p-1)/p}~~{\rm and}~~
				\xi(p,h)=\frac{p-1}{p}\int_0^1 (1-s)^{-1/p}\d{h(s)}.
				$$
				In particular, \eqref{eq:WassWR} and \eqref{eq:WassMA} are second-order conic program (SOCP) when $a=2$; see e.g., \cite{BN01}.
				Coherence of $\rho$ (convexity of $h$) is essential for the WR formula  in \eqref{eq:WassWR} because general formulas are not available for non-convex distortions under Wasserstein uncertainty. In contrast, the MA formula \eqref{eq:WassMA} holds for
				any distortion risk measures (even if they may  not be $\preceq_2$-consistent) which   directly follows from Theorems \ref{prop-WU} and   \ref{th:7}.
				Numerical and empirical results on the above   approaches for robust portfolio selection are presented in Section \ref{sec:72}.

\subsection{Uncertainty induced by mean-variance information}\label{sec:MV}
				Next, we pay attention to an uncertainty set defined by the first two moments, that is, for some $\mu\in\R $ and $\sigma> 0$, the set
				\begin{align}\label{eq-MVU}
					\mathcal F_{\mu,\sigma}:=\left\{F\in \M_2: \mathfrak{m}(F)=\mu~{\rm and}~{\rm var}(F)=\sigma^2\right\},
				\end{align}
				where
				$\mathfrak{m}(F)$ and ${\rm var}(F)$ represent the mean and the variance of $F$, respectively.
				The two equalities in \eqref{eq-MVU} can be safely replaced by inequalities $\mathfrak{m}(F)\le \mu$ and ${\rm var}(F)\le \sigma^2$ in the problems we consider, and we omit the  formulation with inequalities.
				The WR robust risk value for different risk measures based on this uncertainty set $\mathcal F_{\mu,\sigma}$ has been extensively studied in literature, see
				e.g., \cite{G03}, 
				\cite{ZF09}, \cite{N10}, \cite{C11}, 
				\cite{L18} and the references therein. 
				
				
				For the MA approach, we will identify the supremum of  $\mathcal F_{\mu,\sigma}$ with respect to $\preceq_1$ and $\preceq_2$.
				Theorem 1 of \cite{G03} and Corollary 1.1 of \cite{J77} (see also \citet[Theorem 1.10.7]{MS02}) yield
				$$
				\sup_{F\in\mathcal F_{\mu,\sigma}}\VaR_{\alpha}(F)=\mu+\sigma\sqrt{\frac{\alpha}{1-\alpha}},
				~~~\alpha\in(0,1)
				$$
				and
				$$
				\sup_{F\in\mathcal F_{\mu,\sigma}}\pi_F(x)=\frac12\left(\mu-x+\sqrt{(x-\mu)^2+\sigma^2}\right),
				~~x\in\R.
				$$
				Denote by $F_{\mu,\sigma}^1=\bigvee_1 \mathcal F_{\mu,\sigma}$ and $F_{\mu,\sigma}^2=\bigvee_2 \mathcal F_{\mu,\sigma}$ the supremum of $\mathcal F_{\mu,\sigma}$ with respect to $\preceq_1$ and $\preceq_2$, respectively.
				Using Proposition \ref{th-sharp} and above two equations, we immediately get the explicit expressions of  $F_{\mu,\sigma}^1$ and  $F_{\mu,\sigma}^2$.
				
				\begin{proposition}\label{prop-MV} 
					Suppose that $\mu\in\R$ and $\sigma>0$. We have
					\begin{align}\label{eq-F1}
						& F_{\mu,\sigma}^1(x)=\frac{(x-\mu)^2}{\sigma^2+(x-\mu)^2},~~x\ge \mu;
				\\ \label{eq-F2}
					&	F_{\mu,\sigma}^2(x)=\frac12\left(1+\frac{x-\mu}{\sqrt{(x-\mu)^2+\sigma^2}}\right),~~x\in\R.
					\end{align}
				\end{proposition}
				
				We note that both $F_{\mu,\sigma}^1$ and $F_{\mu,\sigma}^2$ are in $\M_1$, so they are ready for implementation with any risk measures or   preferences well-defined on $\M_1$; however, none of $F_{\mu,\sigma}^1$ and $F_{\mu,\sigma}^2$ is in $\M_2$. Most risk measures in practice, including ES and VaR and the other examples in this section, are well-defined and finite on $\M_1$.
				
				
				By Proposition \ref{prop-MV}, for a risk measure that is $\preceq_1$-consistent or $\preceq_2$-consistent,
				the MA robust risk value for the uncertainty set $\mathcal F_{\mu,\sigma}$ can be directly obtained by calculating the risk measure of $F_{\mu,\sigma}^1$ or $F_{\mu,\sigma}^2$.
				To compute the WR robust risk value,  for a coherent risk measure $\rho$,
				\cite{L18} gives the explicit expression of $\rho^{\rm WR}(\mathcal F_{\mu,\sigma})$
				based on the Kusuoka representation.
				In addition, noting that $\mathcal F_{\mu,\sigma}$ is a
				convex set,
				if $\rho$ is an ES,
			 then
				$
				\rho^{\rm WR}(\mathcal F_{\mu,\sigma})=\rho^{{\rm MA}_2}(\mathcal F_{\mu,\sigma})= \rho(F_{\mu,\sigma}^2).
				$  
				If $\rho$ is a 
				VaR, then
				$
				\rho^{\rm WR}(\mathcal F_{\mu,\sigma})=\rho^{{\rm MA}_1}(\mathcal F_{\mu,\sigma})= \rho(F_{\mu,\sigma}^1).
				$
				The explicit WR and MA robust risk values for $\ES_\alpha$, ${\rm RVaR}_{\alpha,\beta}$, the power-distorted risk measure and the expectile are given in Table \ref{tab-R},\footnote{To obtain these formulas, we use the following results. \cite{Li18} showed that ${\rm RVaR}_{\alpha,\beta}^{\rm WR} (\mathcal F_{\mu,\sigma})= {\rm ES}_{\alpha}^{\rm WR} (\mathcal F_{\mu,\sigma})$ for all $\beta\in(\alpha,1)$. The value of $\mathrm{PD}_k$ via the WR approach can be directly derived by \citet[Theorem 2]{L18}. An expectile can be represented as the supremum of convex combinations of ES and expectation; see \citet[Proposition 9]{B14}. By Theorem \ref{th-ESM} and noting that all elements in $\mathcal F_{\mu,\sigma}$ have the same expectation, we obtain $\ex^{\rm WR}_{\alpha}(\mathcal F_{\mu,\sigma})=\ex_\alpha^{{\rm MA}_2}(\mathcal F_{\mu,\sigma})$.}
				and a few  figures on  their numerical values are reported in Section \ref{app:G}. Since those risk measures satisfy translation invariance and positive homogeneity, it suffices to consider the case $(\mu,\sigma)=(0,1)$.
				
				
				\begin{table}
					\caption{WR and MA under uncertainty induced by $\mathcal F_{0,1}$}\label{tab-R}\def\arraystretch{1.5}
					\begin{center}
						\begin{tabular}{c|ccc}
							\hline
							$\rho$ & $\rho^{\rm WR}$ & $\rho^{{\rm MA}_1}$  & $\rho^{{\rm MA}_2}$\\[3pt]
							\hline
							$\ES_{\alpha}$  & $\sqrt{\frac{\alpha}{1-\alpha}}$ & $\frac{1}{1-\alpha}\int_\alpha^1 \sqrt{\frac{s}{1-s}}\d s $ & $\sqrt{\frac{\alpha}{1-\alpha}}$   \\[3pt]
							${\rm RVaR}_{\alpha,\beta}$  & $\sqrt{\frac{\alpha}{1-\alpha}}$ & $\frac{1}{\beta-\alpha}\int_\alpha^\beta \sqrt{\frac{s}{1-s}}\d s$  &  $\frac{1}{\beta-\alpha}\int_\alpha^\beta \frac{s-1/2}{\sqrt{s(1-s)}}\d s$ \\[6pt]
							${\rm  VaR}_{\alpha}$  & $\sqrt{\frac{\alpha}{1-\alpha}}$ & $\sqrt{\frac{\alpha}{1-\alpha}}$  & $\frac{\alpha-1/2}{\sqrt{\alpha(1-\alpha)}}$  \\[3pt]
							${\rm PD}_k$  & $\frac{k-1}{\sqrt{2k-1}}$ & $\frac{\sqrt{\pi}\Gamma\left(k+1/2\right)}{\Gamma(k)}$  & $ \frac{\sqrt\pi(k-1)}{2k-1}\frac{ \Gamma(k+1/2)}{\Gamma(k)}  $   \\[3pt]
							$\ex_\alpha$ & $\frac{\alpha-1/2}{\sqrt{\alpha(1-\alpha)}}$ & $\ex_\alpha(F_{0,1}^1)$   &   $\frac{\alpha-1/2}{\sqrt{\alpha(1-\alpha)}}$\\[3pt]
							\hline
						\end{tabular}
					\end{center}
					~\\[5pt]
					\footnotesize{Note: $\Gamma$ is the gamma function;
      $\ex_\alpha(\mathcal F_{0,1}^1)$ can be numerically computed but it  does not admit an explicit formula.}
				\end{table}

				Similar to Section \ref{sec:63}, we apply the MA approach with mean-variance uncertainty to  robust portfolio selection.
				The portfolio risk is $\rho(F_{\mathbf w^\top \mathbf X})$ for some portfolio weight vector $\mathbf w\in \R^d$ and risk  vector $\mathbf X$ with unknown cdf in the uncertainty set with given first two moments, which can be formulated as, for a feasible set  $\mathcal W$ of $\mathbf w$,
				\begin{align}\label{eq-RPO}
					\min_{\mathbf w\in\mathcal W}~~ \rho^{\rm MA}\left( \mathcal F_{\mathbf w,\bm\mu,\Sigma}\right), \mbox{~~~where~}
					\mathcal F_{\mathbf w,\bm\mu,\Sigma}=\{F_{\mathbf w^\top \mathbf X}: \E[\mathbf X]=\bm\mu,~{\rm Cov}(\mathbf X)=\Sigma\},
				\end{align}
				where $\E[\mathbf X]$ and ${\rm Cov}(\mathbf X)$ represents the mean vector and the covariance of $\mathbf X$.
				Applying the general projection property in \cite{P07} (see also \citet[Lemma 2.2]{C11}), the two sets
				$\mathcal F_{\mathbf w,\bm\mu,\Sigma}$ and $\mathcal F_{\mathbf w^\top\bm\mu,\sqrt{\mathbf w^\top\Sigma \mathbf w} }$ are identical. Hence,
				\eqref{eq-RPO} is equivalent to
				\begin{align*}
					\min_{\mathbf w\in\mathcal W}~~ \rho^{\rm MA}\left(\mathcal F_{\mathbf w^\top\bm\mu,\sqrt{\mathbf w^\top\Sigma \mathbf w} }\right).
				\end{align*}
				In case of ${\rm MA}_1$ or ${\rm MA}_2$,  this leads to
{\color{black}
\begin{align}\label{eq-RPO2}
				\min_{{\mathbf w}\in\mathcal W}:~~ \rho^{{\rm MA}_i}\left(\mathcal F_{\mathbf w^\top\bm\mu,\sqrt{\mathbf w^\top\Sigma \mathbf w} }\right)
				=
				\widetilde\rho^{F_{0,1}^i}\left(\mathbf w^\top\bm\mu+\sqrt{\mathbf w^\top\Sigma \mathbf w}X\right),
\end{align}
where $F_{0,1}^i$ is given by Proposition \ref{prop-MV} in explicit form for $i=1,2$.
}
In particular, if $\rho$ satisfies translation invariance and positive homogeneity, Problem \eqref{eq-RPO2}
leads to the following convenient formulation of SOCP, for $i=1,2$,
				$$
				\min_{{\mathbf w}\in\mathcal W}:~~ \rho^{{\rm MA}_i}\left(\mathcal F_{\mathbf w^\top\bm\mu,\sqrt{\mathbf w^\top\Sigma \mathbf w} }\right)
				=\left\{\mathbf w^\top\bm\mu+\sqrt{\mathbf w^\top\Sigma \mathbf w}~\rho\left(F_{0,1}^i\right)\right\}.
				$$

\section{Characterization of risk measures by equivalence in MA}\label{sec:charac}\label{sec:EMAM}

In this section, we aim to characterize the risk measures under which the WR and MA approaches are equivalent, that is, 
\begin{equation}\label{eq210929-1}
	\rho \left( \bigvee \mathcal F\right)=\sup_{F\in\mathcal F}\rho(F),
\end{equation}
for all $\mathcal F\in\mathcal S$, where
$\mathcal S $ is a collection of subsets of $\mathcal M$. 
We are interested in the case that $\mathcal S$ is the collection of all convex polytopes which is defined by \begin{align*}
	\mathcal F={\rm conv}(F_1,\ldots,F_n)
	=\left\{\sum_{i=1}^n\lambda_iF_i: \bm\lambda\in \Delta_n\right\},
\end{align*}
where $F_1,\dots,F_n$ are finitely many cdfs.
The corresponding property is called {\it equivalence in model aggregation for convex polytopes} (cEMA); that is,  \eqref{eq210929-1} holds for all convex polytopes $\mathcal F\subseteq \mathcal M$.
\begin{itemize}
	\item[]{\it $\preceq$-cEMA}: Let $(\mathcal M,\preceq)$ be an ordered set. A mapping $\rho:\mathcal M\to\R$ satisfies  $\preceq$-cEMA  if  $\rho \left( \bigvee \mathcal F\right)=\sup_{F\in\mathcal F}\rho(F)$ holds  for any {\color{black} nonempty} convex polytope $\mathcal F\subseteq\mathcal M$.
\end{itemize}

All results in this section remain valid if convex polytopes in the above definition are replaced by convex sets bounded from above, and such a property  is stronger than cEMA.\footnote{Recall that  characterization results are generally stronger if imposed properties are weaker, so we aim for a weaker formulation of the properties.}

Our main focus is the partial orders $\preceq_1$ and $\preceq_2$. By Proposition \ref{prop-EMAM}, 
$\preceq_i$-cEMA  is equivalent to
\begin{align}\label{eq-eqcEMA}
	\rho\left(\bigvee_i\{F_1,\dots,F_n\}\right)=\sup\left\{\rho\left(\sum_{i=1}^n\lambda_iF_i\right):\bm\lambda\in \Delta_n\right\},~~~i=1,2,
\end{align}
for all $F_1,\dots,F_n\in \mathcal M$.
By \eqref{eq-eqcEMA},
$\preceq_i$-cEMA is stronger than $\preceq_i$-consistency since for any $F\preceq_i G$, we have
$
\rho(G)=\rho\left(\bigvee_i \{F,G\}\right)=\sup_{\lambda\in[0,1]}\rho(\lambda F+(1-\lambda)G)\ge\rho(F)
$ {\color{black}under $\preceq_i$-cEMA.}

By  Theorem \ref{prop-cxES}, VaR satisfies $\preceq_1$-cEMA,
and
ES  satisfies  $\preceq_2$-cEMA.
The more challenging question is in the opposite direction:
Are VaR and ES  the unique classes of risk measures, with some standard properties, that satisfies $\preceq_1$-cEMA and $\preceq_2$-cEMA, respectively?
This question is particularly important given the special roles of VaR and ES in banking practice.
We obtain  two main results: With the additional standard properties of  translation invariance, positive homogeneity, and lower semicontinuity,
$\preceq_1$-cEMA characterizes   VaR, and $\preceq_2$-cEMA  characterizes  ES. As far as we are aware, this is  the first time that  VaR and ES are axiomatized with parallel properties.

%
\begin{theorem}\label{th-VaRM}\label{th-ESM}
	\color{black} For a mapping $\rho:\mathcal M_1\to\R$, 
	\begin{enumerate}[(a)]
	\item  it satisfies translation invariance, positive homogeneity,  lower semicontinuity and $\preceq_1$-cEMA
	if and only if $\rho=\VaR_\alpha$ for some $\alpha \in (0,1)$;
	\item  it  satisfies  translation invariance, positive homogeneity, lower semicontinuity and $\preceq_2$-cEMA if and only if $\rho=\ES_{\alpha}$ for some $\alpha\in(0,1)$.
	\end{enumerate}
\end{theorem}

The special case of $\ES_0=\mathfrak{m}$ is excluded from Theorem \ref{th-ESM}, as it satisfies $\preceq_2$-cEMA (by Theorem \ref{prop-cxES}) but not lower semicontinuity. 
Theorem \ref{th-VaRM} states that $\preceq_1$-cEMA and $\preceq_2$-cEMA can identify VaR and ES, respectively.
In contrast to VaR which satisfies \eqref{eq210929-1} for any $\mathcal F$ bounded from above (Theorem \ref{prop-cxES}), ES  {fails to satisfy} \eqref{eq210929-1} for non-convex set $\mathcal F$ (Example \ref{ex-1es}).

\begin{remark}
	\label{rem:characterizeVaR}
	There are a few sets of axioms which characterize VaR,  each with the additional help of some standard properties such as continuity, monotonicity, translation invariance or positive homogeneity. In \cite{C09}, the main axiom for $\VaR$ is   ordinal covariance, an invariance property under some risk transforms. In \cite{KP16}, the main axioms for $\VaR$ are elicitability and comonotonic additivity.  In \cite{HP18}, the main axiom for $\VaR$ is surplus-invariance of the acceptance set.
	In \cite{LW20}, the main axioms are tail relevance and elicitability.
	In Theorem \ref{th-VaRM}, the new axiom of $\preceq_1$-cEMA leads to a characterization of $\VaR$,
	and  this new axiom standalone does not imply any axioms mentioned above.
\end{remark}

\begin{remark}\label{rem:characterizeES}
	ES is recently axiomatized by  \cite{WZ20} in the context of portfolio capital requirement. Their key axiom is called no reward for concentration (NRC) which intuitively means that a concentrated portfolio does not receive a diversification benefit.
	\cite{H21}, who also considered concentrated portfolio, obtain another characterization of ES by relaxing  NRC.  
	Another characterization result on ES is obtained by \cite{EMWW21} based on elicitability and Bayes risk.
	In contrast, our characterization result does not involve the consideration of elicitability or  portfolio risk aggregation. Therefore, the interpretation of Theorem \ref{th-ESM} is quite different from results in the literature and can be applied to robust modeling outside of a financial or statistical context.
\end{remark}


\begin{remark}
Equivalence in model aggregation has some similarity to  max-stability studied by \cite{KZ20}, which is defined on the set of random variables with the natural order, i.e., $X\preceq Y$ if and only if $X\le Y$ pointwisely. This leads to completely different interpretations and mathematics.
\end{remark}

\section{Numerical results for financial data}\label{sec:Numerical}
				In this section, we report  some numerical experiments based on real financial data to show the performance of the MA approach. 
				We select 20 stocks and collect their historical loss data   from Yahoo! Finance.\footnote{They are AAPL, MSFT, GOOGL, AMZN, ADBE, NFLX, AMD, V, JNJ, COST, WMT, PG, MA, UNH, DIS, HD, INTC, PYPL, GS, IBM.}
				We use  the period of January 1, 2019, to August 1, 2021,  with a  total of 649 observations of
				the daily losses of the 20 stocks.

				We shall conduct two sets of numerical experiments. First, in Section \ref{sec:71}, we present the robust distributions based on the MA approach when the uncertainty set consists of finite cdfs generated from the historical data, and compare the robust risk values with the WR ones. This analysis is based on data of single asset, and we   only report results on AAPL  for a simple illustration. Second, in Section \ref{sec:72}, we consider the application of the MA approach with Wasserstein and mean-variance uncertainty as in Section  \ref{sec:US}, and data of all 20 stocks will be used.
				
				\subsection{Performance of MA  with finite uncertainty set}\label{sec:71}
				We examine the MA approach for the uncertainty set {\color{black}that} consists of the cdfs generated by the real portfolio data AAPL.
				We use \texttt{Matlab} to fit the data with normal,  t-
				and logistic distributions that will be denoted by $F_{\rm n}$, $F_{\rm  t}$ and $F_{\rm  lg}$, respectively, and the empirical cdf is denoted by $\hat{F}$. Let $\mathcal F$ be the uncertainty set that consists of these four cdfs, i.e., $\mathcal F=\{\hat{F},F_{\rm  n},F_{\rm t},F_{\rm  lg}\}$.  

				Figure \ref{fig-fit} (top panels) shows the cdfs and integrated survival functions defined by \eqref{eq-pi} of the cdfs in $\mathcal F$. 
				Noting that $\bigvee_1 \mathcal F(x)=\inf\{\hat{F}(x),F_{\rm n}(x),F_{\rm t}(x),F_{\rm lg}(x)\}$ for $x\in\R$,
				the supremum $\bigvee_1 \mathcal F$ can be roughly divided into four parts. By Proposition \ref{th-sharp}, $\bigvee_2 F=F^*\in\mathcal F$ on  $(a,b)$ if $F^*$ has the largest value of integrated survival function on $(a,b)$.
				Hence, the figure of integrated survival functions illustrates $\bigvee_2\mathcal F$ can be divided into three parts: ${\bigvee_2\mathcal F}={F_{\rm n}}$ on $(-\infty,0.02)$; ${\bigvee_2\mathcal F}={\hat{F}}$ on $[0.02,0.0445)$; ${\bigvee_2\mathcal F}={F_{\rm t}}$ on $[0.0445,\infty)$.
				The curves of $\bigvee_1\mathcal F$ and $\bigvee_2\mathcal F$ are given in Figure \ref{fig-fit} (bottom panel)  from which we can see that $\bigvee_2\mathcal F\preceq_1 \bigvee_1\mathcal F$. Moreover, $\bigvee_2\mathcal F$ has a jump at $0.02$ which can be explained by the difference between left and right derivatives of the integrated survival function of $\bigvee_2\mathcal F$ at $0.02$.
				\begin{figure}[t]
					\caption{Left: cdf;  Middle: Integrated survival function; Right: Suprema of $\mathcal F$}
					\begin{center}
						\includegraphics[width=5.4 cm]{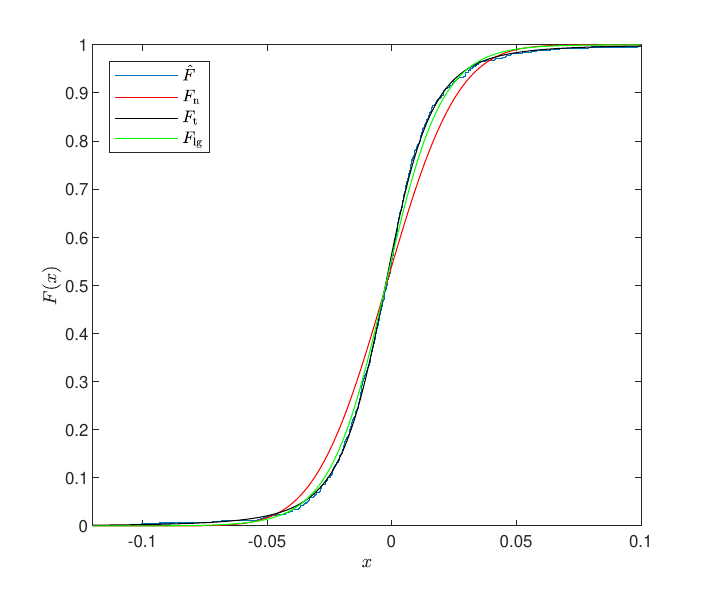}
						\includegraphics[width=5.4 cm]{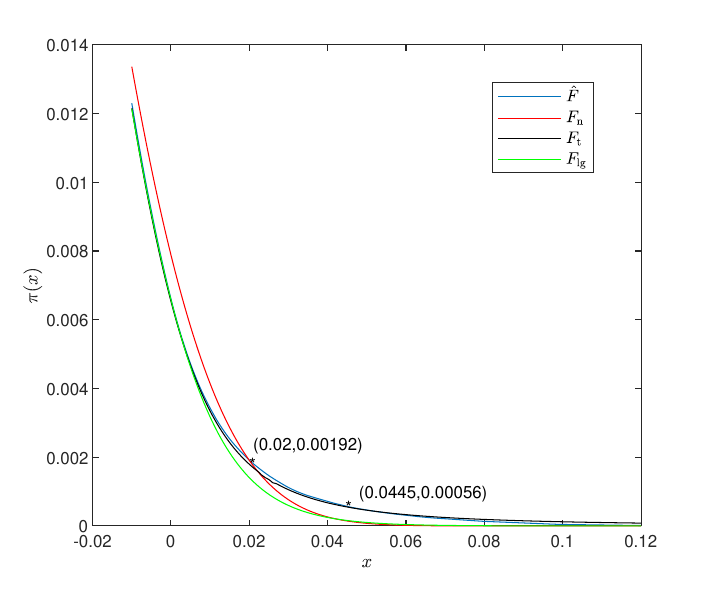}
						\includegraphics[width=5.4 cm]{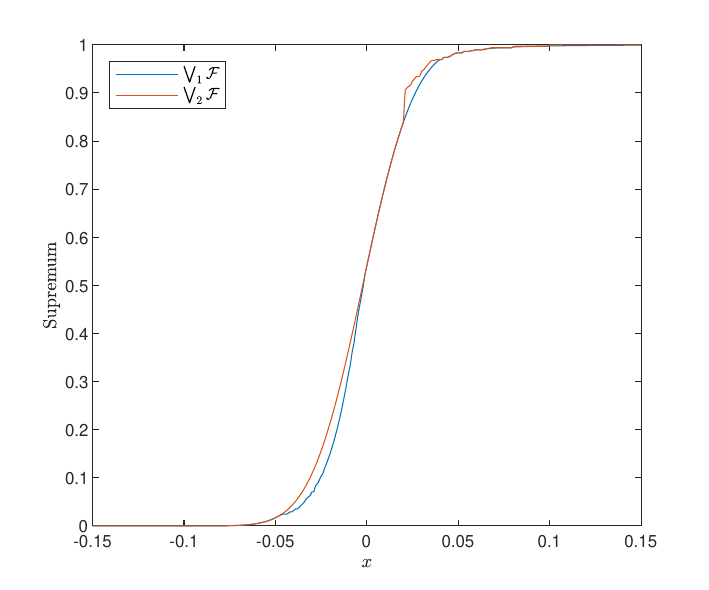}
						\label{fig-fit}
						\medskip
					\end{center}
				\end{figure}

		In the following, we compare  the ${\rm MA}_1$ and ${\rm MA}_2$ robust risk values and the WR ones with the uncertainty set $\mathcal F$. The risk measures are RVaR or ES. In the case of ${\rm RVaR}_{\alpha,\beta}$, we set $\alpha=0.95$ and let $\beta$ 
 vary in $[0.95,1]$. In the case of $\ES_{\alpha}$,   $\alpha$ varies in $[0.9,0.99]$.  
		
		Figure \ref{RVaR_robust} shows the value of RVaR$_{\alpha,\beta}$ of the cdfs in $\mathcal F$,  and RVaR$_{\alpha,\beta}$ based on the MA and WR approaches, and Figure \ref{ES_robust-11} presents the results of ES. From both figures we can see that the MA robust risk value is larger than the  WR  one. Moreover, from Figure \ref{RVaR_robust}, one can find that these two robust approaches have identical performance for $\beta\in[0.95,0.9685]$. This is because the quantile function of $F_{\rm n}$ dominates other elements in $\mathcal F$ on $[0.95,0.9685]$ which implies that ${\rm RVaR}_{0.95,\beta}^{{\rm MA}_1}(\mathcal F)= {\rm RVaR}_{0.95,\beta}^{\rm WR}(\mathcal F)={\rm RVaR}_{0.95,\beta}(F_{\rm n})$ for $\beta\in [0.95,0.9685]$.  From Figure \ref{ES_robust-11}, we find that $\ES_{\alpha}(F_{\rm n})$ and $\ES_\alpha(F_{\rm lg})$ are always smaller than $\ES_{\alpha}(\hat{F})$ and $\ES_\alpha(F_{\rm t})$ for $\alpha\in[0.9,0.99)$. The reason is that financial market loss data are heavy-tailed empirically (see e.g., \cite{MFE15}), and  ES with high level $\alpha$ focuses on the tail loss. 
		In addition, the curve of $\ES^{\rm MA}$ always lies above the curve of  $\ES^{\rm WR}$, which implies that the MA approach is more conservative.

		\begin{figure}[t]
			\caption{RVaR for individual models, via WR and via MA ($\alpha$=0.95)}
			\begin{center}
				\includegraphics[width=7cm]{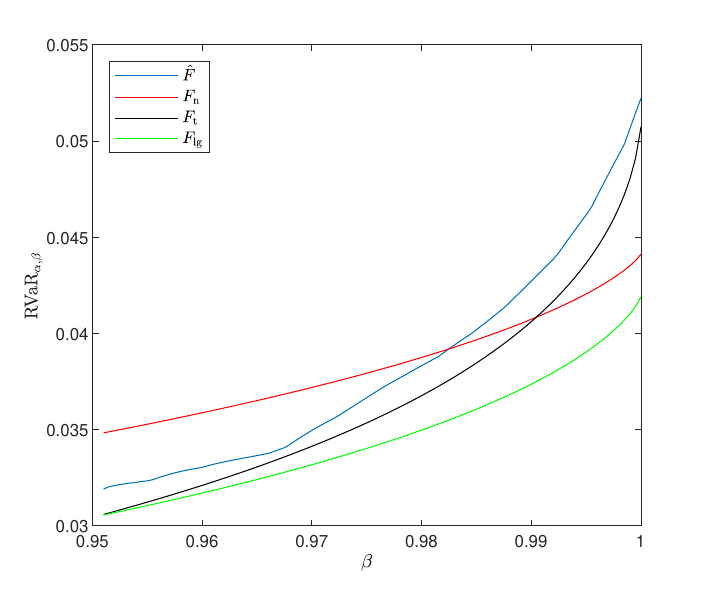}
				\includegraphics[width=7cm]{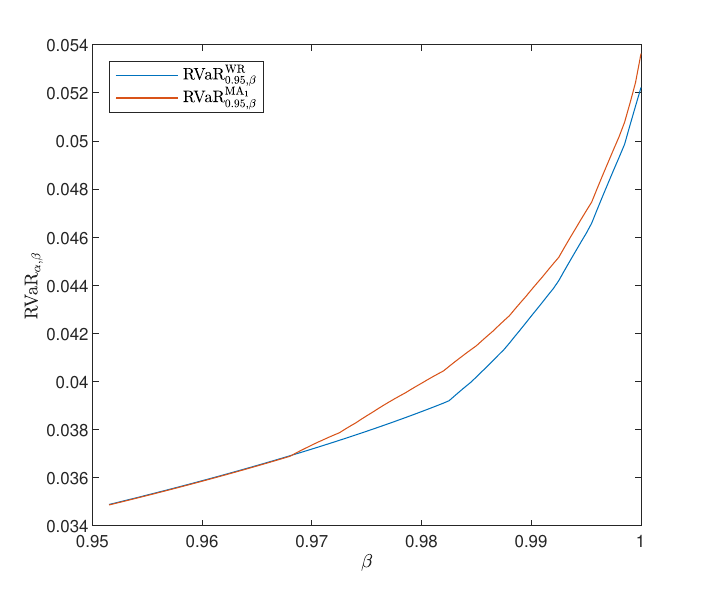}
				\label{RVaR_robust}
				\medskip
			\end{center}
		\end{figure}

		\begin{figure}[t]
			\caption{ES for individual models, via WR and via MA}
			\begin{center}
				\includegraphics[width=7cm]{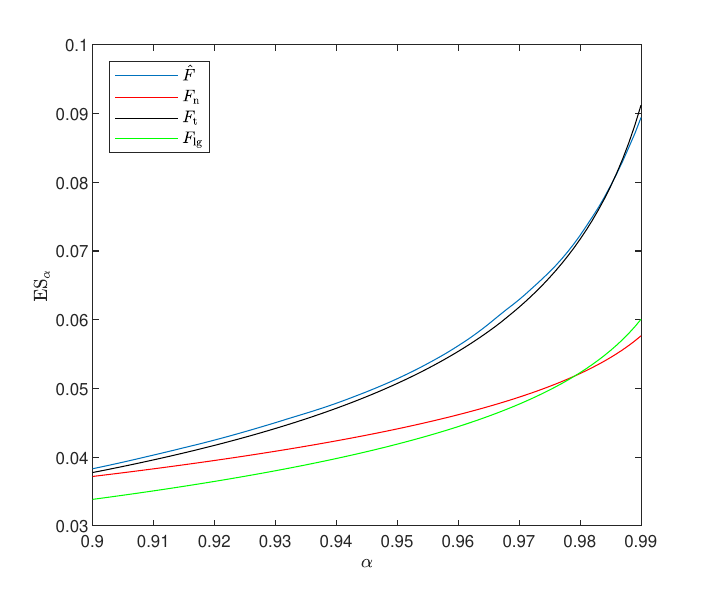}
				\includegraphics[width=7cm]{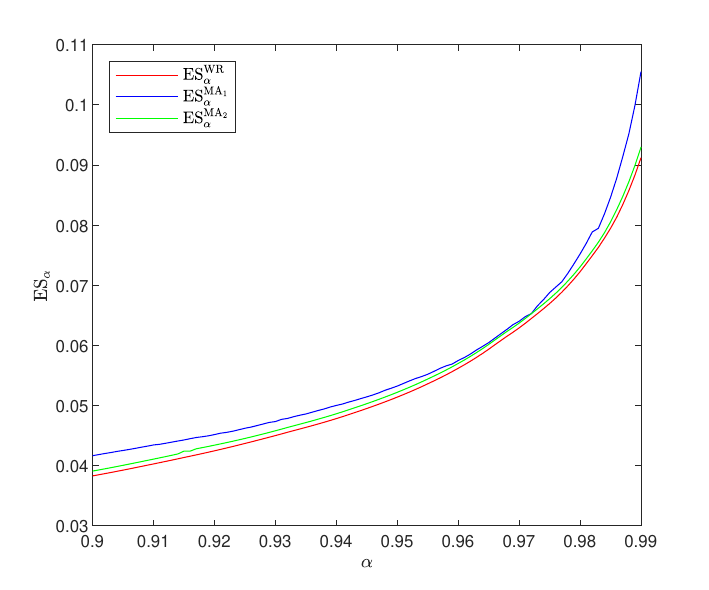}
				\label{ES_robust-11}
				\medskip
			\end{center}
		\end{figure}

		\subsection{MA  approach in robust portfolio selection}
		\label{sec:72}
		In this section, we consider the application of the MA approach with $\preceq_2$  
		in the setting of portfolio selection in Section \ref{sec:US}.
		The MA approach will be contrasted to the WR approach, and the standard sample average approximation (SAA) approach, which evaluates risks using the empirical distribution; see e.g., \cite{SDR21}.
		We construct a portfolio from the 20 stocks mentioned in the beginning of this section, whose daily losses are denoted by $X_1$ (AAPL), $X_2$ (MSFT), ..., $X_{20}$ (IBM). {\color{black}Mean, variance and correlation matrix of the return rate of the 20 stocks are given in Appendix \ref{app:portfolio}.} The wealth invested in the asset $X_i$ is denoted by $w_i$ for $i\in[20]$. Thus, the total loss from the investment of these 20 stocks is $\mathbf{w}^{\top}\mathbf{X}$, where  $\mathbf{w}=(w_1,\dots,w_{20})$ and $\mathbf{X}=(X_1,\dots,X_{20})$. The feasible region of $\mathbf{w}$ is the standard simplex
		$\Delta_{20}.$

We consider the setting in Section \ref{sec:US} where   uncertainty is modeled by a multi-dimensional Wasserstein ball.
For the choice of the risk measure $\rho$, we  will work with  ${\rm PD}_k$ defined  in Section \ref{sec:US} to measure the portfolio risk. There are a few reasons for this choice. First, ${\rm PD}_k$ is  $\preceq_2$-consistent (which also implies $\preceq_1$-consistency). Second, the   WR   and ${\rm MA}_2$ approaches are similar in the portfolio optimization problem under  the Wasserstein or the mean-variance uncertainty if the risk measure is selected as ES or expectile, so we move away from these two   choices. 
Third, the portfolio optimization problem of ${\rm PD}_k$  leads to a convenient formulation of SOCP under the Wasserstein or the mean-variance uncertainty as in Section \ref{sec:US}.

As in many classic settings of portfolio selection, e.g., the classic framework of \cite{M52}, we assume that the investor has a target level of expected {\color{black}annualized return rate} and minimizes the risk. That is,  with the constraint $\E[{\mathbf w}^{\top}{\mathbf X}]\le -r_0/m$ where {\color{black}$r_0$ is the expected annualized return rate and $m=250$}, the investor minimizes $\rho(F_{{\mathbf w}^{\top}{\mathbf X}})$. 

We set the parameter $a=p=2$ in the Wasserstein uncertainty ball $\mathcal F_{a,p,\epsilon}^d(F_0)$, and use a  multivariate t-benchmark distribution $F_0$ fitted to the data. 
The case of  a normal benchmark  distribution, which has a lighter tail, is reported in Appendix \ref{app:portfolio}.
For the whole-period data, the fitted t-distribution has $\nu=3.994$ degrees of freedom.
					The choice of a t-distribution is by no means restrictive, and we will consider the case of normal distribution which has a lighter tail in Appendix.
					We apply the WR  and the $\preceq_2$-MA approaches, and the corresponding  portfolio optimization problems are converted to  SOCPs which can be computed efficiently.
					By \eqref{eq:WassWR} and \eqref{eq:WassMA} in Section \ref{sec:63},
					the optimization problems via the WR   and the MA  approaches are,
					respectively,
					\begin{align}\label{eq-WWCOP}
						&\min_{{\mathbf w}\in\Delta_{20}}: ~~\rho^{\rm WR} \left( \mathcal F _{\mathbf w,2,2,\epsilon }(F_{0})\right)= {\mathbf w}^\top\bm\mu+{\rm PD}_k(F_{ \nu})\sqrt{{\mathbf w}^\top\Sigma {\mathbf w}}+\zeta_k\epsilon\sqrt{{\mathbf w}^\top{\mathbf w}}~~~~{\rm s.t.}~{\mathbf w}^\top\bm\mu\le -r_0/m,
					\end{align}
					and
					\begin{align}\label{eq-WMAOP}
						\min_{{\mathbf w}\in\Delta_{20}}:~~ \rho^{{\rm MA}_2}\left( \mathcal F _{\mathbf w,2,2,\epsilon }(F_{0})\right)= {\mathbf w}^\top{\bm\mu}+{\rm PD}_k(F_{\nu})\sqrt{{\mathbf w}^\top\Sigma {\mathbf w}}+\xi_k\epsilon\sqrt{{\mathbf w}^\top{\mathbf w}}~~~~{\rm s.t.}~{\mathbf w}^\top\bm\mu\le -r_0/m,
					\end{align}
					where
					$\zeta_k= {k}/{\sqrt{2k-1}}$,
					$\xi_k=( \sqrt{\pi}\Gamma(k+1))/(2\Gamma(k+1/2))$,   $F_{\nu}$ is the unit variance t-distribution with the tail parameter $\nu$,
					and
					$\bm\mu$ and $\Sigma$ are the mean and the covariance of the fitted $F_0$,  {\color{black}i.e., the mean vector and the covariance of the whole-period data.}
					The SAA approach optimizes the portfolio according to the empirical cdf of the asset losses. 
					Figure \ref{Robust_Wt_PW_epsilon} presents the optimized risk value of SAA approach and the optimized robust risk values under Wasserstein   uncertainty with the WR and MA approaches for different values of $\epsilon$, $r_0$ and $k$ using the whole-period data. 
     In the left panel, the MA robust risk value is  larger than the  WR one, and both are generally larger than  that of SSA. This is consistent with our intuition as MA is more conservative than WR, and SAA is not a conservative method.
In the middle and the right panels we set $\epsilon=0.01$ and let $k$ and $r$ vary.
	In practice, the parameter  $\epsilon$ should not be too small; one may tune $\epsilon$ so that the empirical cdf remains in the Wasserstein ball.
     

					\begin{figure}[t]
						\caption{The optimized robust values of ${\rm PD}_k$ under Wasserstein uncertainty using the whole-period data. Left: $r_0=0.2$, $k=10$, $\epsilon\in[0,0.1]$; Middle: $r_0=0.2$, $\epsilon=0.01$, $k\in [1,20]$; Right: $k=10$, $\epsilon=0.01$, $r_0\in  [0.15,0.5]$}
						\begin{center}
							\includegraphics[width=5.4 cm]{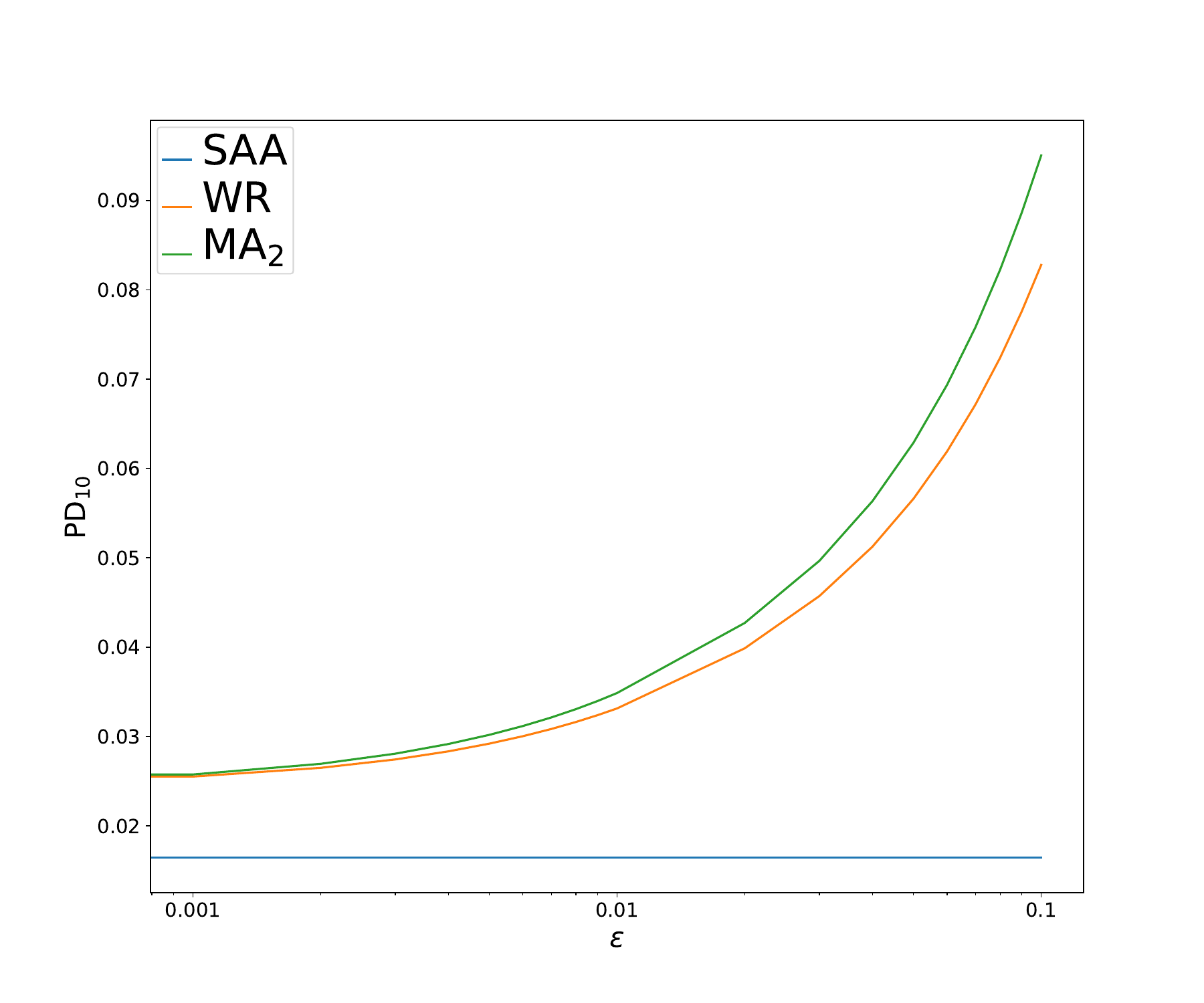}
							\label{Robust_Wt_PW_epsilon}
							\includegraphics[width=5.4 cm]{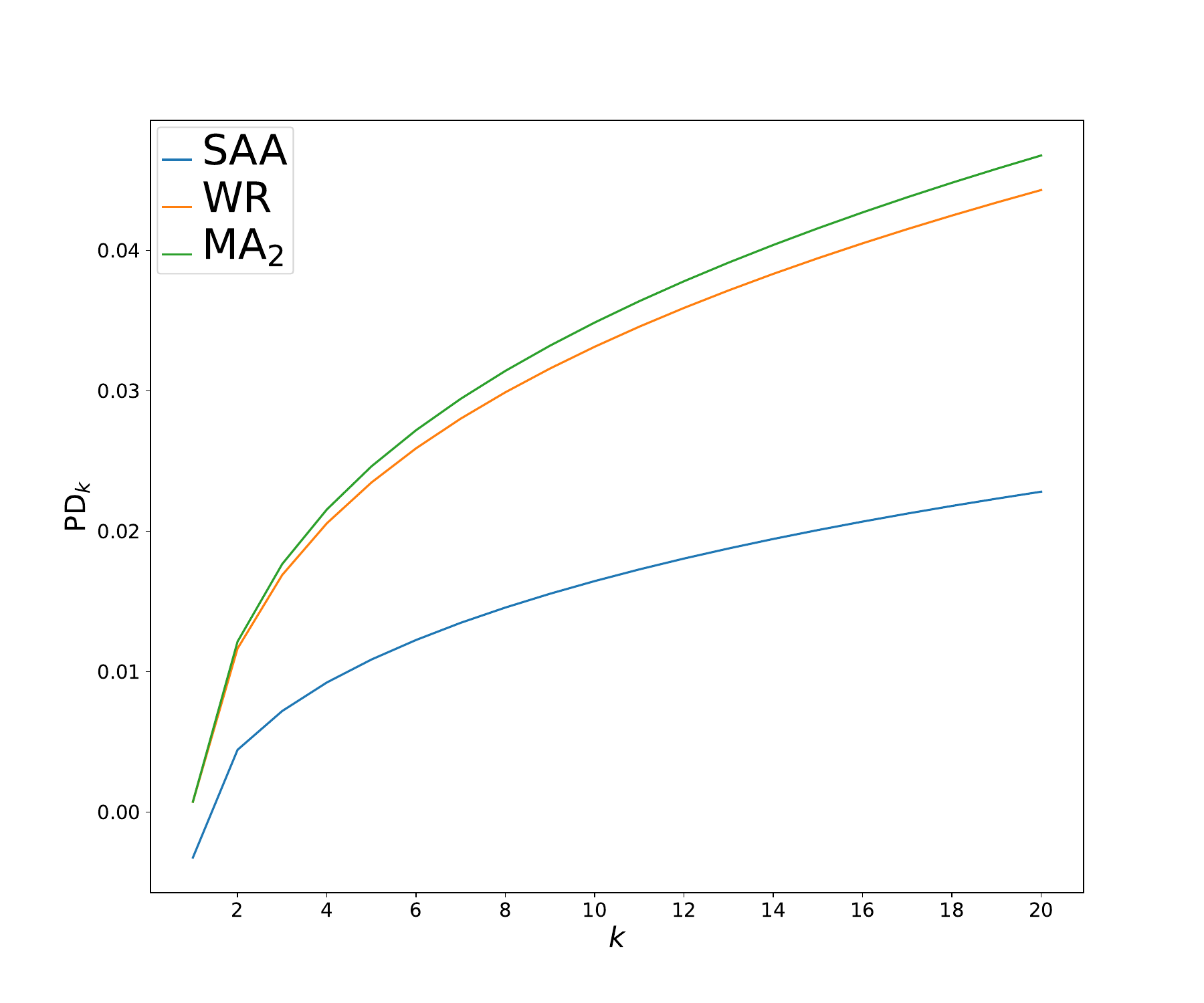}
							\label{Robust_Wt_PW_k}
							\includegraphics[width=5.4 cm]{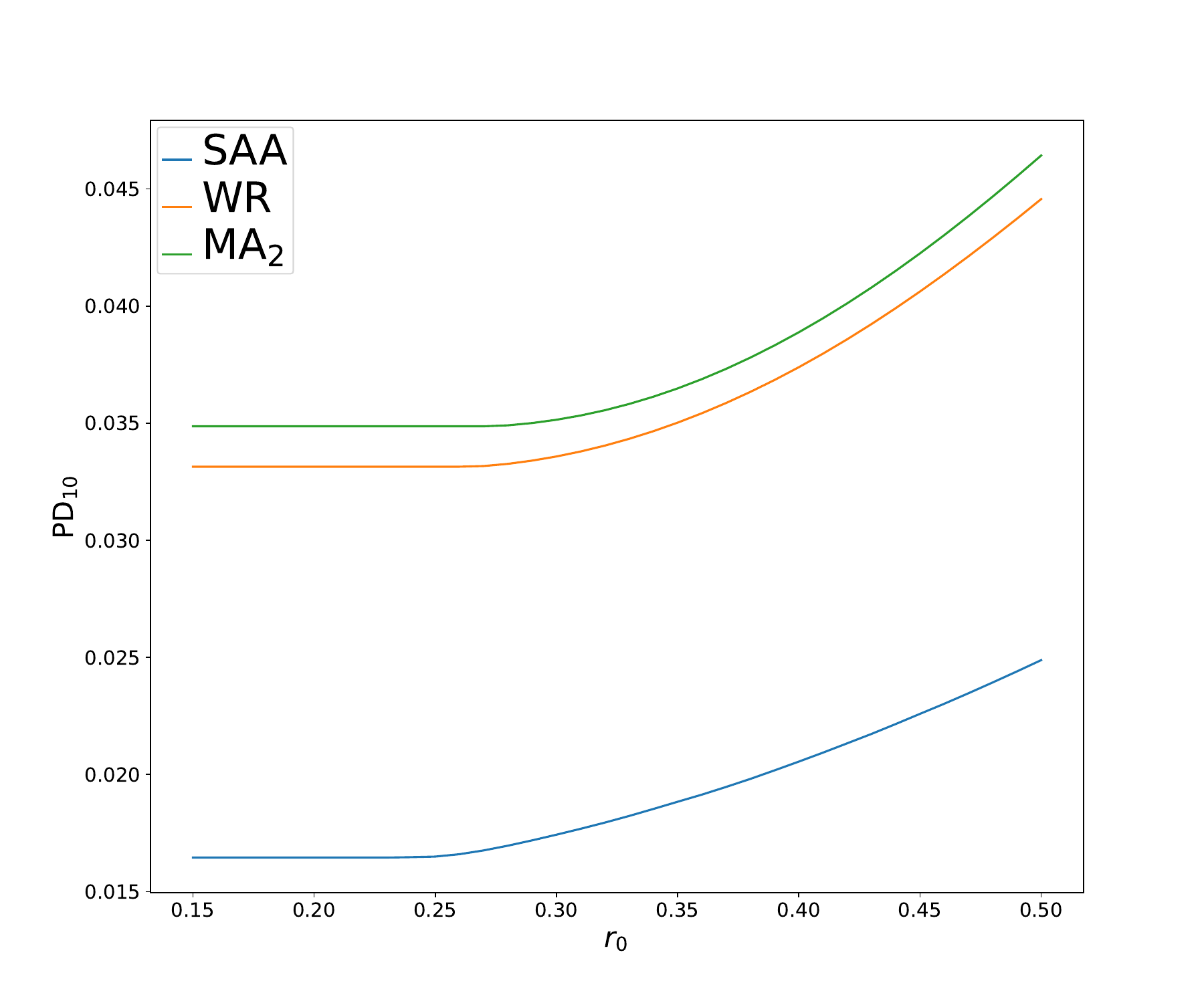}
							\label{Robust_Wt_PW_mu}
							\medskip
						\end{center}
					\end{figure}

					We choose slightly more than half of the period (350 trading days) for the initial training,
					and optimize the portfolio weights in each day with a rolling window. That is, on each trading day starting from day 351 (roughly June 2020), the preceding 350 trading days are used to fit the benchmark distribution, and compute the optimal portfolio weights. Note that the parameter $k$ reflects the degree of risk aversion of the decision maker, that is, a  larger  value of $k$ indicates a more risk-averse decision maker, and thus a larger corresponding risk measure. In this experiment, we choose $k=2$ and $20$,  
					and the decision maker with $k=20$ is more risk-averse than the one with $k=2$.
					{\color{black} Figure \ref{rollingwindow_Wt} depicts the performance of SAA approach, the mean-variance model of  \cite{M52} (minimizing  the variance of $\mathbf w^{\top}{\mathbf X} $ subject to $\E[{\mathbf w}^{\top}{\mathbf X}]\le -r_0/m$), and the MA and WR approaches under  Wasserstein uncertainty} over the remaining 300 trading days with $r_0=0.2$ and $\epsilon=0.01$, and we set $k=2$ (left) or $k=20$ (right). \textcolor{black}{Table \ref{tab:TC-perf} presents  realized annualized return rates and Sharpe ratios of all methods.}
					In all results, the MA and WR approaches,  being robust methods, perform   similarly. 
					{\color{black} MA and WR generally outperform the SAA  and the Markowitz  model, especially after the first 150 trading days. Intuitively, this means that, during the period from Jan to Aug 2021, robust investment strategies likely outperform non-robust strategies.} 
					The similar performance of the MA and  WR approaches under Wasserstein uncertainty is not a coincidence due to the similarity of problems \eqref{eq-WWCOP} and \eqref{eq-WMAOP} by noting that $\epsilon$ is small. 
				}
					
\begin{figure}[t]
\caption{Wealth evolution for different portfolio strategies from May 2020 to Aug 2021 ($\epsilon=0.01$, $r_0=0.2$).
 Left: $k=2$; Right: $k=20$}
						\begin{center}
							\includegraphics[width=7cm]{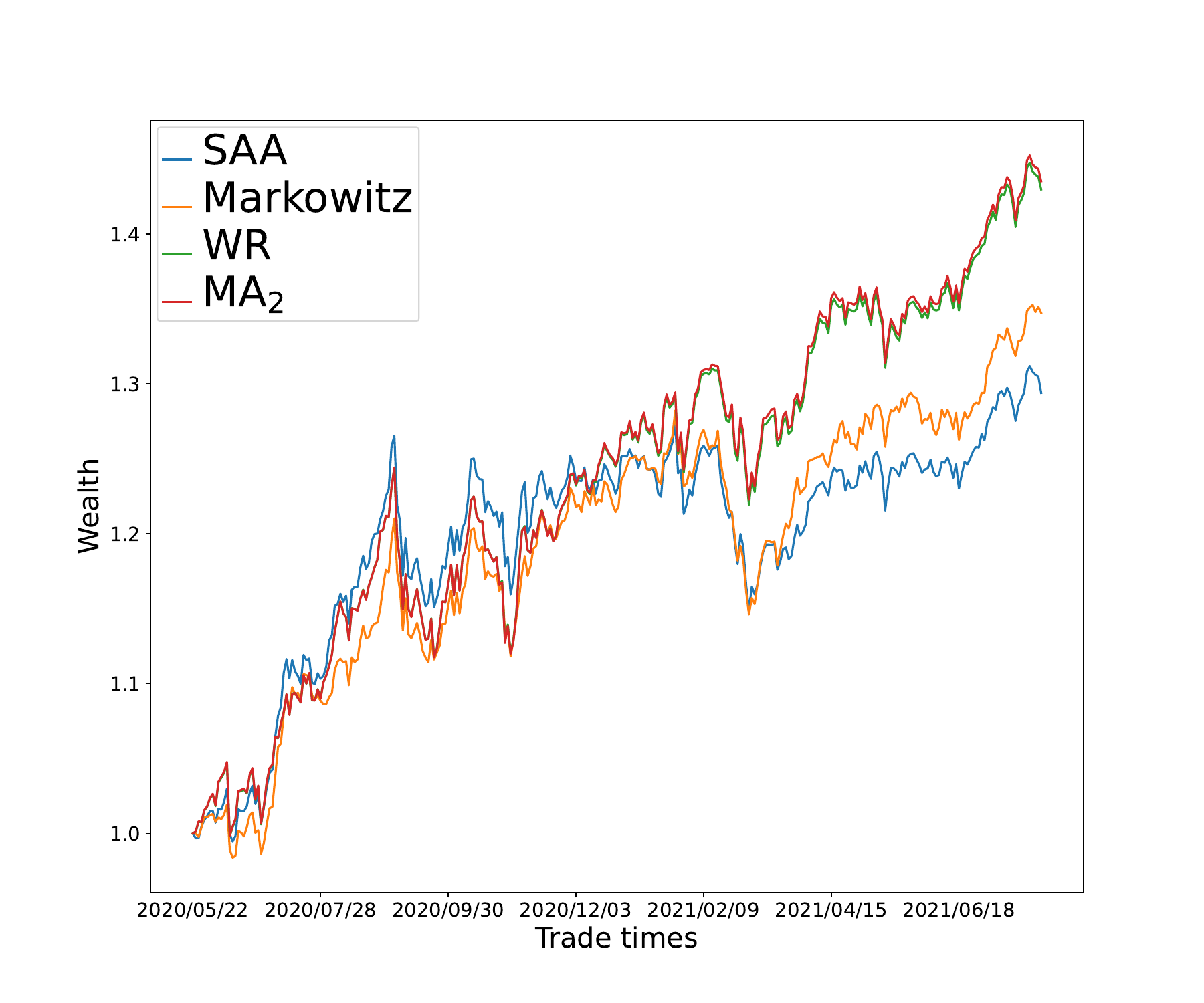}
							\includegraphics[width=7cm]{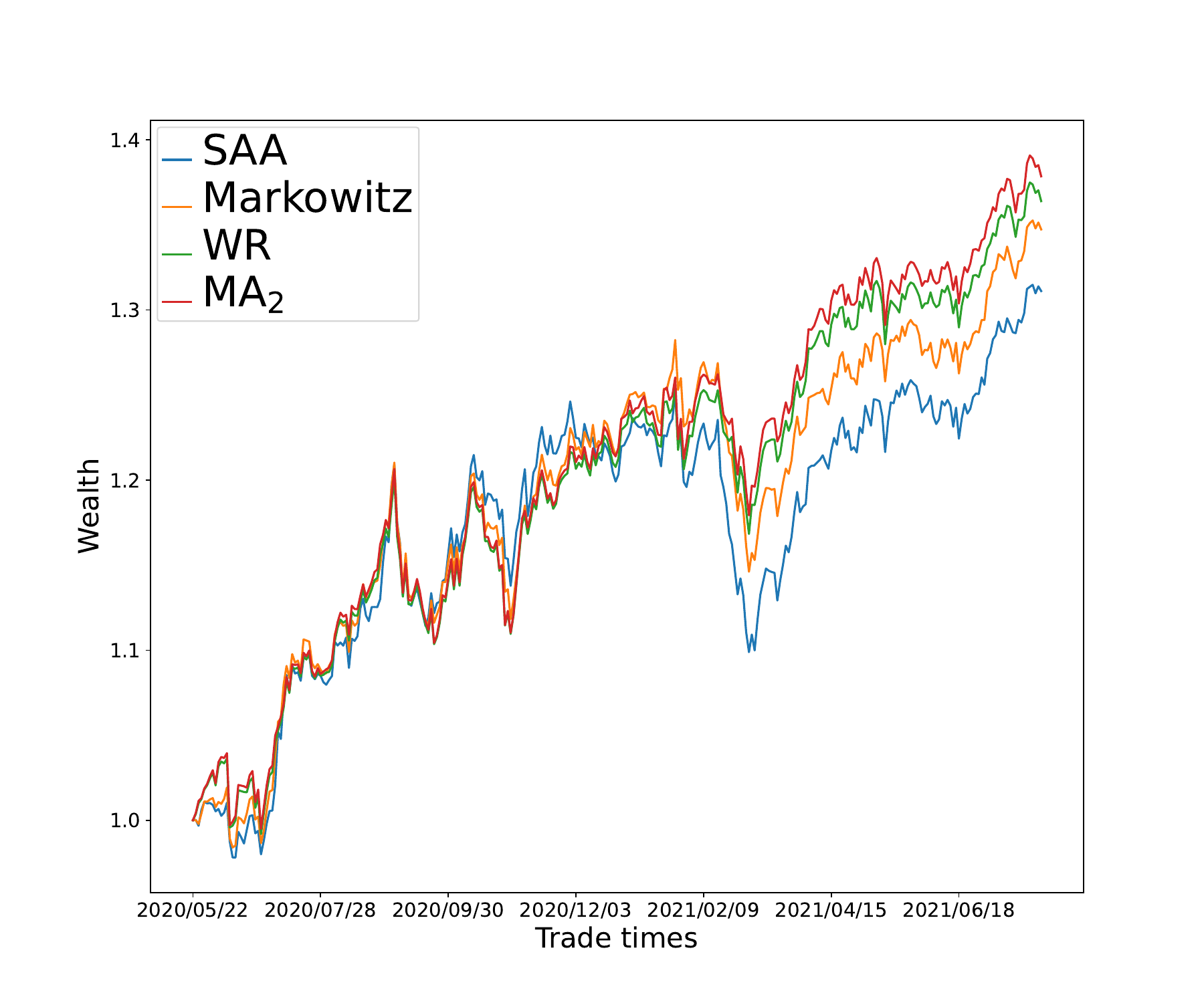}
							\label{rollingwindow_Wt}
							\medskip
						\end{center}
\end{figure}

\begin{table}[t]
\color{black}
\begin{center}
\caption{\color{black}Annualized return (AR), annualized volatility (AV) and Sharpe ratio (SR) for different strategies from May 2020 to Aug 2021; $r_0=0.2$ and  the risk-free rate  is  $0.165\%$ (1-year US Treasury yield on  May 22, 2020)
 }
	\label{tab:TC-perf}
 \renewcommand{\arraystretch}{1.2}
	\begin{tabular}{lcccccc}
		\toprule
		& \multicolumn{2}{c}{AR ($\%$)} & \multicolumn{2}{c}{AV ($\%$)} & \multicolumn{2}{c}{SR ($\%$)} \\
		\cmidrule(lr){2-3} \cmidrule(lr){4-5} \cmidrule(lr){6-7}
		Approach  & $ k = 2 $ & $ k = 20 $ & $ k = 2$ & $ k = 20 $ & $ k = 2$ & $ k = 20 $ \\
		\midrule
			SAA & 25.42 & 26.72 & 14.82 & 14.35 & 170.4 & 185.0 \\
	Markowitz  & 29.50 & 29.50 & \textbf{13.54} & 13.54 & 216.6 & 216.6 \\
	WR   & 32.75 & 30.79 & 14.25 & \textbf{13.30} & 228.7 & 230.3 \\
	${\rm MA}_2$ & \textbf{33.77} & \textbf{32.01} & 14.48 & 13.56 & \textbf{232.0} & \textbf{234.9} \\
		\bottomrule
\end{tabular}
\end{center}
~\\[5pt]
\end{table}

{\color{black}Table \ref{tab:transaction_cost} presents the nominal transaction cost  for different strategies by using the average weight change $\sum_{t=1}^{T} \|\mathbf w_{t+1}-\mathbf w_{t}\|_1/T$
where $\mathbf w_{t}$ is the weight used on day $t\in [T]$ by each strategy (see \cite{OD18}). 
 The MA and WR approaches based on Wasserstein uncertainty have similar  transaction costs, which are smaller than the other methods in most cases. A similar  analysis using the mean-variance uncertainty is reported in Appendix \ref{app:portfolio}.}

\begin{table}[htb]
\centering \color{black}
\caption{\color{black}Nominal transaction cost $\sum_{t=1}^{T} \|\mathbf w_{t+1}-\mathbf w_{t}\|_1/T$ with $\epsilon=0.01$ and $T=299$}
	\label{tab:transaction_cost}
	\begin{tabular}{lcccccc}
		\toprule
		& \multicolumn{2}{c}{$ r_0 = 0.1 $} & \multicolumn{2}{c}{$ r_0 = 0.2 $} & \multicolumn{2}{c}{$ r_0 = 0.3 $} \\
		\cmidrule(lr){2-3} \cmidrule(lr){4-5} \cmidrule(lr){6-7}
		Approach & $ k = 2 $ & $ k = 20 $ & $ k = 2$ & $ k = 20 $ & $ k = 2$ & $ k = 20 $ \\
		\midrule
		SAA & 0.0549 & 0.0071 & 0.0871 & 0.0121 & 0.0969 & 0.0833 \\
	Markowitz  & \textbf{0.0102} & 0.0102 & 0.0110 & 0.0110 & 0.0746 & 0.0746 \\
	WR   & 0.0127 & \textbf{0.0032} & 0.0127 & \textbf{0.0033} & 0.0271 & 0.0382 \\
	${\rm MA}_2$ & 0.0114 & 0.0035 & \textbf{0.0105} & 0.0035 & \textbf{0.0239} & \textbf{0.0348} \\
		\bottomrule
	\end{tabular}
\end{table}

					\section{Concluding remarks and discussions}\label{sec:conclude}
					The MA approach  for robust risk evaluation is proposed. 
					Below, we summarize some the advantages of the MA approach, which are illustrated and discussed through several technical results. 
					\begin{enumerate}
						\item The MA approach is natural to interpret, and it is motivated by the need for a robust distributional model.
						The WR approach is also natural to interpret, but the focus is on the risk value instead of the risk model. Different from the WR approach, the MA approach is built on  stochastic orders and lattice theory.
						The robust model produced by the MA approach can be readily applied to different risk evaluation procedures and decision problems (Section \ref{sec:31}). 
						\item The MA robust risk value is straightforward to compute (Section \ref{sec:32}). 						In some settings of uncertainty, the MA approach leads to explicit formulas for the robust model (Section \ref{sec:US}). In particular, it can   handle Wasserstein uncertainty in portfolio selection, based on 
						a new dimension reduction result on Wasserstein balls (Theorem \ref{th:7}).
						\item The MA approach admits reformulations in distributionally robust optimization similar to the WR approach, and it leads to a convex program when the loss function and the risk measure are convex  (Section \ref{sec:34}).
						\item The MA approach gives rise to the useful property of cEMA which characterizes  VaR and ES (Section \ref{sec:charac}). These results reveal a profound connection of the popular regulatory risk measures to robust risk evaluation methods, and  highlight the special roles of VaR and ES among all risk measures, which is in itself a highly active research topic in risk management. 
					\end{enumerate}

					The MA approach  requires a stochastic order to be specified. For an interpretation of prudent risk evaluation as in \eqref{eq-gap}, the  risk measure of interest should be consistent with this stochastic order.
					We recommend, in most applications, using $\preceq_2$ in an MA approach as the default option, for its nice interpretation in decision theory (strong risk aversion) and   mathematical properties as   developed  in this paper.

					We have focused on studying the MA and WR approaches together with risk measures throughout the paper.  Both approaches can be easily applied to other objectives other than risk measures, such as expected utility functions, rank-dependent expected utilities, or other behaviour   decision criteria. Some decision criteria  may work better with  notions of stochastic dominance other than FSD and SSD, and they may include considerations of model uncertainty by design; see e.g., \cite{HS01}, \cite{MMR06} and \cite{CHMM21}.

					Our theory is built on model spaces of univariate  cdfs on $\R$ for the following reasons. First, classic risk measures, especially the ones used in regulatory practice such as VaR and ES,  are defined on one-dimensional cdfs representing potential (portfolio) losses; second, commonly used stochastic orders, the key tool to build robust model aggregation  in this paper, are usually defined on one-dimensional cdfs and they are naturally interpretable in this setting; third, many problems that are multivariate in natural often boil down to   robust risk evaluation in one-dimension; see the settings in Sections  \ref{sec:63}, \ref{sec:MV}  and \ref{sec:72}.
					If desired by specific applications, the theory of the MA approach can be readily extended to a multi-dimensional setting (see e.g., \cite{EP06}) with the help from multivariate stochastic orders (e.g., \cite{SS07}) and set-valued risk measures (e.g., \cite{HH10}, \cite{HHR11} and \cite{AHR17}).

					In addition to the multi-dimensional extension mentioned above,
					we mention a few promising directions of future study.
					First, one can consider the recently introduced notions of fractional stochastic dominance of \cite{MSTW17} and \cite{HTZ20}, which generalize the first- and second-order stochastic dominance used in this paper.
					Second, instead of relying only on the set $\mathcal F$ of uncertainty, which treats each cdf as an element of equal importance ex ante,  we can equip  a prior probability measure $\mu$ on set $\mathcal F$, and this will open up  many new challenges or conceptualizing  and constructing  robust models in a similar framework to our theory.
					Third, we can apply the MA approach to many other settings of uncertainty other than the ones studied in Section \ref{sec:US}, and this will lead to convenient tools in various  new applications  and contexts.



					
					
					

\newpage

\begin{center}
{\Large \bf  Online Supplement: Technical Appendices}

\medskip

{\Large Model Aggregation for Risk Evaluation  and Robust Optimization}

\medskip

\end{center}

\begin{appendix}
\setcounter{table}{0}
\setcounter{figure}{0}
\setcounter{equation}{0}
\renewcommand{\thetable}{EC.\arabic{table}}
\renewcommand{\thefigure}{EC.\arabic{figure}}
\renewcommand{\theequation}{EC.\arabic{equation}}

\setcounter{theorem}{0}
\setcounter{proposition}{0}
\renewcommand{\thetheorem}{EC.\arabic{theorem}}
\renewcommand{\theproposition}{EC.\arabic{proposition}}
\setcounter{lemma}{0}
\renewcommand{\thelemma}{EC.\arabic{lemma}}

\setcounter{corollary}{0}
\renewcommand{\thecorollary}{EC.\arabic{corollary}}

\setcounter{remark}{0}
\renewcommand{\theremark}{EC.\arabic{remark}}
\setcounter{definition}{0}
\renewcommand{\thedefinition}{EC.\arabic{definition}}

We organize the appendices as follows.
We first introduce some  extra notation and terminology in Appendix \ref{app:0}.
In Appendix \ref{app:lattice}, we formally introduce the lattice theory and prove a generalized version of Proposition \ref{th-sharp} in Section \ref{sec:4}.
The proofs of the main results 
are presented in Appendices \ref{app:A} (other results in  Section \ref{sec:4}), \ref{app:optMA} (results in  Section \ref{sec:34}), \ref{app:E} (results  and omitted figures in Section \ref{sec:US}) and
\ref{app:C} (results in Section \ref{sec:EMAM}).
In Appendix \ref{app:portfolio}, we
present the summary statistics of the return rates as well as the 
numerical results in portfolio selection under mean-variance uncertainty and Wasserstein uncertainty with a normal benchmark distribution, which complements the numerical studies in Section \ref{sec:Numerical}.

\section{Setting and notation}
\label{app:0}

We will use the same notation  as in the main paper. In addition, let $L^p$ be the space of random variables in $(\Omega,\mathcal B,\p)$ with finite $p$th moment, $p\in[0,\infty)$, and $L^\infty$ be the space of all bounded random variables. Accordingly, denote by $\mathcal M_p$, $p\in[0,\infty]$, the set of cdfs of all random variables in $L^p$, i.e., $\mathcal M_p$ is the set of all cdfs $F$ satisfying $\int_{\R}|x|^p\d F(x)<\infty$ for $p\in[0,\infty)$, and $\mathcal M_{\infty}$ is the set of all compactly supported cdfs.
On $\M_\infty$, we can define $\VaR_0$, $\VaR_1$ and $\ES_1$ which are   finite, by  $$\VaR_0 (F)= \inf \{x\in \R: F(x) > 0\}~\mbox{ and }~ \VaR_1 (F)= \inf \{x\in \R: F(x) \ge 1\},~~~F\in \M_{\infty},
$$
and $$
\ES_{1}(F)={\rm VaR}_1(F),~~~F\in \M_{\infty}.$$
Denote by
$\mathcal M_{\rm bb}$ the set of all cdfs $F$ with a support bounded from below, i.e., $F(x_0)=0$ for some $x_0\in\R$.
For two real objects (numbers or functions) $f$ and $g$, $f \vee g$ is their (point-wise) maximum, and $f \wedge g$ is their (point-wise) minimum.



{\color{black}
For $p\in[0,\infty]$, 
the following are some properties of a risk measure $\rho:\mathcal M_p\to\R$ and its associated $\widetilde \rho:L^p\to \R$.  
{\it Translation invariance}:
$\widetilde \rho( X+c)=\widetilde \rho(X)+c$ for any $c\in\R$ and $X\in L^p$.
{\it Positive homogeneity}:
$\widetilde \rho({\lambda X} )=\lambda\widetilde \rho(X)$ for any $\lambda>0$ and $X\in L^p$.
{\it Convexity}:
$\widetilde \rho({\lambda X+(1-\lambda)Y}) \le \lambda \widetilde \rho(X) + (1-\lambda)\widetilde \rho(Y)$ for any $\lambda \in [0,1]$ and $X,Y\in L^p$.
{\it Lower semicontinuity}: $\liminf_{n\to\infty} \widetilde \rho(X_{n})\ge \widetilde \rho(X)$ if $X_n,X\in L^p$ for all $n$ and 
$X_n\stackrel{\rm d}\to X$ as $n\to\infty$, where $\stackrel{\rm d}\to$ denotes convergence in distribution. All the properties are defined for both $\rho$ and $\widetilde \rho$.
}

\section{Lattice theory and the proof of Proposition \ref{th-sharp}}\label{app:lattice}

In this appendix, we introduce  the lattice structure of an ordered set which complements the main paper.
For more details of the lattice theory, the reader is referred to \cite{DP02}. 

\begin{definition}\label{def-lattice}
	Let $(\M, \preceq)$ be an ordered set (i.e., $\preceq$ is a partial order on $\mathcal M$) and $\mathcal F\subseteq \M$.
	\begin{itemize}
		\item[(i)]  A set $\mathcal F$ is said to be \emph{bounded from above} (\emph{below}, resp.) in  $(\M,\preceq)$, if the set of upper (lower, resp.) bounds  on $\mathcal F$, denoted by  $U(\mathcal F)$ ($L(\mathcal F)$, resp.), is nonempty, where
		\begin{align}\label{eq-U}
			U(\mathcal F)=\{G\in\M: F\preceq G,~{\forall}F\in\mathcal F\}
			~~{\rm and}~~
			L(\mathcal F)=\{G\in\M: G\preceq F,~{\forall}F\in\mathcal F\}.
		\end{align}

		\item[(ii)]  For $\mathcal F\subseteq\M$ which is bounded from above (below, resp.), if there exists   $F_0\in U(\mathcal F)$ ($L(\mathcal F)$, resp.) such that $F_0\preceq$ ($\succeq$, resp.) $G$  for all $G\in U(\mathcal F)$ ($L(\mathcal F)$, resp.), then $F_0$ is called the \emph{supremum} (\emph{infimum}, resp.) of $\mathcal F$ and we write $\bigvee\mathcal F=F_0$ ($\bigwedge\mathcal F=F_0$, resp.).
		
		\item[(iii)] If for all $F,G\in \M$, $\bigvee\{F,G\}$ and $\bigwedge\{F,G\}$ exist, then $(\M, \preceq)$ is called a \emph{lattice}.
		If $\bigvee\mathcal F$ exists  for all $\mathcal F\subseteq\M$ that {\color{black}are} bounded from above and $\bigwedge\mathcal F$   exists for all $\mathcal F\subseteq\M$ that {\color{black}are} bounded from below, then $(\M, \preceq)$ is called a \emph{complete lattice}.\footnote{The definition of complete lattice in \cite{DP02} is slightly different to ours. In \cite{DP02}, a complete lattice has the largest and the smallest elements, and our $\mathcal M$ does not. Nevertheless, if we extend $\mathcal M$ to $\overline{\mathcal M}:=\mathcal M\cup\{F_{\min},F_{\max}\}$ where $F_{\min}\preceq F$ and $F\preceq F_{\max}$ for all $F\in\mathcal M$, then our definition of complete lattice on the ordered set $(\overline{\mathcal M},\preceq)$ is equivalent to the one of \cite{DP02}.}
	\end{itemize}
\end{definition}


\begin{remark}
	In case $(\M,\preceq)$ is a lattice which is not complete, $\bigvee \mathcal F$  may not exist even if $\mathcal F$ is bounded from above. In this case, the definition of the MA robust risk value needs to be modified. We can alternatively define
	{\color{black}$  \rho^{\rm MA}(\mathcal F) = \inf_{G\in U(\mathcal F)}\rho(G) $}
	where $U(\mathcal F)$ is defined by \eqref{eq-U},
	and this definition is equivalent to \eqref{eq-MA} if  $(\M,\preceq)$ is a complete lattice and $\rho$ is $\preceq$-consistent.
	{\color{black} For any partial order $\preceq$, the supremum and infimum are both unique whenever they exist.}
\end{remark}

For stochastic dominances $\preceq_1$ and $\preceq_2$, there are several equivalent definitions that are useful throughout the paper; see e.g., \cite{BM06}. In case of $\preceq_1$, the following statements are equivalent: (i) $F\preceq_1 G$; (ii) $F(x)\ge G(x)$ for all $x\in\R$; (iii) $F^{-1}(\alpha)\le G^{-1}(\alpha)$ for all $\alpha\in(0,1)$. In case of $\preceq_2$, the following statements are equivalent: (i) $F\preceq_2 G$; (ii) $\pi_F(x)\le \pi_G(x)$ for all $x\in\R$ where $\pi_F$ is the integrated survival function defined by \eqref{eq-pi}; (iii) $E_F(\alpha)\le E_G(\alpha)$ for all $\alpha\in(0,1)$ where $E_F$ is the \emph{integrated quantile function}\footnote{\color{black}The integrated quantile function $E_F$ is also called the (upper) absolute Lorenz function, see, e.g., \cite{S83} and \cite{C11}.} defined by
\begin{align}\label{eq-ESprofile}
	E_F(\alpha) = (1-\alpha)\ES_\alpha(F) = \int_{ \alpha}^1 F^{-1}(s) {\d} s,~~\alpha\in [0,1].
\end{align}

The complete lattice structure of $(\mathcal M_0,\preceq_1)$ and $(\mathcal M_+,\preceq_2)$
and the formulas for the suprema are known in the literature; see \cite{KR00}. Here $\mathcal M_+=\{F\in\mathcal M_0: \int_0^\infty x{\d}F(x)<\infty\}$.
The {\color{black}following} proposition which is a generalized result of Proposition \ref{th-sharp} considers general space $\mathcal M_p$, $p\in[0,\infty]$ with partial order $\preceq_1$ and $\preceq_2$.
Specifically, in Proposition \ref{prop-lattice-G} below,  (a)  generalizes  Proposition \ref{th-sharp} (a) to the domain $\mathcal M_p$, $p\in [0,\infty]$, and  similarly, (b) and (c) generalize Proposition \ref{th-sharp} (b).
\begin{proposition}\label{prop-lattice-G}
	\begin{itemize}
		\item[(a)] For each $p\in[0,\infty]$, the partially ordered set $(\mathcal M_p,\preceq_1)$ is a complete lattice. If $\mathcal F\subseteq\mathcal M_p$ is bounded from above, then its supremum $\bigvee_1\mathcal F$ is given by $\inf_{F\in\mathcal F}F$, and the left quantile function of $\bigvee_1\mathcal F$ is $\sup_{F\in\mathcal F} F^{-1}$.
		
		\item[(b)] The partially ordered set $(\mathcal M_1,\preceq_2)$ is a complete lattice, and for $\mathcal F$ that is bounded from above,
		\begin{equation*}
			\pi_{\bigvee_2\mathcal F}=\sup_{F\in\mathcal F}\pi_F, ~~~~~ 
			\bigvee_2 \mathcal F=1+\left(\sup_{F\in\mathcal F}\pi_F\right)'_+.
		\end{equation*}		
		\item[(c)] For each $p\in (1,\infty]$, the
		ordered set  $(\mathcal M_p,\preceq_2)$ is a lattice and not a complete lattice. The supremum is given by $\bigvee_2 \{F,G\}=1+( \pi_F\vee \pi_G )_+'$ for $F,G\in\mathcal M_p$.
	\end{itemize}
\end{proposition}

\begin{proof} We first give one fact: For $p\in [0,\infty]$ and  an increasing and right-continuous function $H:\R\to[0,1]$, if  $F,G\in \M_p$ and $F\le H\le G$, then $H\in\mathcal M_p$. It suffices to verify that
	\begin{itemize}
		\item[1.] $0\le \lim_{x\to-\infty} H(x)\le \lim_{x\to-\infty} G(x)=0$  and $1\ge \lim_{x\to \infty} H(x)\ge \lim_{x\to-\infty} F(x)=1$, which imply {\color{black}that $H$ is a cdf on $\R$}, that is, $H\in \M_0$.
		\item [2.] If $p\in (0,\infty)$, then we have $ F\succeq_1 H\succeq_1 G$
		and thus  $\int_0^\infty x^p\d H(x)\le \int_0^\infty x^p\d F(x)<\infty$ and $\int_{-\infty}^0 (-x)^p\d H (x)\le \int_{-\infty}^0 (-x)^p\d G(x)<\infty$. It follows that
		$
		\int_\R|x|^p\d H(x)<\infty,
		$ that is, $H\in\mathcal M_p$.
		\item [3.] If $F,G\in \M_\infty$, then there exists $x$, $y\in\R$ such that $G(x)=0$ and $F(y)=1$. Then we have $H(x)=0$ and $H(y)=1$, that is, $H\in\M_\infty$.
	\end{itemize}
	(a) 
	For $p\in[0,\infty]$, let  $\mathcal F\subseteq\mathcal M_p$.
	{\color{black}Suppose that} $\mathcal F$ is bounded from above. {\color{black}Define} $H=\inf_{F\in\mathcal F} F$ which is increasing and right-continuous.
	{\color{black}Then} there exists $G\in\mathcal M_p$ such that $F\ge H\ge G$ for any $F\in \mathcal F$.
	By the above fact, we have $H\in \M_p$.
	If $\mathcal F$ is bounded from below, define $H(x)=\lim_{y\downarrow x} H_1(y)$ where $H_1 = \sup_{F\in\mathcal F} F$. Then $H$ is increasing and right-continuous and  there exists $G\in\mathcal M_p$ such that $G\ge H\ge F$ for any $F\in \mathcal F$. By the above fact, we have $H\in \M_p$.
	Therefore, we have
	that $(\mathcal M_p,\preceq_1)$ is a complete lattice for $p\in[0,\infty]$.
	The statement on the left quantile of $\bigvee_1\mathcal F$ follows  from
	$(\inf_{F\in \mathcal F}  F)^{-1}  = \sup_{F\in \mathcal F} F^{-1}$.
	Hence, we complete the proof of (a).
	
	(b) The proof is similar to that of Theorem 3.4 of \cite{KR00} which shows that $(\mathcal M_+,\preceq_2)$ is a complete lattice. We give a proof for completeness. Let $\mathcal F\subseteq \mathcal M_1$ be bounded from above. There exists  $G\in\mathcal M_1$ such that $F \preceq_2 G$ for all $F\in\mathcal F$, that is, $\sup_{F\in\mathcal F}\pi_F(x) \le \pi_G(x)$ for $x\in \R$.  One can check that \begin{itemize} \item [1.] $\pi_0(x):=\sup_{F\in\mathcal F}\pi_F(x)$ is decreasing convex as each $\pi_F(x)$ is decreasing convex. This implies $ 1+ (\pi_0)_+'(x)$ is right-continuous and increasing.
		\item [2.] $\lim_{x\to\infty} \pi_0(x) \le\lim_{x\to\infty} \pi_G(x)= 0$
		which implies $\lim_{x\to\infty} (\pi_0)_+'(x)=0$, that is, $\lim_{x\to\infty} ( 1+ (\pi_0)_+'(x))=1$.
		\item [3.] Since $x+\pi_F(x)$ is increasing in $x$ for all $F\in\mathcal F$, we have
		$x+\pi_0(x)$ is increasing in $x$ and thus $ \lim_{x\to-\infty} x+ \pi_0(x)$ exists (may take $-\infty$). Let $F^*\in\mathcal F$, and we have $x+ \pi_0(x)\ge x+\pi_{F^*}(x)$ for all {\color{black}$x\in\R$.} Noting that $\lim_{\color{black}x\to-\infty} x+\pi_{F^*}(x) =\mathfrak{m}({F^*})\in\R$, we have $ \lim_{x\to-\infty} x+ \pi_0(x) \in\R$, which implies $\lim_{x\to-\infty} 1+ (\pi_0)_+'(x) =0$.
	\end{itemize}
	Combining the above three observations, we have $H=1+(\sup_{F\in\mathcal F}\pi_F)_+'\in\M_1$. 
	By definition of supremum, it is standard  to check that   $\bigvee_2\mathcal F=H$.
	
	Let $\mathcal F\subseteq \mathcal M_1$ be bounded from below. There exists  $G\in\mathcal M_1$ such that $G \preceq_2 F$  for all $F\in\mathcal F$, that is, $E_G(\alpha)\le \inf_{F\in\mathcal F}E_F(\alpha)$ for $\alpha\in [0,1]$. Similar to the proof of Steps 1-3 for $\mathcal F$ that is bounded from above, one can show that  $\inf_{F\in\mathcal F} E_F$ is an integrated quantile function of some cdf in $\M_1$, say $H$. By definition of infimum, we have  $H =\bigwedge_2\mathcal F$. It follows from the relation between a cdf and its integrated quantile function that $H^{-1} =-(\inf_{F\in\mathcal F} E_F)_-'$. This completes the proof of (b).
	
	%
	%
	
	(c) 
	For $F,G\in\mathcal M_p$, define $
	F_1=\bigvee_2\{F,G\}$ and $F_2=\bigwedge_2\{F,G\}.$ It follows from (b) that
	$F_1=1+(\pi_F\vee \pi_G)_+'$ which implies
	$\min\{F,G\}\le F_1\le \max\{F,G\}$, and
	$F_2^{-1} =-(E_F\wedge E_G)_-'$ which implies $\min\{F^{-1},G^{-1}\}\le F_2^{-1}\le\max\{F^{-1},G^{-1}\}$, and hence, $\min\{F,G\}\le F_2\le \max\{F,G\}$.
	%
	By the fact in the beginning of the proof, we have  $F_1,F_2\in\mathcal M_p$, and thus $(\mathcal M_p,\preceq_2)$ is a lattice for $p\in(1,\infty]$.
	
	Below, we give a counterexample to illustrate that $(\mathcal M_p,\preceq_2)$  is not complete lattice for $p\in(1,\infty]$.
	For $p\in(1,\infty)$, 
	define $F(x)=(-x)^{-p}$ for $x\le -1$. We have $F\not\in \M_p$ and for $y< -1,$ let $F_y$ be a cdf with
	integrated survival function
	$$
	\pi_{F_y}(x)=\max\left\{\left(-x-\frac p{p-1}\right)_+,\,\, \pi_F'(y)(x-y)+\pi_F(y)\right\}.
	$$
	It is clear that $F_y\in\mathcal M_\infty$ for all $y<-1$ and the set $\{F_y\}_{y<-1}$ is bounded from above as $F_y\preceq_2 \delta_{-1} $ for $y<-1$.
	Noting that $\sup_{y<-1}\pi_{F_y}=\pi_F$ and $F\not\in\mathcal M_p$, we have that $(\mathcal M_p\preceq_2)$ is not a complete lattice.
 \end{proof}


\section{Proofs for other results in Sections \ref{sec:4}}
\label{app:A}

\noindent{\bf Proof of Proposition \ref{prop-EMAM}.}~
For a fixed $x\in\R$, both $F\mapsto F(x)$ and $F\mapsto \pi_F(x)$ are affine on $\mathcal M_1$. Hence, for
	$F\in{\rm conv}\mathcal F$ with $F=\sum_{i=1}^n\lambda_i F_i$ where $(\lambda_1,\dots,\lambda_n)\in\Delta_n$ and $F_i\in\mathcal F$ for $i\in[n]$, there exist $G_1,G_2\in\{F_1,\dots,F_n\}\subseteq\mathcal F$ such that $G_1(x)\le F(x)$ and $\pi_{G_2}(x)\ge \pi_F(x)$.
	The results follow immediately from Proposition \ref{th-sharp}. \qed

\noindent{\bf Proof of Theorem \ref{prop-cxES}.}~
	(a) Since $\mathfrak{m}(F)=\lim_{x\to-\infty}\left\{x+\pi_F(x)\right\}$ for each $F\in \M_1$, we have
	\begin{align*}
		\mathfrak{m}^{{\rm MA}_2}(\mathcal F)  &=\mathfrak{m}\left(\bigvee_2\mathcal F\right)
		{\color{black}=\lim_{x\to-\infty}\left\{x+\pi_{\bigvee_2\mathcal F} (x) \right\}}
		=\lim_{x\to-\infty}\left\{x+\sup_{F\in\mathcal F}\pi_{F}(x)\right\}\\
		&=\lim_{x\to-\infty}\sup_{F\in\mathcal F} \left\{\mathfrak{m}(F)+\int_{\R}(x-y)_+\d F(y)\right\}
		\le\sup_{F\in\mathcal F}\mathfrak{m}(F)+\lim_{x\to-\infty}\sup_{F\in\mathcal F} \int_{\R}(x-y)_+\d F(y)\\
		&=\sup_{F\in\mathcal F}\mathfrak{m}(F)=\mathfrak{m}^{\rm WR}(\mathcal F),
	\end{align*}
	{\color{black}where the third equality comes from  (ii) of Proposition \ref{th-sharp}, and the forth equality follows from $x+ \pi_F(x) = x+ \E^F[(X-x)_+] = \mathfrak{m}(F) + \E^F[(x-X)_+]$.}
	The converse direction $\mathfrak{m}^{{\rm MA}_2}(\mathcal F)\ge\mathfrak{m}^{\rm WR}(\mathcal F)$ is trivial. Hence, we complete the proof of (a).

	(b) Suppose that $\mathcal F\subseteq \M_1$ is a convex set which is $\preceq_2$-bounded.
	Denote by $\Pi_{\mathcal G}= \sup_{F\in \mathcal G} \pi_F$ for any set $\mathcal G\subseteq \M_1$.
	If $\mathcal G$ is a convex polytope, then by Theorem 1 of \cite{ZF09}, we have
	\begin{align}  \ES_\alpha  ^{\rm WR}(\mathcal G) &= \ES_\alpha^{{\rm MA}_2} (\mathcal G) .\label{eq:useconvex}\end{align}
	Let  $c=   \ES_\alpha^{{\rm MA}_2} (\mathcal F )$. Using \eqref{eq:illustrateES2}, we get
	\begin{align}  x+\frac{1}{1-\alpha}\Pi_{\mathcal F }(x)\ge c\mbox{~~for all $x \in \R$.}\label{eq:useNN}\end{align}
	Take an arbitrary $G\in \mathcal F$. Since  $(\pi_{G})_+'(x)\to -1$  as $x\to -\infty$, we have $(1-\alpha) x+ \pi_G(x) \to \infty$ as $x\to -\infty$. There exists $x_0 <c$ such that  \begin{align}x+\frac{1}{1-\alpha}\pi_G(x)\ge c\mbox{~~for all $x<x_0$.}\label{eq:usetrick}\end{align}
	Fix $\epsilon  >0$. Let $\mathcal G\subseteq \mathcal F$  be a   convex polytope such that
	\begin{align}
		\Pi_{\mathcal  G }(x) \ge \Pi_{\mathcal F}(x)-\epsilon   \mbox{~~for all $x\in [x_0,c]$}.	\label{eq:useapprox}\end{align}
{\color{black}
We illustrate why such $\mathcal G$ exists. Let $\mathbb{Q}$ be a set of all rational numbers on $\R$, and we represent it as $\mathbb{Q}=\{q_i\}_{i\in\N}$. Suppose that $\{F_{i,j}\}_{j\in\N}\subseteq \mathcal F$  satisfies $\lim_{j\to\infty}\pi_{F_{i,j}}(q_i)=\Pi_{\mathcal F}(q_i)$ for $i\in\N$. Define $\mathcal G_i=\{F_{1,i},F_{2,i},\dots,F_{i,i}\}$ for $i\in\N$. It holds that $\Pi_{\mathcal G_i}(x)\to \Pi_{\mathcal F}(x)$ on $\mathbb Q$. By Theorem 10.8 of \cite{R70}, we have $\{\Pi_{\mathcal G_i}\}_{i\in\N}$ uniformly converges to $\Pi_{\mathcal F}$ on $[x_0,c]$. This implies that such $\mathcal G$ in \eqref{eq:useapprox} exists.
} 
Let $\mathcal G_0=\mathrm{conv}(\mathcal G\cup \{G\}) \subseteq\mathcal F $, which is again a convex polytope. Using \eqref{eq:useconvex}, \eqref{eq:useNN}, \eqref{eq:usetrick} and \eqref{eq:useapprox}, we obtain
	\begin{align*}
		\ES_\alpha  ^{\rm WR}(\mathcal G_0) &= \ES_\alpha^{{\rm MA}_2} (\mathcal G_0) =  \min_{x\in\R}\left\{x+\frac{\Pi_{\mathcal G_0}(x)}{1-\alpha} \right\}
		\\ &=  \min\left\{  \inf_{x < x_0} \left\{x+\frac{\Pi_{\mathcal G_0}(x)}{1-\alpha} \right\},\min_{x\in [x_0,c]} \left\{x+\frac{\Pi_{\mathcal G_0}(x) }{1-\alpha}\right\},\inf_{x >c} \left\{x+\frac{\Pi_{\mathcal G_0}(x)}{1-\alpha} \right\}\right\}
		\\ &\ge    \min\left\{  \inf_{x < x_0} \left\{x+\frac{\pi_{ G }(x)}{1-\alpha} \right\},\min_{x\in [x_0,c]} \left\{x+\frac{\Pi_{\mathcal G }(x) }{1-\alpha}\right\}, c\right\}
		\\ &\ge    \min\left\{  \min_{x\in [x_0,c]} \left\{x+\frac{\Pi_{\mathcal F}(x)  }{1-\alpha}-\frac\epsilon{1-\alpha}\right\}, c\right\}
		\\ &\ge    \min\left\{   \min_{x\in \R } \left\{x+\frac{\Pi_{\mathcal F}(x)  }{1-\alpha}\right\}, c\right\} -\frac\epsilon{1-\alpha}  =   c -\frac{\epsilon}{1-\alpha}.
	\end{align*}
	Note that $  \ES_\alpha  ^{\rm WR}(\mathcal F)  \ge   \ES_\alpha  ^{\rm WR}(\mathcal G_0) \ge c-\epsilon/(1-\alpha) $ because $ \mathcal G_0\subseteq \mathcal F$.    Since $\epsilon$ is arbitrary, we get
	$
	\ES_\alpha  ^{\rm WR}(\mathcal F)  \ge   c =  \ES_\alpha^{{\rm MA}_2} (\mathcal F ) .
	$ Together with $   \ES_\alpha  ^{\rm WR}(\mathcal F)  \le \ES_\alpha ^{{\rm MA}_2} (\mathcal F ) $, we obtain the desired equality $  \ES_\alpha  ^{\rm WR}(\mathcal F) =\ES_\alpha^{{\rm MA}_2} (\mathcal F ) $.
	
	(c) It follows directly from Proposition \ref{th-sharp}.\qed

{\color{black}
\noindent{\bf Proof of Theorem \ref{th-cxca}.}~
(a) 
We first consider the convexity of $\widetilde\rho^{\rm WR}$.
Suppose that $\widetilde\rho$ is convex.
Note that $\widetilde\rho^{\rm WR}$  is the supremum of a family of functionals $\widetilde\rho^Q$, $Q\in\mathcal Q$.
It suffices to verify that $\widetilde \rho^Q$ is convex on $L$ for all $Q\in\mathcal Q$. Let $Q\in\mathcal Q$ and $X_1,X_2\in L$. By a version of Skorhod’s
Theorem (see, e.g., Theorem 3.1 of \cite{BPR07}), we can construct {\color{black} measurable mappings} $X_1',X_2' $ on {\color{black} $(\Omega,\mathcal B)$} such that $\p(X_1'\le x_1,X_2'\le x_2)=Q(X_1\le x_1,X_2\le x_2)$ for all $x_1,x_2\in\R$, i.e., the joint cdf of $(X_1,X_2)$
under $Q$ is same as the joint cdf of $(X_1',X_2')$ under $\p$. Hence, we have {\color{black}$X_1',X_2' \in L^1$, and }
\begin{align*}
	\widetilde \rho^Q(\lambda X_1+(1-\lambda)X_2)
	=\widetilde \rho(\lambda X_1'+(1-\lambda)X_2')
	\le \lambda\widetilde\rho({X_1'})+(1-\lambda)\widetilde\rho({X_2'})
	=\lambda\widetilde\rho^Q(X_1)+(1-\lambda)\widetilde\rho^Q(X_2),
\end{align*}
where the inequality follows from the convexity of $\rho$. This yields the convexity of $\widetilde{\rho}^{\rm WR}$.


To see the convexity of $\widetilde\rho^{{\rm MA}_2}$, by Theorem 2.2 of \cite{KR09}, $\widetilde\rho: L^1\to\R$ is continuous with respect to the $L^1$-norm because $\widetilde\rho(X)\in\R$ for all $X\in L^1$. It then follows from Theorem 5.1 of \cite{CMMM11} that $\rho$ is $\preceq_2$-consistent.
{\color{black}For any $F\in \mathcal M_1$, denote $g_F(\alpha)=\ES_\alpha(F)$ for $\alpha\in(0,1)$, and define 
	$$
	\mathcal G=\{g_F: (0,1)\to\R\mid F\in\mathcal M_1\}~~~{\rm and}~~~\widehat\rho :\mathcal G\to \R ~~{\rm as}~~ \widehat\rho(g_F)={\rho}(F),~~ g_F\in \mathcal G.
	$$
	We assert that $\mathcal G$ is a convex set and $ \widehat\rho$ is convex in $\mathcal{G}$. To see it, take $F,G\in \mathcal M_1$ and $\lambda\in [0,1]$. Define $H$ as a distribution whose quantile is $H^{-1}(\alpha) = \lambda F^{-1}(\alpha) + (1-\lambda) G^{-1}(\alpha)$, $\alpha\in [0,1]$. One can verify that $H\in\mathcal M_1$ and $g_H= \lambda g_F+ (1-\lambda) g_G$, and thus, $\lambda g_F+ (1-\lambda) g_G\in \mathcal G$, which implies that $\mathcal G$ is a convex set. 
	Let $U\in L^1$ be a uniform random variable on $[0,1]$ under $\p$, i.e., $\p(U\le x)=x$ for $x\in[0,1]$, and such random variable exists because $(\Omega,\mathcal F,\p)$ is a nonatomic space (see e.g., Lemma A.27 of \cite{FS16}). It holds that $\lambda F^{-1}(U)+(1-\lambda)G^{-1}(U)\in L^1$ has the distribution $H$. Therefore, we have
	\begin{align*}
		\widehat\rho(\lambda g_F+ (1-\lambda) g_G)&=\widehat\rho(g_H)=\rho(H) \\
		&= \widetilde{\rho}(\lambda F^{-1}(U)+(1-\lambda)G^{-1}(U)) 
		\le
		\lambda \widetilde \rho(F^{-1}(U))+(1-\lambda) \rho(G^{-1}(U))\\
		&=
		\lambda \rho(F)+(1-\lambda) \rho(G) =\lambda\widehat\rho( g_F)+ (1-\lambda)\widehat\rho(g_G),
	\end{align*}
	where we have used the convexity of $\widetilde\rho$ in the first inequality.
	This implies that $ \widehat\rho$ is convex in $\mathcal{G}$.
	We assert that 
	\begin{align}\label{R2-1}
		\widetilde\rho^{{\rm MA}_2}(X)=\inf_{g\in\mathcal G}\left\{\widehat\rho(g)+\Theta(X,g)\right\},
	\end{align}
	where $\Theta(X,g)=0$ 
	if $\sup_{Q\in\mathcal Q}\widetilde{\ES}_\alpha^Q(X)\le g(\alpha)$ for all $\alpha\in(0,1)$, and $\Theta(X,g)=\infty$ otherwise. To see it, note that if ${\mathcal F_{X|\mathcal Q}}$ is not bounded, then $\widetilde\rho^{{\rm MA}_2}(X)=\infty$ and $\Theta(X,g)=\infty$ for all $g\in\mathcal G$, which imply \eqref{R2-1} holds. If ${\mathcal F_{X|\mathcal Q}}$ is bounded, then  note that $F\preceq_2 G$ if and only if $g_F(\alpha)\le g_G(\alpha)$ for all $\alpha\in(0,1)$. Since $\rho$ is $\preceq_2$-consistent, we have
	\begin{align*}
		\widetilde\rho^{{\rm MA}_2}(X)=\rho\left(\bigvee_2{\mathcal F_{X|\mathcal Q}}\right)&=\inf\left\{\rho(F): F\in\mathcal M_1,~\sup_{Q\in\mathcal Q}\widetilde{\ES}_\alpha^Q(X)\le g_F(\alpha),~~\forall \alpha\in(0,1)\right\}\\
		&=\inf\left\{\widehat\rho(g): g\in\mathcal G,~  \sup_{Q\in\mathcal Q}\widetilde{\ES}_\alpha^Q(X)\le g(\alpha),~~\forall \alpha\in(0,1)\right\}\\
		&=\inf_{g\in\mathcal G}\left\{\widehat\rho(g)+\Theta(X,g)\right\}.
	\end{align*}
	Therefore, \eqref{R2-1} holds for all $X\in L$. 
	It remains to show that $\widehat\rho(g)+\Theta(X,g)$ is convex in $(X,g)\in L\times\mathcal G$. 
	%
	Since $ \widehat\rho$ is convex in $\mathcal{G}$, 
	it remains to verify that $\Theta(X,g)$ is convex on $L\times\mathcal G$. To see this, by definition of $\Theta(X,g)$, it suffices to show that the set $\{(X,g)\in L\times \mathcal G: \sup_{Q\in\mathcal Q}\widetilde{\ES}_\alpha^Q(X)\le g(\alpha),~\alpha\in [0,1]\}$ is a convex set.
	Take $X_1,X_2\in L$ and $g_1,g_2\in\mathcal G$. 
	For $\alpha\in [0,1],$ it holds that
	\begin{align*}
		\sup_{Q\in\mathcal Q}\widetilde{\ES}_\alpha^Q(\lambda X_1+(1-\lambda) X_2)	&\le\sup_{Q\in\mathcal Q}\left\{\lambda \widetilde{\ES}_\alpha^Q(X_1)+(1-\lambda)\widetilde{\ES}_\alpha^Q( X_2)\right\}\\
		&\le \lambda\sup_{Q\in\mathcal Q} \widetilde{\ES}_\alpha^Q(X_1)+(1-\lambda)\sup_{Q\in\mathcal Q} \widetilde{\ES}_\alpha^Q( X_2)\le \lambda g_1+(1-\lambda)g_2,
	\end{align*}
	where the first inequality follows from the convexity of ES.
	Hence, we have that $\widetilde\rho^{{\rm MA}_2}$ is convex.}

(b) Define $f_F(\alpha)=\VaR_\alpha(F)$ for $\alpha\in(0,1)$ and $F\in\mathcal M_1$.
Let $\mathcal M_1^{-1}=\{f_F:(0,1)\to\R| F\in\mathcal M_1\}$.
Define $\widehat\rho:\mathcal M_1^{-1}\to\R$ as the risk measure satisfying $\widehat\rho(f_F)=\rho(F)$. Noting that $\widetilde\rho$ satisfies comonotonic additivity, we have $\widehat\rho(f_1+f_2)=\widehat\rho(f_1)+\widehat\rho(f_2)$ for all $f_1,f_2\in\mathcal M_1^{-1}$. Suppose now
$X_1,X_2\in L$ are two comonotonic random variables. 
By Proposition 4.6 of \cite{WZ18} who prove that the mapping $X\mapsto \sup_{Q\in\mathcal Q}\widetilde\VaR_\alpha^{Q}(X)$ satisfies comonotonic additivity for all $\alpha\in(0,1)$, we have 
\begin{align*}
	\sup_{Q\in\mathcal Q} f_{F_{X_1+X_2}^Q}(\alpha)
	&=\sup_{Q\in\mathcal Q}\widetilde{\VaR}_{\alpha}^Q(X_1+X_2)\\
	&=\sup_{Q\in\mathcal Q}\widetilde{\VaR}_{\alpha}^Q(X_1)+\sup_{Q\in\mathcal Q}\widetilde{\VaR}_{\alpha}^Q(X_2)
	=\sup_{Q\in\mathcal Q} f_{F_{X_1}^Q}(\alpha)+\sup_{Q\in\mathcal Q} f_{F_{X_2}^Q}(\alpha),~~\forall \alpha\in(0,1).
\end{align*}
Hence,
\begin{align*}
	\widetilde\rho^{{\rm MA}_1}(X_1+X_2)&=\rho\left(\bigvee_1{\mathcal F_{X_1+X_2|\mathcal Q}}\right)=\widehat\rho\left(\sup_{Q\in\mathcal Q}f_{F_{X_1+X_2}^Q}\right)\\
	&=\widehat\rho\left(\sup_{Q\in\mathcal Q}f_{F_{X_1}^Q}+\sup_{Q\in\mathcal Q}f_{F_{X_2}^Q}\right)
	=\widehat\rho\left(\sup_{Q\in\mathcal Q}f_{F_{X_1}^Q}\right)+\widehat\rho\left(\sup_{Q\in\mathcal Q}f_{F_{X_2}^Q}\right)\\
	&=\widetilde\rho^{{\rm MA}_1}(X_1)+\widetilde\rho^{{\rm MA}_1}(X_2).
\end{align*}
This yields the comonotonic additivity of $\widetilde{\rho}^{{\rm MA}_1}$.

(c) We only consider the case of translation invariance as the case of positive homogeneity is similar. Suppose that $\widetilde\rho$ satisfies translation invariance. We have $\rho(G)=\rho(F)+c$ whenever $F,G\in\mathcal M_1$ satisfy $G(x)=F(x-c)$ for all $x\in\R$.
The translation invariance of $\widetilde\rho^{\rm WR}$ is trivial because we have 
$$
\widetilde{\rho}^{Q}(X+c)=\rho(F_{X+c}^Q)=\rho(F_X^Q)+c=\widetilde{\rho}^{Q}(X)+c
$$ 
for any $X\in L$, $Q\in \mathcal Q$ and $c\in\R$. 
To see the case of ${\rm MA}_1$,
it follows from Proposition \ref{th-sharp} that
$$
\VaR_{\alpha}\left(\bigvee_1\mathcal F_{X+c|\mathcal Q}\right)
=\sup_{Q\in\mathcal Q}\widetilde{\VaR}_{\alpha}^Q(X+c)
=\sup_{Q\in\mathcal Q}\widetilde{\VaR}_{\alpha}^Q(X)+c
=\VaR_{\alpha}\left(\bigvee_1\mathcal F_{X|\mathcal Q}\right)+c,~~\forall \alpha\in(0,1).
$$
This means that $\bigvee_1\mathcal F_{X+c|\mathcal Q} (x)=\bigvee_1\mathcal F_{X|\mathcal Q}(x-c)$ for all $x\in\R$.
Hence, we have
$$
\widetilde{\rho}^{{\rm MA}_1}(X+c)=\rho\left(\bigvee_1\mathcal F_{X+c|\mathcal Q}\right)=\rho\left(\bigvee_1\mathcal F_{X|\mathcal Q}\right)+c=\widetilde{\rho}^{{\rm MA}_1}(X)+c,
$$
which yields translation invariance of $\widetilde{\rho}^{{\rm MA}_1}$. For ${\rm MA}_2$, using Proposition \ref{th-sharp} again, we have
$$
\pi_{\bigvee_2\mathcal F_{X+c|\mathcal Q}}(x)
=\sup_{Q\in\mathcal Q}\E^Q[(X+c-x)_+]
=\pi_{\bigvee_2\mathcal F_{X|\mathcal Q}}(x-c),~~\forall x\in\R,
$$
which implies $\bigvee_2\mathcal F_{X+c|\mathcal Q}(x)=\bigvee_2\mathcal F_{X|\mathcal Q}(x-c)$ for all $x\in\R$. Hence, we have
$$
\widetilde{\rho}^{{\rm MA}_2}(X+c)=\rho\left(\bigvee_2\mathcal F_{X+c|\mathcal Q}\right)=\rho\left(\bigvee_2\mathcal F_{X|\mathcal Q}\right)+c=\widetilde{\rho}^{{\rm MA}_2}(X)+c,
$$
which shows that $\widetilde{\rho}^{{\rm MA}_2}$ is translation invariant.  \qed
}

{\color{black}
\section{\color{black}Proofs for results in Sections \ref{sec:34} and omitted examples}\label{app:optMA}

\subsection{\color{black} Proofs} \label{app:optMA-1-1}

\noindent{\bf Proof of Proposition \ref{prop:cxMA}.}~
We introduce some notation that defined in the proof of Theorem \ref{th-cxca} (a).
Let $g_F(\alpha)=\ES_\alpha(F)$ for $\alpha\in(0,1)$ and $F\in\mathcal M_1$.
Define $\mathcal G=\{g:(0,1)\to\R| (1-\alpha)g(\alpha)~{\rm is~concave~for~}\alpha\in(0,1),~\lim_{\alpha\to1}(1-\alpha)g(\alpha)=0\}$, and $\widehat\rho:\mathcal G\to\R$ as $\widehat\rho(g_F)=\rho(F)$. As shown in the proof of Theorem \ref{th-cxca} (a), we know that $\rho$ is $\preceq_2$-consistent. Hence, we have
\begin{align*}
\rho^{{\rm MA}_2}(\mathcal F_{\mathbf a,f})&=\inf\left\{\rho(G):G\in\mathcal M_1,~ \sup_{F\in\mathcal F}\widetilde{\ES}_\alpha^F(f(\mathbf a,\mathbf X))\le g_G(\alpha),~~\forall \alpha\in(0,1)\right\}\\
	&=\inf\left\{\widehat\rho(g): g\in\mathcal G,~ \sup_{F\in \mathcal F}\widetilde{\ES}_\alpha^F(f(\mathbf a,\mathbf X))\le g(\alpha),~~\forall \alpha\in(0,1)\right\}\\
	&=\inf_{g\in\mathcal G}\left\{\widehat\rho(g)+\Theta(\mathbf a,g)\right\},
\end{align*}
where $\Theta(\mathbf a,g)=0$ if $\sup_{F\in\mathcal \mathcal F}\widetilde{\ES}_\alpha^F(f(\mathbf a,\mathbf X))\le g(\alpha)$ for all $\alpha\in(0,1)$, and $\Theta(\mathbf a,g)=\infty$ otherwise. The remained proof is similar to that of Theorem \ref{th-cxca} (a) by noting that $f(\mathbf a, \mathbf x)$ is convex in $\mathbf a$.
\qed

The following proposition shows that any law-invariant coherent risk measure on $\mathcal M_1$ 
can be approximated by risk measures  in the form of \eqref{eq:0818-3}. 
Recall that  any law-invariant coherent risk measure that is finite on $\mathcal M_1$ admits a 
Kusuoka representation  (see e.g., \cite{S13})
\begin{equation}\label{eq:kusuoka}\rho(F) = \sup_{\mu\in \mathcal P_0}~\int_0^1  {\rm ES}_{\alpha} (F) \d \mu(\alpha), ~~~~F\in \mathcal M_1,\end{equation}
where  $\mathcal P_0$ is a subset of the set $ \mathcal P$  of all probability measures on $[0,1]$.
\begin{proposition}\label{pro:A4}
	Let $\rho  :  \mathcal M_1 \to \R$ be a  law-invariant coherent risk measure. 
	There exists a sequence of risk measures  $(\rho_n)_{n\in \N}$ of the form  \eqref{eq:0818-3}  such that   $\lim_{n\to\infty}\rho_n(F) =\rho(F)$ for any $F\in \mathcal M_1$. 
	Moreover, if $\mathcal P_0$ in \eqref{eq:kusuoka} is finite, then for any set $\mathcal G$ of distributions supported within a common compact interval, uniform convergence holds: 
	$\sup_{F\in \mathcal G}|\rho_n(F) -\rho(F)|\to0$ as $n\to\infty$.
\end{proposition}

\begin{proof}
	Denote by  
	$D_n = \left\{ i/2^n: i=0,\ldots,2^n \right\}$, $n\in \N$, and $D= \bigcup_{n\in\N} D_n$. Note that $D$ is  countable  and dense in $[0,1]$. 
	For $n\in \N$, denote by $\mathcal P_n$ the finite set of convex combinations of elements of $\{\delta_j:j\in D_n\}$ with  
	weights in $D_n$, that is, 
	$$\mathcal P_n= \left\{\sum_{i=0}^{2^n} \alpha_i \delta_{i/2^n} :~ \alpha_i\in D_n, ~i=0,\ldots,2^n,~ \sum_{i=0}^{2^n} \alpha_i=1\right\}.$$
	Let $\mathcal P_0$ be given by \eqref{eq:kusuoka}, which represents $\rho$. 
	For each $\mu\in \mathcal P_0$, define $\mu_n=\sum_{i=0}^{2^n-1} \alpha_{n,i} \delta_{i/2^n}\in\mathcal P_n$, $n\in\N$, where $\alpha_{n,i} = \beta_{n,i}- \beta_{n,i-1} $ 
	and 
	$$\beta_{n,i} = 2^{-n}\left \lceil 2^n \mu \left(\left[0, \frac {i+1}{2^n} \right]\right)\right\rceil,~~~i=0,\ldots, 2^n-1;~~~\beta_{n,-1}=0. $$ 
	Here, $\lceil x \rceil$ is the smallest integer dominating $x$. 
	This construction guarantees that $\mu_n \preceq_1 \mu$ and $\mu_n$ converges to $\mu$ in distribution as $n\to\infty$. Therefore, for $F\in \mathcal M_1$, we have that $\alpha\mapsto {\rm ES}_{\alpha} (F) $ is  continuous and bounded from below, and hence,
	$$\liminf_{n\to\infty} \int_0^1  {\rm ES}_{\alpha} (F) \d \mu_n(\alpha) \ge \int_0^1  {\rm ES}_{\alpha} (F) \d \mu(\alpha) .$$
	Let $\mathbb W_n = \bigcup_{\mu\in \mathcal P_0} \{\mu_n\} $ and 
	$$
	\rho_n(F)=\sup_{\nu\in\mathbb W_n } \int_0^1  {\rm ES}_{\alpha} (F) \d \nu(\alpha),~~~F\in \mathcal M_1 .
	$$ Note that $\mathbb W_n  \subseteq \mathcal P_n$ is a finite set, and hence $\rho_n$ is a risk measure of the form \eqref{eq:0818-3}. For each $\mu\in \mathcal P_0 $,
	$$\liminf_{n\to\infty} \rho_n(F) \ge \liminf_{n\to\infty}  \int_0^1  {\rm ES}_{\alpha} (F) \d \mu_n(\alpha)  \ge \int_0^1  {\rm ES}_{\alpha} (F) \d \mu(\alpha).$$
	Therefore, we have  for $F\in \mathcal M_1$,
	$$\liminf_{n\to\infty} \rho_n(F)  \ge \sup_{\mu\in \mathcal P_0}\int_0^1  {\rm ES}_{\alpha} (F) \d \mu(\alpha) =\rho(F).$$
	On the other hand,  by $\mu_n \preceq_1 \mu$, we have  $ \rho_n(F)  \le \rho(F)$ for all $n\in\N$, and thus, $\lim_{n\to\infty}\rho_n(F) =\rho(F)$. This completes the proof of the first statement.

	Next, we show uniform convergence on $\mathcal G$, assuming that $\mathcal P_0$ is finite. Note that
	\begin{align*}
		\sup_{F\in\mathcal G}|\rho_n(F)-\rho(F)|&=\sup_{F\in\mathcal G}\left|\sup_{\mu\in \mathcal P_0}~\int_0^1  {\rm ES}_{\alpha} (F) \d \mu_n(\alpha)-\sup_{\mu\in \mathcal P_0}~\int_0^1  {\rm ES}_{\alpha} (F) \d \mu(\alpha)\right|\\
		&\le\sup_{\mu\in \mathcal P_0}\sup_{F\in\mathcal G}\left|\int_0^1  {\rm ES}_{\alpha} (F) \d \mu_n(\alpha)-\int_0^1  {\rm ES}_{\alpha} (F) \d \mu(\alpha)\right|.
	\end{align*}
	Recall that $\rho$ and $\rho_n$ for $n\in\N$ satisfy translation invariance and positive homogeneity.
	To see $\lim_{n\to\infty}\sup_{F\in \mathcal F}|\rho_n(F)-\rho(F)|=0$, it suffices to verify that for $\mu\in\mathcal P_0$
	\begin{align}\label{eq-UC}
		\lim_{n\to\infty}\sup_{F\in\mathcal M[0,1]}\left|\int_0^1  {\rm ES}_{\alpha} (F) \d \mu_n(\alpha)-\int_0^1  {\rm ES}_{\alpha} (F) \d \mu(\alpha)\right|=0,
	\end{align}
	where $\mathcal M[0,1]$ is the set of all distributions with support in $[0,1]$. Below let us prove \eqref{eq-UC}. 
	Define, for $s\in[0,1]$,
	\begin{align*}
		G(s)=\mu([0,s]), ~~G_n(s)=\mu_n([0,s]),~~
		h(s)=\int_0^s \frac{1}{1-\alpha} \d G(\alpha),~~\mbox{and}~~h_n(s)=\int_0^s \frac{1}{1-\alpha} \d G_n(\alpha).
	\end{align*}
	Applying integration by parts, we have
	\begin{align*}
		h_n(s)-h(s)&=\frac{G_n(s)}{1-s}-\int_0^s \frac{G_n(\alpha)}{(1-\alpha)^2}\d \alpha
		-\frac{G(s)}{1-s}+\int_0^s \frac{G(\alpha)}{(1-\alpha)^2}\d \alpha.
	\end{align*} Let $A$ be the set of all continuity points of $G$. 
	Since $\mu_n\to \mu$ in distribution, we have $G_n(\alpha)\to G(\alpha)$ for all $\alpha\in A$, and this also implies $\int_0^s{G_n(\alpha)}/{(1-\alpha)^2}\d \alpha\to \int_0^s {G(\alpha)}/{(1-\alpha)^2}\d \alpha$ for all $s\in[0,1)$. Hence, we have 
	\begin{align}\label{eq-UC0}
		h_n(s)-h(s)\to 0 ~~{\rm for}~s\in A.
	\end{align}
	Using the relation that $(1-\alpha){\rm ES}_{\alpha} (F)=\int_\alpha^1 \VaR_s(F)\d s$, we have
	\begin{align*}
		&~~~~\sup_{F\in\mathcal M[0,1]}\left|\int_0^1  {\rm ES}_{\alpha} (F) \d \mu_n(\alpha)-\int_0^1  {\rm ES}_{\alpha} (F) \d \mu(\alpha)\right|\\
		&=\sup_{F\in\mathcal M[0,1]}\left|\int_0^1  \frac{1}{1-\alpha}\int_\alpha^1 \VaR_s(F)\d s \d \mu_n(\alpha)-\int_0^1  \frac{1}{1-\alpha}\int_\alpha^1 \VaR_s(F)\d s \d \mu(\alpha)\right|\\
		&=\sup_{F\in\mathcal M[0,1]}\left|\int_0^1 \VaR_s(F)h_n(s)\d s
		-\int_0^1 \VaR_s(F)h(s)\d s\right|\\
		&\le \sup_{F\in\mathcal M[0,1]}\int_0^1\left|(h_n(s)-h(s))\VaR_s(F)\right|\d s
		\le \int_0^1 |h_n(s)-h(s)|\d s,
	\end{align*}
	where we have used Fubini's theorem in the second step. By \eqref{eq-UC0}, we know that $h_n\to h$ almost surely on $[0,1]$. On the other hand, it follows from Theorem 4 of \cite{HW24} that $h(1)<\infty$. Note that $\mu_n \preceq_1 \mu$ and $h_n$ is increasing on $[0,1]$. For any $s\in[0,1]$ and $n\in\N$, we have
	\begin{align*}
		0\le h_n(s)\le h_n(1)\le h(1)<\infty,
	\end{align*}
	and this implies that $\{h_n\}_{n\in\N}$ is a bounded sequence. Therefore, we conclude that $\lim_{n\to \infty}\int_0^1 |h_n(s)-h(s)|\d s=0$, and \eqref{eq-UC} holds. This completes the proof.
	\end{proof}

\noindent{\bf Proof of Proposition \ref{th-MAcr}.}~
Note that $\min_{\mathbf a\in A} \rho^{\rm MA_2} (\mathcal F_{\mathbf a,f})$ is equivalent to
\begin{align}\label{eq-MAcr1}
	\min_{\mathbf a\in A}\sup_{w\in \mathbb W} \sum_{j=1}^{n^w} p_j^w {\rm ES}_{\alpha_j^w} \left(\bigvee_2 \mathcal F_{\mathbf a,f}\right),
\end{align}
where $\mathcal F_{\mathbf a,f}=\{F_{f(\mathbf a,\mathbf X)}: F_{\mathbf X}\in\mathcal F\}$. 
It holds that
\begin{align*}
	{\rm ES}_{\alpha_j^w} \left(\bigvee_2 \mathcal F_{\mathbf a,f}\right)
	=
	\inf_{x\in\R}\left\{x+\frac{1}{1-\alpha_j^w}\pi_{\vee_2 \mathcal F_{\mathbf a,f}}(x)\right\}=\inf_{x\in\R}\left\{x+\frac{1}{1-\alpha_j^w}\sup_{F\in\mathcal F}\E^F\left[(f(\mathbf a,\mathbf X)-x)_+\right]\right\},
\end{align*}
where the first and the second steps follow from \eqref{eq:RU02} and
Proposition \ref{th-sharp}, respectively. 
Substituting the above equation into \eqref{eq-MAcr1} yields the following equivalent problem
\begin{align*}
	\min_{\mathbf a\in A, h_j^w\in\R}   &\sup_{w\in \mathbb W}~\sum_{j=1}^{n^w} p_j^w h_j^{w}\\
	{\rm s.t.} ~~ & \inf_{x\in\R} \left\{ x+ \frac1{1-\alpha_j^w} \sup_{F\in\mathcal F}\E^F[( f({\bf a},\mathbf X)-x)_+] \right\} \le h_j^w,~~  j\in[n^w],~  w\in\mathbb W.
\end{align*}
This is again equivalent to Problem \eqref{prob-MAcr} which completes the proof.\qed

\noindent{\bf Proof of \eqref{eq:0820-1}.}~
 Note that $\min_{\mathbf a\in A} \rho^{\rm WR} (\mathcal F_{\mathbf a,f} )$ is equivalent to 
	\begin{align*}
		\min_{\mathbf a\in A } ~ & \sup_{i\in[n]}\sup_{w\in \mathbb W} \sum_{j=1}^{n^w} p_j^w \widetilde{{\rm ES}}_{\alpha_j^w}^{F_i} (f({\bf a},\mathbf X)),
	\end{align*}
	that is,
	\begin{align*}
		\min_{\mathbf a\in A,\,h\in\R } ~ &  h \\
		{\rm s.t.} ~~ &  \sum_{j=1}^{n^w} p_j^w \widetilde{{\rm ES}}_{\alpha_j^w}^{F_i} (f({\bf a},\mathbf X)) \le h,~~ i\in[n],~ w\in\mathbb W.
	\end{align*}
	This is again equivalent to 
	\begin{align*}
		\min_{\mathbf a\in A, h,h_{i,j}^{w}\in\R} ~ &  h \\
		{\rm s.t.} ~~ &  \sum_{j=1}^{n^w} p_j^w h_{i,j}^{w}\le h,~~ i\in[n],~ w\in\mathbb W\\
		&\widetilde{{\rm ES}}_{\alpha_j^w}^{F_i} (f({\bf a},\mathbf X))\le    h_{i,j}^{w},~ i\in[n],~ j\in[n^w],~ w\in\mathbb W.
	\end{align*}
	Substituting the representation $\widetilde{{\rm ES}}_{\alpha}^F (f({\bf a},\mathbf X)) = \inf_{x\in\R} \{x+ \frac1{1-\alpha}\E^F [(f({\bf a},\mathbf X)-x)_+]\}$  into the above problem yields  \eqref{eq:0820-1}. \qed


\subsection{\color{black} Examples}\label{app:ex}
First, we propose the following example to demonstrate that the convexity of WR and ${\rm MA}_2$ robust optimization problems may not always coincide.
\begin{example}
Let $A\subseteq \R_+^2$ be convex. Define $f:A\times \R\to\R$ as $f(\mathbf a,x)=(a_1\vee a_2) x$, where $a_1\vee a_2=\max\{a_1,a_2\}$. Denote by $\mathcal F=\{F_1,F_2\}$ with $F_1=\delta_{-1}$ and $F_2=2\delta_{-9}/3+\delta_{9}/3$ where $\delta_t$ represents the point-mass at $t\in\R$.
For $\alpha\in[1/6,1/3)$, we consider the WR and ${\rm MA}_2$ robust optimization problems:
\begin{align}\label{prob-WRMA2}
\min_{\mathbf a\in A}\ES_\alpha^{\rm WR}(\mathcal F_{\mathbf a,f})~~~~{\rm and }~~~~\min_{\mathbf a\in A}\ES_\alpha^{{\rm MA}_2}(\mathcal F_{\mathbf a,f}).
\end{align}
For WR robust optimization, we have
\begin{align*}
\ES_\alpha^{\rm WR}(\mathcal F_{\mathbf a,f})
&=\max_{F\in\mathcal F}\widetilde{\ES}_\alpha^F(f(\mathbf a,X))
=(\ES_{\alpha}(F_1)\vee \ES_{\alpha}(F_2))(a_1\vee a_2)\\
&=\left((-1)\vee \left(\frac{6}{1-\alpha}-9\right)\right)(a_1\vee a_2).
\end{align*}
Note that $6/(1-\alpha)-9<0$ as $\alpha\in[1/6,1/3)$. We have WR robust optimization in \eqref{prob-WRMA2} is not a convex optimization problem. For ${\rm MA}_2$ robust optimization, one can check that 
\begin{align*}
\bigvee_2\mathcal F_{\mathbf a,f}
=\frac{2}{3}\delta_{-6(a_1\vee a_2)}
+\frac{1}{3}\delta_{9(a_1\vee a_2)},~~\forall \mathbf a\in A.
\end{align*}
Hence, 
\begin{align*}
\ES_\alpha^{{\rm MA}_2}(\mathcal F_{\mathbf a,f})
=\ES_\alpha\left(\bigvee_2\mathcal F_{\mathbf a,f}\right)
=\left(\frac{5}{1-\alpha}-6\right)(a_1\vee a_2).
\end{align*}
Note that $5/(1-\alpha)-6\ge 0$ as $\alpha\in[1/6,1/3)$. We have ${\rm MA}_2$ robust optimization in \eqref{prob-WRMA2} is a convex optimization problem. If we let $f(\mathbf a,x)=(a_1\wedge a_2)\cdot x$, where $a_1\wedge a_2=\min\{a_1,a_2\}$, then it follows the similar arguments previously that the WR optimization problem is convex, and the ${\rm MA}_2$ optimization problem is not convex.

\end{example}

Next, we study the tractability of the ${\rm MA}_2$ robust optimization under  two common convex uncertainty sets, namely box uncertainty and ellipsoid uncertainty, that are specifically associated with discrete cdfs.  
\begin{example}   
Suppose that the sample space is   $\{\mathbf x_{1},\dots,\mathbf x_{N}\}$, 
and consider the uncertainty set 
\begin{align*}
	\mathcal F_{\Theta}=\left\{\sum_{i=1}^N \theta_i\delta_{\mathbf x_{i}}: \bm\theta\in\Theta\right\},
\end{align*}
where $\Theta$ is a subset of $\Delta_N$. 
We have that for any fixed $\mathbf a\in A$ and $x\in\R$, 
\begin{align*}
	\sup_{F\in\mathcal F_{\Theta}}\E^F[( f({\bf a},\mathbf X)-x)_+]
	=
	\sup_{\bm\theta\in\Theta}\sum_{i=1}^N
	\theta_i (f(\mathbf a,\mathbf x_{i})-x)_+.
\end{align*}
This implies that Problem \eqref{prob-MAcr} can be equivalently reformed as the following problem with variables $\mathbf a\in A$ and $h\in\R$; $\mathbf h^w, \mathbf x^w\in\R^{n^w}$ for $w\in\mathbb W$;
$\mathbf u^{j,w}\in\R^N$ for $j\in n^w$ and $w\in\mathbb W$:
\begin{align}\label{prob-MAcrdis}
	\min ~ &  h \\
	{\rm s.t.} ~~ & (\mathbf p^w)^\top \mathbf h^w\le h,~~ w\in\mathbb W\notag\\
	&   x_j^{w}+ \frac1{1-\alpha_j^w} \sup_{\bm\theta\in\Theta} \bm\theta^\top\mathbf u^{j,w}  \le h_j^w,~~ j\in[n^w],~ w\in\mathbb W\notag\\
	& u_i^{j,w}\ge f(\mathbf a,\mathbf x_{i})-x_j^w,~~i\in[N],~j\in[n^w],~ w\in\mathbb W\notag\\
	& u_i^{j,w}\ge 0,~~i\in[N],~j\in[n^w],~ w\in\mathbb W.\notag
\end{align}
When $\Theta$ is chosen as a box uncertainty or an ellipsoid uncertainty, we can further reform the above problem as follows. 
\begin{itemize}
\item [(i)]\textbf{Box uncertainty.} 
We set $\Theta$ as a box, i.e.,
\begin{align*}
	\Theta=\Theta_{B}=\{\bm\theta: \bm\theta=\bm\theta^0+\bm\eta,~
	\mathbf e^\top\bm \eta=0,~
	\underline{\bm\eta}\le \bm\eta\le \overline{\bm\eta}\},
\end{align*}
where $\bm\theta^0$ is a nominal distribution,  $\mathbf e$ is the vector of ones, and $\underline{\bm\eta}$
and $\overline{\bm\eta}$
are given vectors. The constraints of $\bm\eta$ in the above set ensure that $\Theta_{B}\subseteq \Delta_N$. 
Using a similar argument of Section 2.2.1 in \cite{ZF09},
Problem \eqref{prob-MAcrdis} is equivalent to 
\begin{align}\label{prob-MAcrbox}
	\min ~ &  h \\
	{\rm s.t.} ~~ & (\mathbf p^w)^\top \mathbf h^w\le h,~~ w\in\mathbb W\notag\\
	&   x_j^{w}+ \frac1{1-\alpha_j^w}(\bm\theta^0)^\top\mathbf u^{j,w}+\frac1{1-\alpha_j^w}(
	\overline{\bm\eta}^{\top} \bm\xi^{j,w}+\underline{\bm\eta}^\top \bm\gamma^{j,w}) \le h_j^w,~~ j\in[n^w],~ w\in\mathbb W\notag\\
	& \mathbf e z^{j,w}+\bm\xi^{j,w}+\bm\gamma^{j,w}=0,~\bm\xi^{j,w}\ge 0,~\bm\gamma^{j,w}\le 0,~~
	j\in[n^w],~ w\in\mathbb W\notag\\
	& u_i^{j,w}\ge f(\mathbf a,\mathbf x_{i})-x_j^w,~~i\in[N],~j\in[n^w],~ w\in\mathbb W\notag\\
	& u_i^{j,w}\ge 0,~~i[N],~j\in[n^w],~ w\in\mathbb W,\notag
\end{align}
where variables are $\mathbf a\in A$ and $h\in\R$; $\mathbf h^w, \mathbf x^w\in\R^{n^w}$ for $w\in\mathbb W$;
$z^{j,w}\in\R$ and $\mathbf u^{j,w}, \bm \xi^{j,w}, \bm \gamma^{j,w}\in\R^N$  for $j\in[n^w]$ and $w\in\mathbb W$. 
If $f(\mathbf a,\mathbf x)$ is convex in terms of 
$\mathbf a$ and 
$A$ is convex, then Problem  \eqref{prob-MAcrbox} is a convex optimization problem that is consistent with Proposition \ref{prop:cxMA}. Specifically, when 
$f(\mathbf a,\mathbf x)$ is linear in $\mathbf a$ and 
$A$ is a convex polyhedron, the optimization problem simplifies to a linear program.

\item [(ii)]\textbf{Ellipsoid uncertainty.}  
We set $\Theta$ as an ellipsoid, i.e.,
\begin{align*}
	\Theta=\Theta_{E}=\{\bm\theta: \bm\theta=\bm\theta^0+C\bm\eta, ~\mathbf e^{\top}C\bm\eta=0, ~\bm\theta^0+C\bm\eta\ge 0,~ \|\bm\eta\|_2\le 1\},
\end{align*}
where $\|\bm\eta\|_2=\sqrt{\bm\eta^\top\bm\eta}$, $\bm\theta^0$ is a nominal distribution which is also the center of the ellipsoid, and $C\in\R^{N\times N}$ is the scaling matrix of the ellipsoid. The constraints of $\bm\eta$ in the above set ensure that $\Theta_{E}\subseteq \Delta_N$. 
The dual problem of
\begin{align*}
	\sup_{\bm\eta\in\R^N}\{\mathbf u^\top A\bm\eta: \mathbf e^{\top}C\bm\eta=0, ~\bm\theta^0+C\bm\eta\ge 0, ~\|\bm\eta\|_2\le 1\}
\end{align*}
is an SOCP, which has the form:
\begin{align*}
	\inf_{(\bm\xi,\bm\gamma,\zeta,z)\in\R^N\times \R^N\times \R\times \R}
	\{\zeta+(\bm\theta^0)^\top \bm\gamma:
	-\bm\xi-C^\top\bm\gamma+C^\top\mathbf e z=C^\top \mathbf u,~\|\bm\xi\|_2\le\zeta,~ \bm\gamma\ge 0\}.
\end{align*}
The strong duality between above two problems holds under some mild condition such as the Slater's condition. In this case, using a similar argument of Section 2.2.2 in \cite{ZF09},
Problem \eqref{prob-MAcrdis} can be equivalently reformed as
\begin{align}\label{prob-MAcrell}
	\min ~ &  h \\
	{\rm s.t.} ~~ & (\mathbf p^w)^\top \mathbf h^w\le h,~~ w\in\mathbb W\notag\\
	&   x_j^{w}+ \frac1{1-\alpha_j^w}(\bm\theta^0)^\top\mathbf u^{j,w}+\frac1{1-\alpha_j^w}(
	\zeta^{j,w}+(\bm\theta^0)^\top\bm\gamma^{j,w}) \le h_j^w,~~ j\in[n^w],~ w\in\mathbb W\notag\\
	& -\bm\xi^{j,w}-C^\top \bm\gamma^{j,w}+C^\top \mathbf e z^{j,w} =C^\top \mathbf u^{j,w},~\|\bm\xi^{j,w}\|_2\ge \zeta^{j,w},~\bm\gamma^{j,w}\le 0,~~
	j\in[n^w],~ w\in\mathbb W\notag\\
	& u_i^{j,w}\ge f(\mathbf a,\mathbf x_{i})-x_j^w,~~i\in[N],~j\in[n^w],n^w,~ w\in\mathbb W\notag\\
	& u_i^{j,w}\ge 0,~~i\in[N],~j\in[n^w],~ w\in\mathbb W,\notag
\end{align}
where variables are $\mathbf a\in A$ and $h\in\R$; $\mathbf h^w, \mathbf x^w\in\R^{n^w}$ for $w\in\mathbb W$;
$z^{j,w}, \zeta^{j,w}\in\R$ and $\mathbf u^{j,w}, \bm \xi^{j,w}, \bm \gamma^{j,w}\in\R^N$  for $j\in[n^w]$ and $w\in\mathbb W$. 
If $f(\mathbf a,\mathbf x)$ is convex in $\mathbf a$ and $A$ is a convex set, then Problem \eqref{prob-MAcrell} is a convex optimization problem that is consistent with Proposition \ref{prop:cxMA}. Additionally, if $f(\mathbf a,\mathbf x)$ is linear in $\mathbf a$ and $A$ is a convex polyhedron, the problem simplifies to an SOCP, which can be efficiently solved using interior-point methods.
\end{itemize}
\end{example}

To end this section, we compare the tractability of   \eqref{eq:0820-1} and \eqref{eq:0820-2}.  Set $$f({\bf a}, {\bf x}) = \sum_{i=1}^d [ \beta_i(a_i-x_i)_+ +\eta_i (x_i-a_i)_+],$$ where $\beta_i,\eta_i\in\R_+$. This is the loss function of the newsvendor problem.\footnote{See \cite{FLZ21} for a related distributionally robust newsvendor optimization based on ambiguity set constructed based on $\preceq_1$.}  For $i\in[n]$, we take $F_i$ as the empirical distribution of $N$ sample points from  the distribution ${\rm N}({\boldsymbol \mu}_i,\Sigma_i)$; this can be seen as an approximation of the corresponding normal distributions via Monte Carlo simulation. We generate the parameters by 
$
({\boldsymbol
	\mu}_i ,\boldsymbol \sigma_i) \stackrel{{\rm i.i.d.}}{\sim}  {\rm N}_{2d}(0,I_{2d}) ~\mbox{and}~\Sigma_i=\mathrm{diag}(\exp( \boldsymbol \sigma_{i})) \mbox{~for  $i\in[n]$},
$
where $I_d$ is the $d\times d$ identity matrix and $\exp$ is applied component-wise.  
Figure \ref{fig-calculationtime} presents the computation times of two approaches where we use the \texttt{cvxpy} library with \texttt{MOSEK} as the chosen solver. 
The corresponding parameters are summarized below. We set $\beta_i=i$ and $\eta_i=d-i+1$ for $i\in[d]$. 
We assume that the values of $n^w$ are identical for all $w\in\mathbb W$, and $p_j^w$, $j\in[n^w]$ are uniformly generated from the standard $n^w$-simplex, and $\alpha_j^w=(2j-1)/(2n^w)$ for $j\in[n^w]$. Each panel in Figure \ref{fig-calculationtime} illustrates the performance as one parameter varies, while the other parameters are held constant at  $(n^w,d,n,|\mathbb{W}|,N)=(10,3,3,3,100)$. 
As shown in Figure \ref{fig-calculationtime}, 
{the computation times of the two approaches appear to be similar for most values of $d$, $n$ and $N$.
	In the case of large $n^w$ or large $|\mathbb W|$, the MA approach is visibly faster. This is consistent with our intuition that MA may work better when the risk measure itself is more complicated. }

\begin{figure}[t]
\caption{Computation times of MA and WR methods 
} \label{fig-calculationtime}
\medskip	
\begin{center}
	\includegraphics[width=5.4 cm]{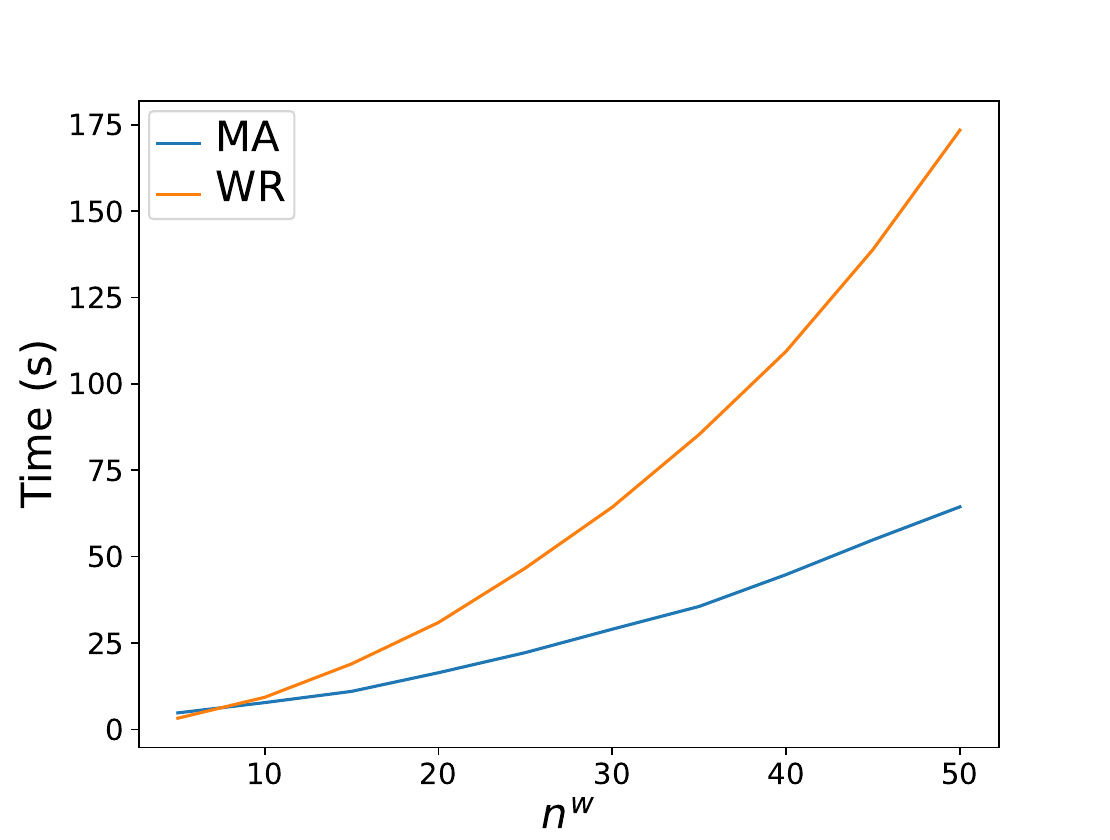}
	\includegraphics[width=5.4 cm]{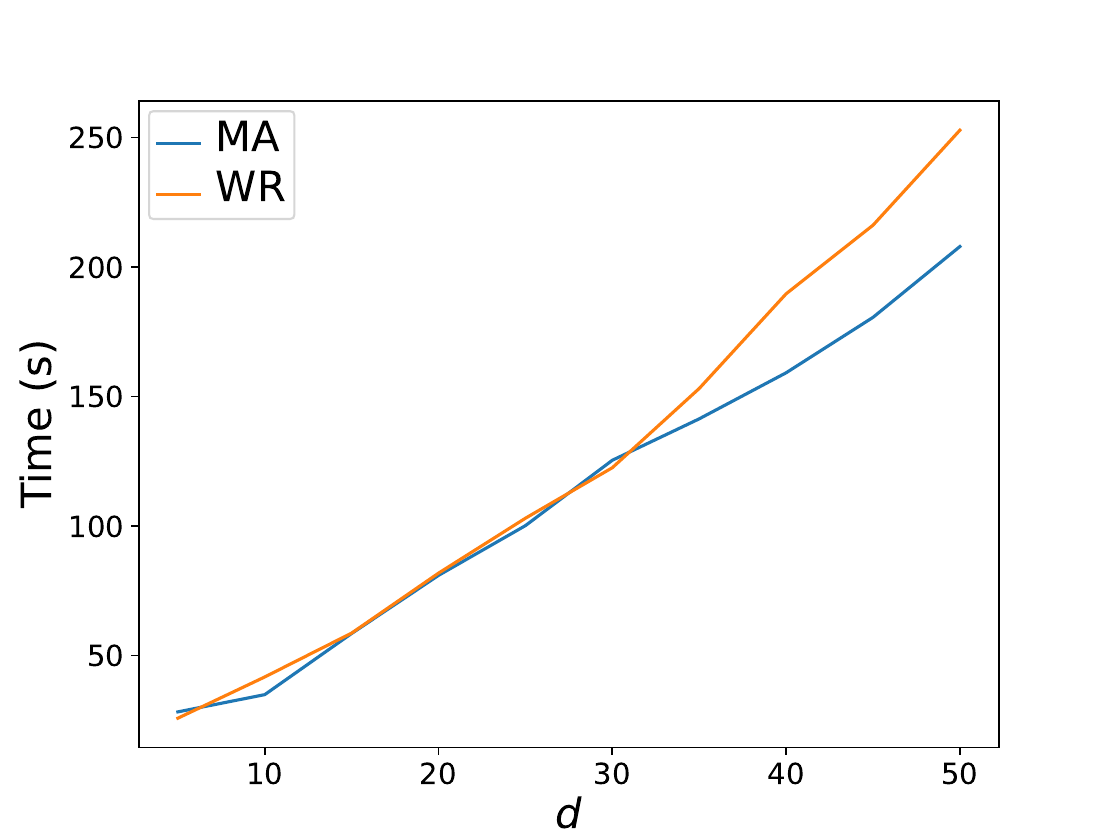}
	\includegraphics[width=5.4 cm]{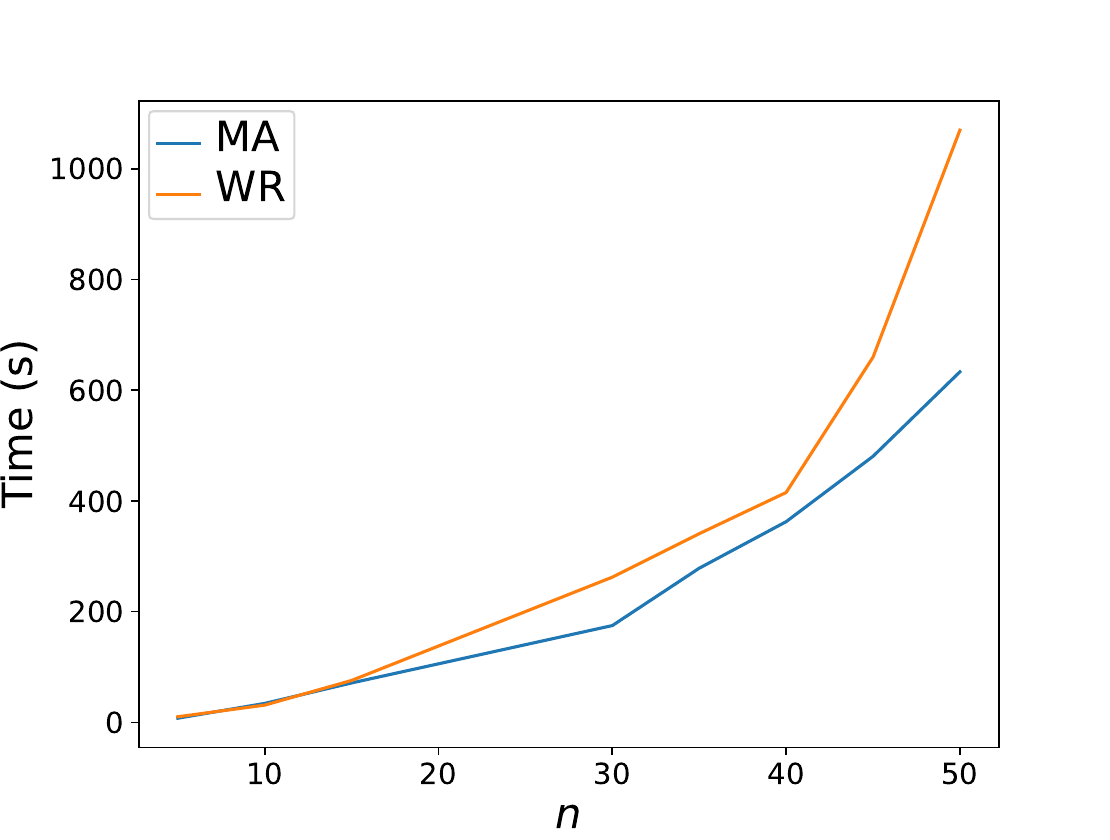}
	
	\includegraphics[width=5.4 cm]{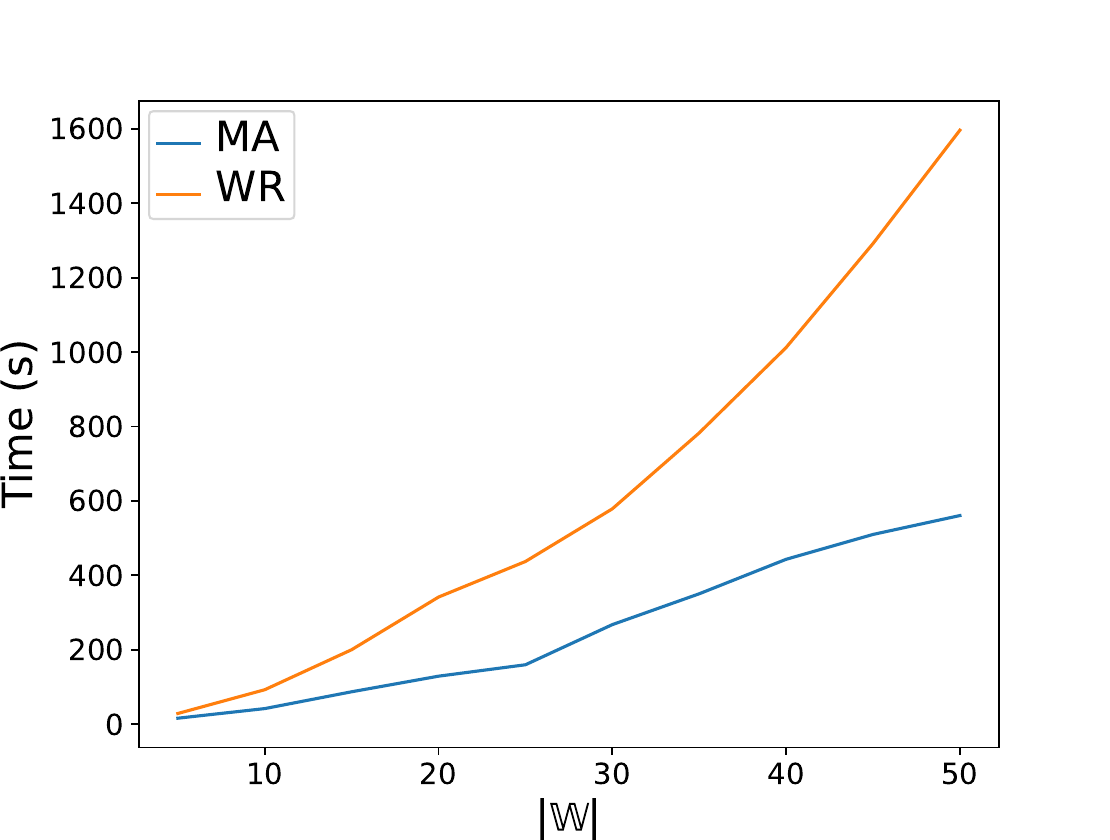}
	\includegraphics[width=5.4 cm]{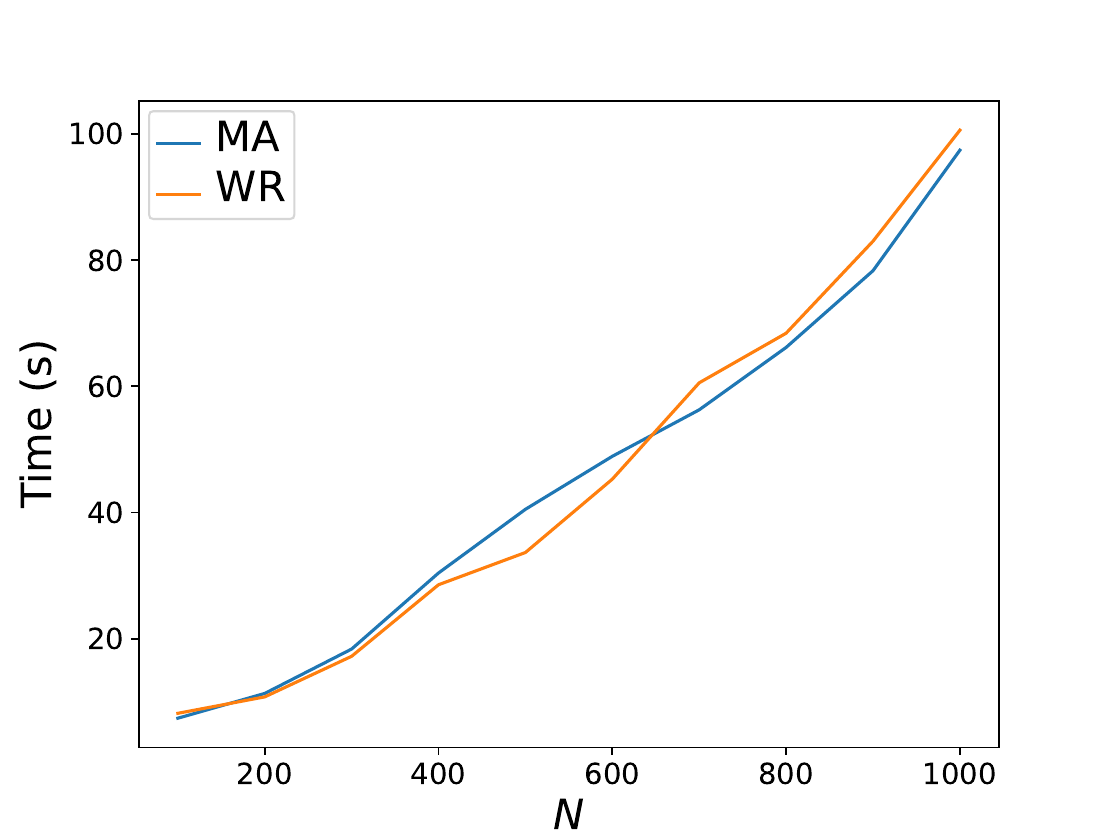}
	\medskip
\end{center}
\footnotesize{Note: $d$ is the dimension of the random vector; $n$ is the number of cdfs in $\mathcal F$; 
	$|\mathbb{W}|$ is the cardinality of $\mathbb{W}$; 
	$N$ is the sample size drawn from each $F_i\in \mathcal F$. We assume that $n^w$ is identical for each $w\in\mathbb{W}$. 
}
\end{figure}

}

\section{Proofs for results in Section \ref{sec:US} and omitted figures}
\label{app:E}
\subsection{Proofs}

\noindent{\bf Proof of Theorem \ref{prop-WU}.}~
Statement (a) can be directly obtained by applying Proposition 4 (i) of \cite{LMWW21}. To see (b), if $p=1$, one can check that $\sup_{F\in \mathcal F_{1,\epsilon}(F_0)}\pi_F(x)=\pi_{F_0}(x)+\epsilon$ for all $x\in\R$. There does not exist $G\in\mathcal M_1$ such that $\pi_{G}=\sup_{F\in \mathcal F_{1,\epsilon}(F_0)}\pi_F$ since $\lim_{x\to\infty}\sup_{F\in \mathcal F_{1,\epsilon}(F_0)}\pi_F(x)=\epsilon>0$. Hence, the set $\mathcal F_{1,\epsilon}(F_0)$ is not $\preceq_2$-bounded.
For $p>1$,
Since the Wasserstein ball is convex, it follows from Theorem \ref{prop-cxES} that $\sup_{F\in \mathcal F_{p,\epsilon}(F_0)} \ES_{\alpha}(F)=\ES_\alpha (F_{p,\epsilon|F_0}^2 )$.
By Proposition 4 (ii) of \cite{LMWW21}, we have
$\sup_{F\in \mathcal F_{p,\epsilon}(F_0)} \ES_{\alpha}(F)
=\ES_{\alpha}(F_0)+(1-\alpha)^{-1/p}\epsilon$ for $\alpha\in(0,1)$.
Therefore, one can obtain
$$
\int_{\alpha}^1 \VaR_s\left(F_{p,\epsilon|F_0}^2\right)\d s
=\int_{\alpha}^1 \VaR_s(F_0)\d s+(1-\alpha)^{1-\frac1p}\epsilon,~~\alpha\in(0,1).
$$
Take the derivative on the left and right sides of the above formula for $\alpha$, we have
$$
\VaR_\alpha\left(F_{p,\epsilon|F_0}^2\right)=
\VaR_\alpha(F_0)+\left(1-\frac1p\right)(1-\alpha)^{-\frac1p}\epsilon.
$$
Hence, we complete the proof.\qed

The following remark is related to the robust distributions in Theorem \ref{prop-WU}.

\begin{remark}\label{rem:omit-WD}
	In this remark, we collect some observations related to the robust distributions $F_{p,\epsilon|F_0}^1$ and $ F_{p,\epsilon|F_0}^2$ obtained in Theorem \ref{prop-WU}.
	\begin{enumerate}[(i)]
		\item 
		The order   $F_{p,\epsilon|F_0}^2\preceq_1F_{p,\epsilon|F_0}^1$ holds   since $(F_{p,\epsilon|F_0}^1)^{-1}(\alpha)\ge (F_{p,\epsilon|F_0}^2)^{-1}(\alpha)$ for all $\alpha\in(0,1)$.
		\item Both $F_{p,\epsilon|F_0}^1$ and $F_{p,\epsilon|F_0}^2$ are increasing in $\epsilon$ with respect to $\preceq_1$. 
		
		\item  The left-hand side of equation \eqref{eq-WU1} is increasing in $p$. Hence, a larger value of $p$ leads to a smaller cdf $F_{p,\epsilon|F_0}^1$ with respect to $\preceq_1$.
		\item  The left quantile functions $(F_{p,\epsilon|F_0}^1)^{-1}$ and $(F_{p,\epsilon|F_0}^2)^{-1}$ has the same limit  $F_0^{-1}+\epsilon$ as $p\to\infty$.
		\item
		None of $F_{p,\epsilon|F_0}^1$ and $F_{p,\epsilon|F_0}^2$ is in any Wasserstein ball of the form \eqref{eq-WU} since
		$
		{W_p}(F_{p,\epsilon|F_0}^1,F_0) =   {W_p}(F_{p,\epsilon|F_0}^2,F_0)  =\infty.
		$   This is not surprising, as $F_{p,\epsilon|F_0}^1$ and $F_{p,\epsilon|F_0}^2$ dominate every element in the Wasserstein ball and their quantile functions are  of a different shape in general.
	\end{enumerate}
\end{remark}

\noindent{\bf Proof of Theorem \ref{th:7}.}~
For two random vectors $\mathbf X$ and $\mathbf Y$ of the same dimension,  define $\mathcal L_{a,p}$ as
$$
\mathcal L_{a,p} (\mathbf X, \mathbf Y)^p
=  \E\left[  \Vert\mathbf  X - \mathbf Y \Vert^p_a\right] .
$$
For any $F \in  \mathcal F _{\mathbf w,a, p,\epsilon }(F_{\mathbf X}) $, by definition, there exists $\mathbf Z$ with $F_{\mathbf Z} \in \mathcal F^d _{ a, p,\epsilon }(F_{\mathbf X})
$ and
$ F=F_{ \mathbf w^{\top}   \mathbf Z}   $.
We can verify that
\begin{align*}
	W_p ( F_{\mathbf w^{\top} \mathbf X},   F_{\mathbf w^{\top} \mathbf Z})
	& = \inf_{ X \laweq \mathbf w^{\top} \mathbf X ,~ Z \laweq \mathbf w^{\top} \mathbf Z}  (\E[|X-Z|^p])^{1/p}
	\\ &
	= \inf_{\mathbf w ^\top  \mathbf X'\laweq \mathbf w^{\top} \mathbf X ,~\mathbf w ^\top  \mathbf Z'\laweq \mathbf w^{\top} \mathbf Z}   (\E[|\mathbf w ^\top  \mathbf X'-\mathbf w ^\top  \mathbf Z'|^p])^{1/p}
	\\ &
	\le \inf_{\mathbf w ^\top  \mathbf X'\laweq \mathbf w^{\top} \mathbf X ,~\mathbf w ^\top  \mathbf Z'\laweq \mathbf w^{\top} \mathbf Z}   \Vert \mathbf w \Vert_b   \mathcal L_{a,p} (  \mathbf X',    \mathbf Z')
	\\ & \le \inf_{ \mathbf X'\laweq  \mathbf X ,~   \mathbf Z'\laweq  \mathbf Z}  \Vert \mathbf w \Vert_b  \mathcal L_{a,p} (  \mathbf X',    \mathbf Z')
	=  \Vert \mathbf w \Vert_b      W_{a,p}^d  ( F_{ \mathbf X},  F_{\mathbf Z})  \le  \epsilon\Vert \mathbf w \Vert_b,
\end{align*}
where the infima are taken over $(X,Z)$ or $(\mathbf X',\mathbf Z')$, $b$ satisfies $1/a+1/b=1$,
and the first inequality follows from the H\"{o}lder inequality.
Hence, $\mathcal F _{\mathbf w,a, p,\epsilon }(F_{\mathbf X})\subseteq \mathcal F_{p,\epsilon \Vert \mathbf w \Vert_b    }(F_{\mathbf w^\top \mathbf X})$.

We next show the opposite direction of the set inclusion.
For any $F \in \mathcal F_{p,\epsilon \Vert \mathbf w \Vert_b    }(F_{\mathbf w^\top \mathbf X})$, since the set
$\{Y: (\E[|Y-\mathbf w^{\top}\mathbf X|^p])^{1/p}\le \epsilon\Vert \mathbf w \Vert_b  \}$ is closed,
there exists  $Z$ such that $F_Z=F$ and $(\E[|Z-\mathbf w^{\top}\mathbf X|^p])^{1/p}\le \epsilon\Vert \mathbf w \Vert_b.$ 
Below we consider the case that $a>1$.
Denote by 
$Y=Z-\mathbf w ^\top \mathbf X$ and $\mathbf Z = \mathbf X+({\mathbf w^{b/a}Y})/({\mathbf w^\top \mathbf w^{b/a}}),$ where $\mathbf w^{b/a}=({\rm sign}(w_1)|w_1|^{b/a},\dots,{\rm sign}(w_d)|w_d|^{b/a})$ and sign$:\R\to\{-1,1\}$ is the sign function.  We have
$\E[|Y|^p]\le \Vert \mathbf w \Vert_b^p  \epsilon^p$.
Moreover, noting that $\mathbf w^{\top}\mathbf w^{b/a}=\sum_{i=1}^d |w_i|^{1+b/a}=\sum_{i=1}^d |w_i|^{b}=\|\mathbf w\|_b^b$, we have
\begin{align*}
	\left(W_{a,p}^d (F_{\mathbf X},F_\mathbf {Z})\right)^p \le \E[ \Vert \mathbf Z-\mathbf X \Vert ^p_a  ]
	&= \E\left[\left\Vert \frac{\mathbf w^{b/a}}{\mathbf w^\top \mathbf w^{b/a}}{Y}\right\Vert ^p_a \right] \\[10pt]
	&= \frac {\left\Vert  \mathbf w^{b/a} \right\Vert ^p_a \E[|Y|^p]}{\|\mathbf w\|_b^{pb}}
	= \frac {\|\mathbf w\|_b^{pb/a} \E[|Y|^p]}{\|\mathbf w\|_b^{pb}}
	=\frac{\E[| Y | ^p ]}{\Vert \mathbf w\Vert ^p_b }
	\le \epsilon^p.
\end{align*}
where the forth equality is due to $\left\Vert  \mathbf w^{b/a} \right\Vert _a= \left(\sum_{i=1}^d |w_i|^{b}\right)^{1/a}=\|\mathbf w\|_b^{b/a}$ and the last equality comes from $b(1-1/a)=1$. Hence, $F_{\mathbf Z} \in  \mathcal F _{a, p,\epsilon }^d(F_{\mathbf X})$,
which implies $F_{\mathbf w^\top\mathbf Z } \in \mathcal F _{\mathbf w,a, p,\epsilon }(F_{\mathbf X}).$ Noting that $\mathbf w^\top\mathbf Z = \mathbf w^\top\mathbf X +  Y =Z$,
we obtain $F\in    \mathcal F _{\mathbf w, a,p,\epsilon }(F_{\mathbf X})$.
This implies  $\mathcal F _{\mathbf w, a,p,\epsilon }(F_{\mathbf X}) \supseteq  \mathcal F_{p,\Vert \mathbf w \Vert_b  \epsilon   }(F_{\mathbf w^\top \mathbf X})$. Hence, we conclude that $\mathcal F _{\mathbf w, a,p,\epsilon }(F_{\mathbf X}) =  \mathcal F_{p,\Vert \mathbf w \Vert_b  \epsilon   }(F_{\mathbf w^\top \mathbf X})$ for $a>1$.
If $a=1$, then we let $i_0\in[d]$ be such that $|w_{i_0}|=\|\mathbf w\|_\infty$, and
denote by $Y=Z-\mathbf w^{\top}\mathbf X$ and $\mathbf Z=\mathbf X+\mathbf {\rm sign}(w_{i_0}) \mathbf e_{i_0} Y/\|\mathbf w\|_\infty$, where $\mathbf e_{i_0}$ is a vector with its $i_0$-th element equal to 1 and all other elements equal to $0$. Similar to the previous proof of the case that $a>1$, the desired result holds.
\qed

\subsection{Omitted figures from Section \ref{sec:US}}\label{app:G}
We present a few figures omitted from Section \ref{sec:US}.
Figures \ref{quantile_W} and \ref{ES_power_W} are related to Section \ref{sec:62} for  the Wasserstein uncertainty set  $\mathcal F_{k,\epsilon}(F_0)$ with $k=2$, $\epsilon=0.1$,
where the benchmark distribution $F_0$ is the standard normal distribution function.
Figure \ref{quantile_W} shows the left quantile functions
of the supremum.
In Figure \ref{ES_power_W}, we obtain  the WR and  MA robust risk values and the risk measure is chosen as $\ES_\alpha$
or ${\rm PD}_k$.

Figure \ref{fig-R} shows the curves of robust risk evaluation via the WR and MA approaches for the mean-variance uncertainty set $\mathcal F_{0,1}$ in Section \ref{sec:MV}. The risk measure is chosen as $\ES_\alpha$, ${\rm RVaR}_{\alpha,\beta}$, ${\rm PD}_k$ or ${\rm ex}_\alpha$.

\begin{figure}[t]
	\caption{The supremum of $\mathcal F_{2,0.1}(F_0)$ with $F_0\sim  N(0,1)$.}
	\begin{center}
		\includegraphics[width=7.5cm]{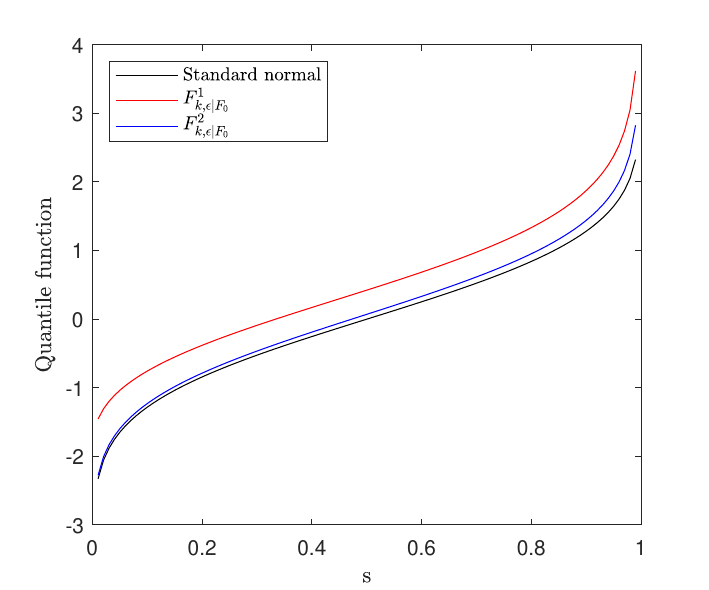}
		\label{quantile_W}
		\medskip
	\end{center}
\end{figure}
\begin{figure}[t]
	\caption{The WR and MA approaches under $\mathcal F_{2,0.1}(F_0)$ with $F_0\sim N(0,1)$.}
	\begin{center}
		\includegraphics[width=7.5cm]{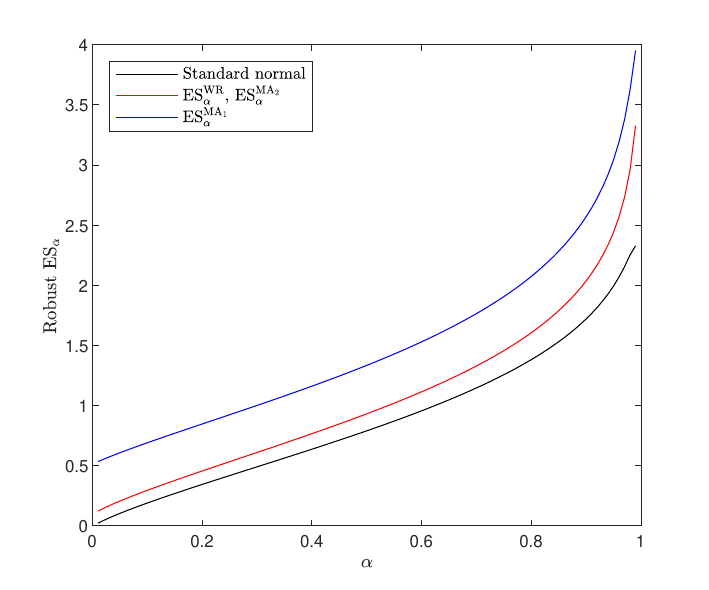}
		\includegraphics[width=7.5cm]{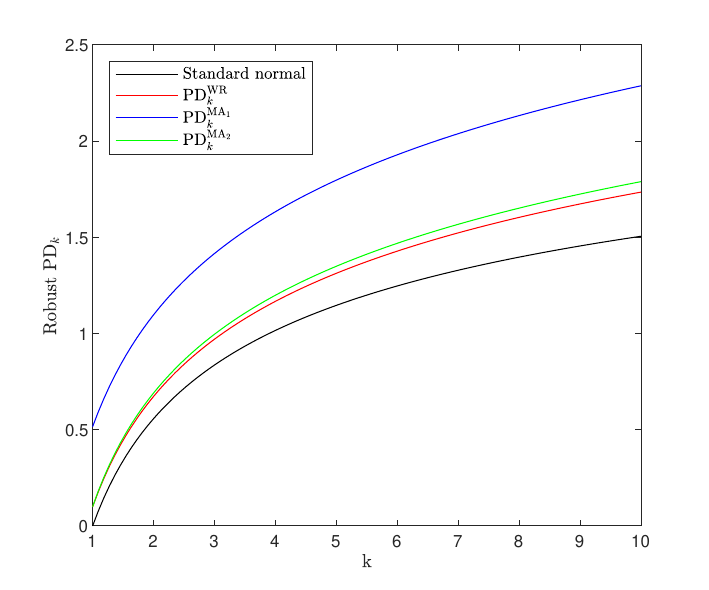}
		\label{ES_power_W}
		\medskip
	\end{center}
\end{figure}

\begin{figure}[t]
	\caption{The WR and MA approaches under mean-variance uncertainty $\mathcal F_{0,1}$.}
	\begin{center}
		\includegraphics[width=7.5cm]{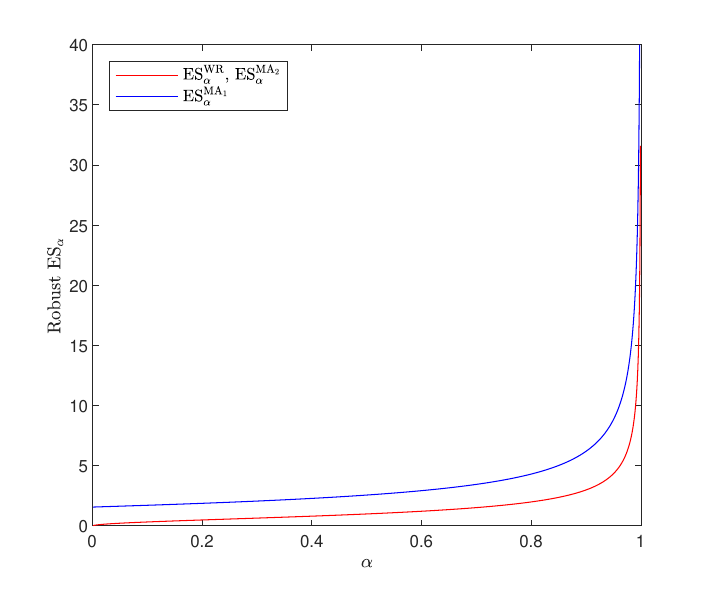}
		\includegraphics[width=7.5cm]{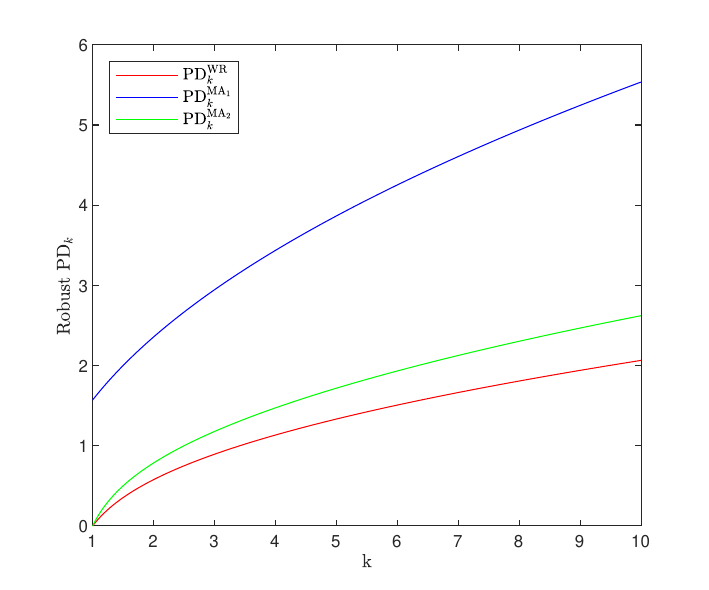}
		\includegraphics[width=7.5cm]{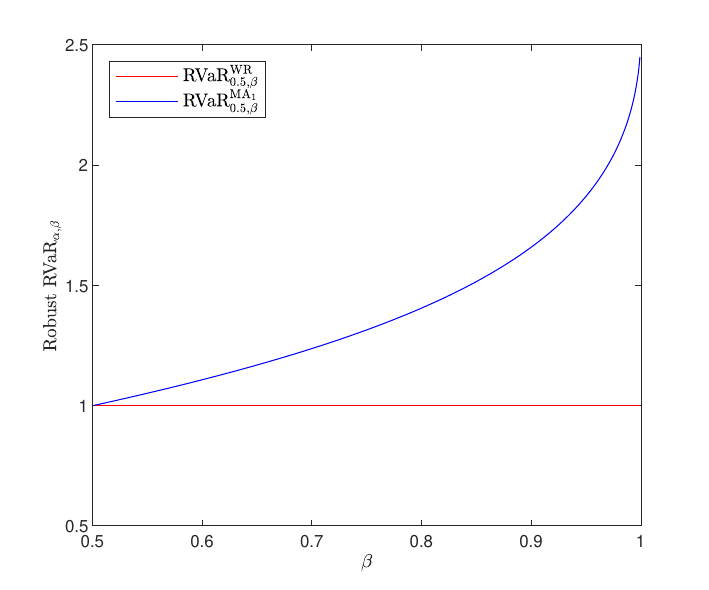}
		\includegraphics[width=7.5cm]{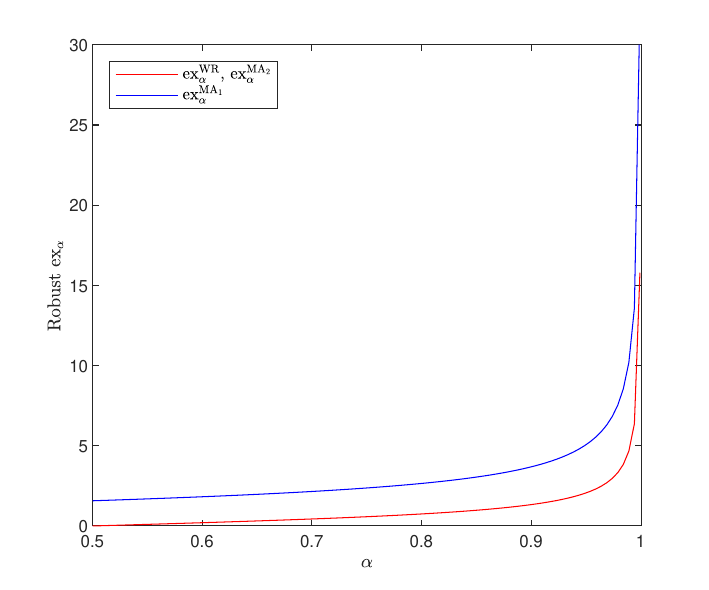}
		\label{fig-R}
		\medskip
	\end{center}
\end{figure}

	\section{Proofs and generalizations for results in    Section \ref{sec:EMAM}}
	\label{apdx-EMAM}
	\label{app:C}
	We first recall that, by Proposition \ref{prop-EMAM},
	$\preceq_1$-cEMA and $\preceq_2$-cEMA imply the properties $\preceq_1$-consistency and $\preceq_2$-consistency, respectively, and
	$\preceq_i$-cEMA ($i=1,2$) is equivalent to
	\begin{align*}
		\rho\left(\bigvee_i\{F_1,\dots,F_n\}\right)\notag
		=\sup\left\{\rho\left(\sum_{i=1}^n\lambda_iF_i\right):
		\bm\lambda\in\Delta_n\right\}
	\end{align*}
	for all $F_1,\dots,F_n\in\mathcal M$ and $n\ge 1$.
	
	%

\subsection{A generalization of Theorem \ref{th-VaRM} (a) and related results}
The following theorem is a generalized version of Theorem \ref{th-VaRM} (a) to the domain $\mathcal M_p$, $p\in [0,\infty)$.
\begin{theorem}\label{th-VaRM-G}
	Let $p\in[0,\infty)$.
	A mapping $\rho:\mathcal M_p\to\R$ satisfies translation invariance, positive homogeneity, lower semicontinuity and $\preceq_1$-cEMA
	if and only if $\rho=\VaR_\alpha$ for some $\alpha \in (0,1)$.
\end{theorem}


To prove Theorem \ref{th-VaRM-G}, we need some technical lemmas.
The following lemma
shows that a $\preceq_1$-consistent and lower semicontinuous risk measure can be uniquely extended from $\mathcal M_\infty$ to $\mathcal M_{\rm bb}$.

\begin{lemma}\label{lm-extend}
	Let $p\in [0,\infty)$ and $\rho_1, \rho_2: \mathcal M_{p}\to\R\cup\{\infty\}$ be two $\preceq_1$-consistent and lower semicontinuous risk measures that coincide on $\mathcal M_\infty$. Then $\rho_1(F)=\rho_2(F)$ for all $F\in\mathcal M_{\rm bb}\cap \mathcal M_p$.
\end{lemma}
\begin{proof}
	For $F\in\mathcal M_{\rm bb}\cap \M_p$,  {\color{black}without loss of generality assume that the support of $F$  is bounded from below by $-M$.  Define $ F_n(x) =  \sum_{k=0}^{n2^n} F(k/2^n-M) \id_{\{\frac{k-1}{2^n}\le x+M< \frac{k}{2^n} \}}+ \id_{\{x+M\ge  n\}},$ $n\in\N$.
	We have $F_n\in \mathcal M_\infty$, $F_n\preceq_1 F_{n+1}$ for all $n\ge 1$, $F_n\lawcn F$ as $n\to\infty$,} and
	$$
	\rho_1(F)\ge \limsup_{n\to\infty} \rho_1(F_n)=\limsup_{n\to\infty} \rho_2(F_n)\ge \liminf_{n\to\infty} \rho_2(F_n)\ge \rho_2(F),
	$$
	where the first inequality follows from the $\preceq_1$-consistency of $\rho_1$, and the last inequality is due to the lower semicontinuity of $\rho_2$. By symmetry, we have $\rho_1=\rho_2$.
 \end{proof}

Denote by $\ell$ the Lebesgue measure on $[0,1]$ and $\mathcal M_{1,f}(\ell)$ the space of all finitely additive probability measures on $([0,1],\mathfrak B([0,1]))$ that are absolutely continuous with respect to $\ell$. By Theorem 4.5 of \cite{J20},
for any $\preceq_1$-consistent, translation invariant and positively homogeneous risk measure $\rho:\mathcal M_\infty\to\R$, there exists a family $\{\mathcal M_\xi: \xi\in\Xi\}$ ($\Xi$ is an index set) of nonempty, ${\rm weak}^*$-compact and convex subsets of $\mathcal M_{1,f}(\ell)$ such that
\begin{align}\label{eq-MR}
	\rho(F)=\min_{\xi\in\Xi}\max_{\mu\in\mathcal M_{\xi}}\int_0^1\VaR_{s}(F)\mu(\d s),~~F\in\mathcal M_\infty.
\end{align}
Applying this representation, we can establish the following lemma.
\begin{lemma}\label{prop-EMAM1}
	Let $\rho:\mathcal M_p\to\R$, $p\in[0,\infty)$, be a mapping satisfying $\preceq_1$-consistency, translation invariance, positive homogeneity and lower semicontinuity. Denote by $\widehat\rho$ the {\color{black}restriction} of $\rho$ on $\mathcal M_\infty$, i.e., $\widehat{\rho}(F)=\rho(F)$ for all $F\in\mathcal M_\infty$.  Then the following three statements are equivalent.
	\begin{itemize}
		\item[(a)] $\widehat\rho=\VaR_{\alpha}$ for some $\alpha\in(0,1)$ on $\mathcal M_\infty$.
		\item[(b)] $\widehat\rho$ satisfies $\preceq_1$-cEMA.
		\item[(c)] $\widehat\rho$ admits a representation \eqref{eq-MR} which
		satisfies $\max_{\xi\in\Xi}\min_{\mu\in\mathcal M_\xi}\mu([0,s])=\id_{\{s\ge\alpha\}}$ for some $\alpha\in(0,1)$.
	\end{itemize}
\end{lemma}
\begin{proof}
		The implication (a)$\Rightarrow$(b) is straightforward to check by Proposition \ref{prop-lattice-G}. 

		\underline{(b)$\Rightarrow$(c)}:  By the discussion before this lemma, there exists  a family $\{\mathcal M_\xi: \xi\in\Xi\}$ consisting of nonempty, ${\rm weak}^*$-compact and convex subsets of $\mathcal M_{1,f}(\ell)$  such that  $\widehat\rho$ admits a representation \eqref{eq-MR}.
		Denote by $$g_\mu(\beta )=\mu([0,\beta ]),~\beta \in[0,1],~ ~g_\xi=\min_{\mu\in\mathcal M_\xi}g_\mu~~{\rm and}~~ g_{\Xi}=\max_{\xi\in\Xi}g_{\xi}.$$
		All these three functions are nonnegative, increasing, and take value one at $1$.
		We complete the proof of (b)$\Rightarrow$(c) by verifying  the following three facts. 
		\begin{itemize}
			\item [1.]
			\emph{Right-continuity of $g_{\Xi}$}. Note that  for any $\beta \in[0,1)$,
			\begin{align}\rho(\beta \delta_0+(1-\beta )\delta_1 )
				& =\min_{\xi\in\Xi}\max_{\mu\in\mathcal M_{\xi}}\int_0^1\VaR_{s}(\beta \delta_0+(1-\beta )\delta_1 )\mu(\d s)\nonumber\\
				& =\min_{\xi\in\Xi}\max_{\mu\in\mathcal M_{\xi}}\int_{(\beta ,1]}1\mu(\d s)\nonumber\\
				& =\min_{\xi\in\Xi}\max_{\mu\in\mathcal M_{\xi}} \mu((\beta ,1]) =\min_{\xi\in\Xi}\max_{\mu\in\mathcal M_{\xi}}(1-g_\mu(\beta ))= 1-g_\Xi(\beta ).\label{eq-220104-1}
			\end{align}
			Fix $\beta \in[0,1)$. Let $\{\beta _n\}_{n\in\N}\subseteq [0,1)$ such that $\beta _n>\beta $ and $\beta _n\downarrow \beta $ as $n\to\infty$. We have  $\beta _n\delta_0+(1-\beta _n)\delta_1 \preceq_1 \beta \delta_0+(1-\beta )\delta_1$, $n\in\N$, and $\beta _n\delta_0+(1-\beta _n)\delta_1\stackrel{\rm d}\to \beta \delta_0+(1-\beta )\delta_1$ as $n\to\infty$. Therefore, by \eqref{eq-220104-1}, we have 
			\begin{align*}
				\limsup_{n\to\infty}\{1-g_\Xi(\beta _n)\}
				=\limsup_{n\to\infty}\rho(\beta _n\delta_0+(1-\beta _n)\delta_1)
				\le \rho(\beta \delta_0+(1-\beta )\delta_1)
				=1-g_\Xi(\beta ),
			\end{align*}
			where the inequality comes from $\beta _n\delta_0+(1-\beta _n)\delta_1 \preceq_1 \beta \delta_0+(1-\beta )\delta_1$ {\color{black}and $\preceq_1$-consistency of $\rho$.}
			By   lower semicontinuity of $\rho$, we have
			\begin{align*}
				\liminf_{n\to\infty}\{1-g_\Xi(\beta _n)\}
				=\liminf_{n\to\infty}\rho(\beta _n\delta_0+(1-\beta _n)\delta_1)
				\ge \rho(\beta \delta_0+(1-\beta )\delta_1)
				=1-g_\Xi(\beta ).
			\end{align*}
			Hence, we obtain $g_\Xi(\beta )=\lim_{n\to\infty}g_\Xi(\beta _n)$ which implies the right-continuity of $g_\Xi$.
			
			\item [2.] \emph{$g_{\Xi}(1-)=g_{\Xi}(1)=1$}. Assume by contradiction that $g_{\Xi}(1-)=1-\delta$ with $\delta\in(0,1]$. By the definition of $g_\Xi$, we obtain that for all $\xi\in\Xi$, there exists $\mu\in\mathcal M_\xi$ such that $g_\mu(1-)\le g_{\Xi}(1-)=1-\delta$.  Take $G\in \mathcal M_p\setminus \mathcal M_\infty$, $p\in[0,\infty)$ with support on $\R_+$. For any $M\in\R_+$, {\color{black}take $\beta=G(M) \in [0,1)$
			and define $F:=\beta \delta_0+(1-\beta ) \delta_M\in\mathcal M_\infty$. One can verify that $F\ge G$ and thus, $F\preceq_1 G$.} Therefore,
			$$ \rho(G)\ge  \rho(F)=\widehat{\rho}(F)=\min_{ \xi\in\Xi}\max_{\mu\in\mathcal M_{\xi}}\int_0^1\VaR_{s}(F)\mu(\d s)= M(1-g_\Xi(\beta ))\ge \delta M,
			$$
	{\color{black}where the last inequality follows from 	$g_\Xi(\beta )\le g_\Xi(1-)=1- \delta$ for any $\beta<1$.}	
			Since $M$ is arbitrary, this yields a contradiction to that  $\rho:\mathcal M_p\to\R$, $p\in[0,\infty)$.

			\item [3.] \emph{$g_{\Xi}(\beta )=\id_{\{\beta \ge \alpha\}}$ for some $\alpha\in(0,1)$}. Define
			\begin{align}\label{eq-1alpha*}
				\alpha=\inf\left\{\beta : g_\Xi(\beta )
				>0\right\}\in [0,1].
			\end{align}
			We assert $\alpha\in (0,1)$. To see this, first note that $g_{\Xi}(1-)=1$ which implies $\alpha<1$. We next show $\alpha>0$ by contradiction. Suppose that $\alpha=0$, which means $ g_\Xi(\beta )>0$ for each $\beta >0$.  Let $F=\beta \delta_{-a}+(1-\beta )\delta_1$ and $G=\delta_0$ with $\beta \in(0,1)$ and $a>0$. We can calculate the MA robust risk value
			\begin{align*}
				\widehat{\rho}\left(\bigvee_1\{F,G\}\right)=\widehat{\rho}(\beta \delta_0+(1-\beta )\delta_1)
				=1-g_{\Xi}(\beta ),
			\end{align*}
			and the WR robust risk value
			\begin{align*}
				\sup_{\lambda\in[0,1]}\widehat{\rho}\left(\lambda F+(1-\lambda)G\right)
				&=\sup_{\lambda\in[0,1]}\widehat{\rho}(\lambda \beta \delta_{-a}+(1-\lambda)\delta_0+\lambda(1-\beta )\delta_1)\\
				&=\sup_{\lambda\in[0,1]}\min_{\xi\in\Xi}\max_{\mu\in\mathcal M_{\xi}}\{-a  g_\mu(\lambda \beta )+1-g_\mu(1-\lambda(1-\beta ))\}\\
				&=1-\inf_{\lambda\in[0,1]}\max_{\xi\in\Xi}\min_{\mu\in\mathcal M_{\xi}}\{a  g_\mu(\lambda \beta )+g_\mu(1-\lambda(1-\beta ))\}\\
				&\le 1-\inf_{\lambda\in[0,1]}\max_{\xi\in\Xi}\{a \min_{\mu\in\mathcal M_{\xi}} g_\mu(\lambda \beta )+\min_{\mu\in\mathcal M_{\xi}}g_\mu(1-\lambda(1-\beta ))\}\\
				&= 1-\inf_{\lambda\in[0,1]}\max_{\xi\in\Xi}\{a\,  g_{\xi}(\lambda \beta )+ g_{\xi}(1-\lambda(1-\beta ))\}.
			\end{align*}
			The property $\preceq_1$-cEMA implies $\widehat{\rho}\left(\bigvee_1\{F,G\}\right)=\sup_{\lambda\in[0,1]}\widehat{\rho}\left(\lambda F+(1-\lambda)G\right)$, and thus, 
			\begin{align}\label{eq-VaReq1}
				g_{\Xi}(\beta )\ge \inf_{\lambda\in[0,1]}\max_{\xi\in\Xi} \{a \, g_{\xi}(\lambda \beta )+ g_{\xi}(1-\lambda(1-\beta ))\},~~\beta \in[0,1],~a>0.
			\end{align}
			Fix $\beta >0$. By \eqref{eq-VaReq1} and the definition of infimum, for each $a>0$, there exist $\epsilon_a$ and $\lambda_a\in [0,1]$ such that $\lim_{a\to\infty}{\epsilon_a}=0$ and
			\begin{align}\label{eq-220108-1}
				g_{\Xi}(\beta )\ge  \max_{\xi\in\Xi} \{a \, g_{\xi}(\lambda_a \beta )+ g_{\xi}(1-\lambda_a(1-\beta ))\} -\epsilon_a.
			\end{align}
			It follows that $g_{\Xi}(\beta )\ge  \max_{\xi\in\Xi}  a \, g_{\xi}(\lambda_a \beta )-\epsilon_a =a\,g_{\Xi}(\lambda_a \beta )-\epsilon_a$ which implies $g_{\Xi}(\lambda_a \beta )\to 0$ as $a\to\infty$. That is, $\lambda_a\to0$ as $a\to\infty$ since $g(\beta )>0$ for $\beta >0$.
			Therefore, \eqref{eq-220108-1} implies
			$g_{\Xi}(\beta ) \ge \max_{\xi\in\Xi}   g_{\xi}(1-\lambda_a(1-\beta ))  -\epsilon_a = g_{\Xi}(1-\lambda_a(1-\beta ))  -\epsilon_a.
			$
			Letting $a\to\infty$, we have $g_{\Xi}(\beta ) \ge g_{\Xi}(1-)=1
			$ for any $\beta >0$ which implies $g_{\Xi}(s)=\id_{\{s>0\}}$ for $s\in[0,1]$, and this
			yields a contradiction
			to that $g$ is right-continuous on $[0,1)$.
			Hence, we have $\alpha>0$, and thus, $\alpha\in(0,1)$.
			By the definition of $\alpha$ in \eqref{eq-1alpha*}, we have $g_{\Xi}(\beta )=0$ for all $\beta \in[0,\alpha)$.
			Fix $\beta >\alpha$, there exist $\epsilon_a$ and $\lambda_a\in [0,1]$ such that $\lim_{a\to\infty}{\epsilon_a}=0$ and \eqref{eq-220108-1} holds.
			Similarly, we have $g_{\Xi}(\lambda_a \beta )\to 0$ as $a\to\infty$ which implies $\limsup_{a\to\infty}\lambda_a \le \alpha/\beta $. It then follows that
			$g_{\Xi}(\beta ) \ge   g_{\Xi}(1-\lambda_a(1-\beta ))  - \epsilon_a.
			$
			Letting $a\to\infty$, we have
			$$g_{\Xi}(\beta ) \ge  \limsup_{a\to\infty} g_{\Xi}(1-\lambda_a(1-\beta ))
			\ge   g_{\Xi}\left(\left(1-\frac{\alpha(1-\beta )}{\beta }\right)-\right).
			$$
			Since $\beta <1-\alpha(1-\beta )/\beta $ for $\beta >\alpha$, and the sequence $\{\beta_n\}_{n\in\N}$, where $\beta_0=\beta$ and $\beta_{n+1}=1-\alpha(1-\beta_n)/\beta_n$ for $n\ge 0$, converges to 1,
			and $g_{\Xi}$ is increasing, we have $g_{\Xi}(\beta )$ takes constant on $\beta \in (\alpha,1)$, that is, $g_{\Xi}(\beta ) =g_{\Xi}(1-)=1$. By right-continuity of $g_{\Xi}$, 
			we have $g_{\Xi}(\beta)=\id_{\{\beta\ge \alpha\}}$, which completes the proof of (b)$\Rightarrow$(c).
		\end{itemize}

		\underline{(c)$\Rightarrow$(a)}: By \eqref{eq-220104-1}, under the condition of (c), for $F=\beta\delta_0 + (1-\beta)\delta_1$, we have
		$
		\rho(F) =1- g_{\Xi}(\beta)= {\rm VaR}_\alpha(F)
		$. By positive homogeneity and translation invariance of $\rho$, this implies that for any $F=\beta\delta_a + (1-\beta)\delta_b$,
{\color{black}\begin{align*}
	\rho(F)=a+ (b-a) \min_{\xi\in\Xi}\max_{\mu\in\mathcal M_{\xi}} \mu((\beta,1])= a+(b-a) {\rm VaR}_\alpha(\beta\delta_0+ (1-\beta)\delta_1) = {\rm VaR}_\alpha(F),
\end{align*}
where the second equality follows from $\min_{\xi\in\Xi}\max_{\mu\in\mathcal M_{\xi}} \mu((\beta,1]) = 1-\max_{\xi\in\Xi}\min_{\mu\in\mathcal M_{\xi}} \mu([0,\beta]) = \id_{\{s<\alpha\}}.$}

		For $F\in\mathcal M_\infty$, define $G= \beta \delta_{ {\rm VaR}_0(F)} + (1-\beta) \delta_{ {\rm VaR}_\beta (F)}$ for $\beta < \alpha$.  One can check $G \preceq_1 F$ and thus, $\rho(F)\ge \rho(G) = {\rm VaR}_\alpha(G)=  {\rm VaR}_\beta(F)$ for $\beta<\alpha$. By left-continuity of ${\rm VaR}$, we have $\rho(F) \ge  \lim_{\beta\uparrow \alpha}\VaR_\beta(F)={\rm VaR}_\alpha(F)$.
		On the other hand,
		define $H = \alpha \delta_{{\rm VaR}_\alpha(F)} + (1-\alpha) \delta_{ {\rm VaR}_1(F)}$.
		One can check $F \preceq_1 H$ and thus, $\rho(F)\le \rho(H) = {\rm VaR}_\alpha(H)=  {\rm VaR}_\alpha(F)$.
		Therefore, we have $\rho(F)= {\rm VaR}_\alpha(F)$.
  \end{proof}

\noindent{\bf Proof of Theorem \ref{th-VaRM-G}.}~
Sufficiency follows directly from Proposition \ref{prop-lattice-G}. Below we show necessity. Note that $\preceq_1$-cEMA implies $\preceq_1$-consistency, and hence,
by Lemma \ref{prop-EMAM1}, we have $\rho=\VaR_{\alpha}$ for some $\alpha\in(0,1)$ on $\mathcal M_\infty$.
It remains to show that $\rho=\VaR_{\alpha}$ on $\mathcal M_p$, $p\in[0,\infty)$. By Lemma \ref{lm-extend}, we have  $\rho=\VaR_{\alpha}$ on $\mathcal M_{\rm bb}\cap \M_p$.
Next, we will prove $\rho = \VaR_\alpha$ on $\mathcal M_p\setminus \mathcal M_{\rm bb}$.
Take $F\in\mathcal M_p\setminus \mathcal M_{\rm bb}$ and $\zeta<\VaR_\alpha(F)$. We have
\begin{align}\label{eq-VaReq4}
	\rho\left(\bigvee_1\{F,\delta_{\zeta}\}\right)
	=\VaR_{\alpha}\left(\bigvee_1\{F,\delta_{\zeta}\}\right)
	=\VaR_{\alpha}(F)\vee\zeta
	=\VaR_{\alpha}(F),
\end{align}
where the first equality is due to $\bigvee_1\{F,\delta_{\zeta }\}\in\mathcal M_{\rm bb}\cap \M_p$ and the second equality follows from Proposition \ref{prop-lattice-G}. By $F\preceq_1\bigvee_1\{F,\delta_{\zeta}\}$ and $\preceq_1$-consistency, we have \eqref{eq-VaReq4}  implies that
\begin{align}\label{eq-GM0}
	\rho(F)\le \VaR_{\alpha}(F)~{\rm for~all}~F\in\mathcal M_p\setminus \mathcal M_{\rm bb}.
\end{align}
By $\preceq_1$-cEMA  of $\rho$, and together with \eqref{eq-VaReq4}, we have
\begin{align}\label{eq-VaReq5}
	\sup_{0\le\lambda\le 1}\rho(\lambda F+(1-\lambda)\delta_{\zeta})=\VaR_\alpha(F)~~{\rm for~}\zeta<\VaR_\alpha(F).
\end{align}
We assert that if $F$ is continuous at $\VaR_\alpha(F)$, then
\eqref{eq-VaReq5} implies
\begin{align}\label{eq-limlambda}
	\limsup_{\lambda\to 1}\rho(\lambda F+(1-\lambda)\delta_\zeta)=\VaR_\alpha(F),~~~\zeta<\VaR_\alpha(F).
\end{align}
To see this,   fix $\epsilon>0$ and $\zeta<\VaR_\alpha(F)$, and let $\lambda\in[0,1-\epsilon)$.
Since $F$ is continuous at  $\VaR_\alpha(F)$,  we have $\zeta':=\VaR_{(\alpha-\epsilon)/(1-\epsilon)}(F)<\VaR_\alpha(F)$.
Hence, 
$\lambda F +(1-\lambda)\delta_{\zeta} $ has probability at least $ \alpha-\epsilon +\epsilon \ge \alpha$ for the interval  $(-\infty,\zeta'\vee \zeta]$. Therefore, we have $\VaR_{\alpha}(\lambda F+(1-\lambda)\delta_{\zeta}) \le \zeta'\vee \zeta <\VaR_\alpha(F)$. Hence,
$$
\sup_{0\le\lambda\le 1-\epsilon}\rho(\lambda F+(1-\lambda)\delta_{\zeta})\le \zeta'\vee \zeta <\VaR_\alpha(F).
$$
Therefore,
the supremum in \eqref{eq-VaReq5} is not attained on $[0,1-\epsilon)$ for any $\epsilon>0$,
and we have
\begin{align*}
	\limsup_{\lambda\to 1} \rho(\lambda F+(1-\lambda)\delta_\zeta)\vee\rho(F)=\VaR_\alpha(F),~~~\zeta<\VaR_\alpha(F).
\end{align*}
By lower semicontinuity of $\rho$, we have
\begin{align*}
	\rho(F)\le\liminf_{\lambda\to 1} \rho(\lambda F+(1-\lambda)\delta_\zeta)\le\limsup_{\lambda\to 1} \rho(\lambda F+(1-\lambda)\delta_\zeta).
\end{align*}
Combining above two equations, the assertion \eqref{eq-limlambda} is verified. In the following, we will show that $\rho(G)=\VaR_\alpha(G)$ for $G\in\mathcal M_p\setminus \mathcal M_{\rm bb}$ in two cases by applying \eqref{eq-GM0} and \eqref{eq-limlambda}.\\
\underline{Case 1}: \emph{$G$ is continuous at   $\VaR_\alpha(G)$}.
By \eqref{eq-GM0}, $\rho(G)\le \VaR_\alpha(G)$, and we suppose by contradiction that $\rho(G)<\VaR_\alpha(G)$. Since $G$ is continuous at $\VaR_\alpha(G)$, there exist $x^*\in(\rho(G),\VaR_\alpha(G))$ such that $G$ is continuous at $x^*$, and $x^*$ is an element of the support of $G$, which implies
$\VaR_{G(x^*)}(G)=x^*$. Noting that $x^*<\VaR_{\alpha}(G)$, we have $\lambda^*:=G(x^*)/\alpha\in(0,1)$.
Define a cdf as
\begin{align*}
	H(x)=\begin{cases}
		\frac{1}{\lambda^*}G(x)\wedge 1,~~~& x<\VaR_\alpha(G)\\
		1,~~~&x\ge \VaR_\alpha(G).
	\end{cases}
\end{align*}
One can check that $H$ is continuous at $x^*$ and $\VaR_{\alpha}(H)=x^*$.
By \eqref{eq-limlambda}, we obtain for $\zeta<x^*$,
\begin{align}\label{eq-cont}
	\limsup_{\lambda\to 1}\rho(\lambda H+(1-\lambda)\delta_\zeta)=\VaR_\alpha(H)=x^*.
\end{align}
For $\zeta<x^*$ and $\lambda\in[\lambda^*,1]$, we have
$\lambda H+(1-\lambda) \delta_\zeta\ge G$ pointwise, which implies $\lambda H+(1-\lambda) \delta_\zeta\preceq_1 G$. It follows from $\preceq_1$-consistency that $\rho(G)\ge \rho(\lambda H+(1-\lambda)\delta_\zeta)$ for all $\zeta<x^*$ and $\lambda\in[\lambda^*,1]$. Hence, by \eqref{eq-cont}, we obtain $\rho(G)\ge x^*$, and this yields a contradiction.\\
\underline{Case 2}: \emph{$G$ has a jump at $\VaR_\alpha(G)$}. In this case, we can construct a sequence $\{G_n\}_{n\in\N}$ such that $G_n\preceq_1 F$ for all $n\in\N$, $G_n$ is continuous at point $\VaR_{\alpha}(G_n)$ and $\VaR_{\alpha}(G_n)\to\VaR_{\alpha}(G)$. By Case 1 and $\preceq_1$-consistency of $\rho$, we have $\rho(G)\ge \rho(G_n)\to\VaR_{\alpha}(G)$. Note that the converse direction holds by \eqref{eq-GM0}. Hence, we obtain $\rho(G)=\VaR_{\alpha}(G)$ for all $F\in\mathcal M_p\setminus\mathcal M_{\rm bb}$ such that $G$ has a jump at point $\zeta$.

In summary, we compete the proof of this theorem.\qed

\subsection{A generalization of Theorem \ref{th-ESM}  (b)  and related   results  }
The following theorem is a generalized result of Theorem \ref{th-ESM} (b) for the space $\mathcal M_p$, $p\in [1,\infty)$.
\begin{theorem}\label{th-ESM-G}
	A mapping $\rho:\mathcal M_p\to\R$, $p\in[1,\infty)$, satisfies  translation invariance, positive homogeneity, lower semicontinuity and $\preceq_2$-cEMA if and only if $\rho=\ES_{\alpha}$ for some $\alpha\in(0,1)$.
\end{theorem}

By Theorem \ref{prop-cxES}, we know that $\ES$ satisfies $\preceq_2$-cEMA.
In order to prove the necessity of Theorem \ref{th-ESM-G},
we need to apply Corollary 5.9 of \cite{J20}.
Define $$\mathcal H:=\{h:[0,1]\to[0,1]:~h~{\rm is~increasing~convex},~h(0)=0,~h(1)=1\}.$$
For any $\preceq_2$-consistent, translation invariant  and positively homogeneous risk measure $\rho:\mathcal M_\infty\to\R$, there exists a family $\{\mathcal H_{\xi}:\xi\in\Xi\}$ of nonempty, compact and convex subsets of $\mathcal H$ such that
\begin{align}\label{eq-MRC}
	\rho(F)=\min_{\xi\in\Xi}\max_{h\in\mathcal H_{\xi}}\int_0^1 \VaR_{s}(F)\d h(s),~~F\in\mathcal M_\infty.
\end{align}
\textcolor{black}{Here, compactness is defined with respect to the weak topology induced by all continuous functions on $[0,1]$.}
As pointed out by \cite{J20}, both the $\min$ and $\max$ can be attained, that is, for each $F\in\mathcal M_\infty$, there exists   $h\in \mathcal H_{\xi}$ for some $\xi\in \Xi$ such that $\rho(F)=\int_0^1 \VaR_{s}(F)\d h(s)$. Therefore, for any $\beta\in [0,1]$, by  $
\rho(\beta \delta_0 + (1-\beta)\delta_1)=\min_{\xi\in\Xi}\max_{h\in\mathcal H_{\xi}}\int_\beta^1 1\d h(s) =1-\max_{\xi\in \Xi}\min_{h\in\mathcal H_\xi}h(\beta),
$ we can 
define\begin{align} \label{Eq-220108-5}
	h_\xi(\beta)=\min_{h\in\mathcal H_\xi}h(\beta)~~{\rm  and}~~ h_{\Xi}(\beta)=\max_{\xi\in\Xi}h_{\xi}(\beta),~~\beta\in [0,1].\end{align}
Applying this representation, we can establish the following lemma.
\begin{lemma}\label{prop-EMAM2}
	Let $\rho:\mathcal M_p\to\R$, $p\in[1,\infty)$, be a mapping satisfying $\preceq_2$-consistency, translation invariance and positive homogeneity. Denote by $\widehat\rho$ the {\color{black}restriction}  of $\rho$ on $\mathcal M_\infty$. Then the following three statements are equivalent.
	\begin{itemize}
		\item[(a)] $\widehat\rho=\ES_{\alpha}$ for some $\alpha\in[0,1)$ on $\mathcal M_\infty$.
		\item[(b)] $\widehat\rho$ satisfies $\preceq_2$-cEMA.
		\item[(c)] $\widehat\rho$ admits a representation \eqref{eq-MRC} which
		satisfies $h_{\Xi}(\beta )=(\beta -\alpha)_+/(1-\alpha)$ for some $\alpha\in[0,1)$ where $h_{\Xi}$ is defined by \eqref{Eq-220108-5}.
	\end{itemize}
\end{lemma}

\begin{proof}
	The implication (a)$\Rightarrow$(b) follows immediately from Theorem \ref{prop-cxES} (b).
	
	\underline{(b)$\Rightarrow$(c)}: By the discussion before this lemma, there exists  a family $\{\mathcal H_{\xi}:\xi\in\Xi\}$ of nonempty, compact and convex subsets of $\mathcal H$ such that  $\widehat\rho$ admits a representation \eqref{eq-MRC}.  Define  $h_{\xi}$
	and $h_{\Xi}$  by \eqref{Eq-220108-5}. One can check that both $h_\xi$ and $h_\Xi$ are  increasing and satisfy  $h_\xi(0)=h_\Xi(0)=0$  and $h_\xi(1)=h_\Xi(1)=1.$
	In the following, we show (c) by verifying the following facts.\\ 
	1. 
	\emph{$h_{\Xi}(1-)=h_{\Xi}(1)=1$}. This  can be showed by {\color{black}similar arguments as} in (b)$\Rightarrow$(c) of Lemma \ref{prop-EMAM1}.\\
	2. \emph{$h_{\Xi}(\beta )=(\beta -\alpha)_+/(1-\alpha)$ for some $\alpha\in[0,1)$}.
	Let $F=\beta \delta_{-a}+(1-\beta )\delta_1$ and $G=\delta_0$ with $\beta \in(0,1)$ and $a>(1-\beta )/\beta $, and calculate the MA robust risk value
	\begin{align*}
		\widehat{\rho}\left(\bigvee_2\{F,G\}\right)=\widehat{\rho}
		\left(\beta \delta_{-\frac{1-\beta }{\beta }}+(1-\beta )\delta_1\right)
		=1-\frac{h_\Xi(\beta )}{\beta },
	\end{align*}
	and the WR robust risk value
	\begin{align*}
		\sup_{\lambda\in[0,1]}\widehat{\rho}\left(\lambda F+(1-\lambda)G\right)
		&=\sup_{\lambda\in[0,1]}\widehat{\rho}(\lambda \beta \delta_{-a}+(1-\lambda)\delta_0+\lambda(1-\beta )\delta_1)\\
		&=\sup_{\lambda\in[0,1]}\min_{\xi\in\Xi}\max_{h\in\mathcal H_{\xi}}\{-a\, h(\lambda \beta )+1-h(1-\lambda(1-\beta))\}\\
		&=1-\inf_{\lambda\in[0,1]}\max_{\xi\in\Xi}\min_{h\in\mathcal H_{\xi}}\{a\, h(\lambda \beta )+h(1-\lambda(1-\beta))\}\\
		&\le 1-\inf_{\lambda\in[0,1]}\max_{\xi\in\Xi}\{a\, \min_{h\in\mathcal H_{\xi}} h(\lambda \beta )+ \min_{h\in\mathcal H_{\xi}} h(1-\lambda(1-\beta))\}\\
		&= 1-\inf_{\lambda\in[0,1]}\max_{\xi\in\Xi}\{a\, h_{\xi}(\lambda \beta )+ h_{\xi}(1-\lambda(1-\beta))\}.
	\end{align*}
	The property $\preceq_2$-cEMA implies $\sup_{\lambda\in[0,1]}\widehat{\rho}\left(\lambda F+(1-\lambda)G\right)=\widehat{\rho}\left(\bigvee_2\{F,G\}\right)$, and thus,
	\begin{align}\label{eq-ESeq2}
		\frac{h_{\Xi}(\beta )}{\beta }\ge \inf_{\lambda\in[0,1]}\max_{\xi\in\Xi}\{a\, h_{\xi}(\lambda \beta )+ h_{\xi}(1-\lambda(1-\beta))\},~~\beta \in[0,1],~a>\frac{1-\beta }{\beta }.
	\end{align}
	Define
	\begin{align}\label{eq-2alpha*}
		\alpha=\inf\left\{\beta : h_\Xi(\beta )
		>0\right\}\in [0,1),
	\end{align}
	where the fact   $\alpha<1$ comes from   $h_{\Xi}(1-)=1$.
	Fix $\beta >\alpha$. By \eqref{eq-ESeq2} and the definition of infimum, for each $a>{(1-\beta) }/{\beta }$, there exist $\lambda_a\in [0,1]$ and $\epsilon_a$ such that $\lim_{a\to\infty} \epsilon_a=0$ and
	\begin{align} \label{eq-220108-3}
		\frac{h_{\Xi}(\beta )}{\beta } \ge  \max_{\xi\in\Xi}\{a h_{\xi}(\lambda_a \beta )+ h_{\xi}(1-\lambda_a(1-\beta ))\}-\epsilon_a.
	\end{align}
	It follows that ${h_{\Xi}(\beta )}/{\beta } \ge  \max_{\xi\in\Xi}a\, h_{\xi}(\lambda_a \beta )-\epsilon_a=a\, h_{\Xi}(\lambda_a \beta )-\epsilon_a.$ Letting $a\to\infty$, we obtain  $\lim_{a\to\infty}h_{\Xi}(\lambda_a \beta )=0$. By definition of $\alpha$,  $\limsup_{a\to\infty} \lambda_a\le \alpha/\beta .$
	Again, by \eqref{eq-220108-3}, we have
	\begin{align}  \label{eq-220108-31}
		\frac{h_{\Xi}(\beta )}{\beta } &\ge  \max_{\xi\in\Xi}\{h_{\xi}(1-\lambda_a(1-\beta ))\}-\epsilon_a  = h_{\Xi}(1-\lambda_a(1-\beta ))-\epsilon_a
	\end{align}
	By monotonicity of $h_{\Xi}$ and $\limsup_{a\to\infty} \lambda_a\le \alpha/\beta $, we get $\limsup_{a\to\infty}h_{\Xi}(1-\lambda_a(1-\beta )) -\epsilon_a \ge   h_{\Xi}((1-\alpha(1-\beta )/\beta )-)$. This combined with \eqref{eq-220108-31} implies
	${h_{\Xi}(\beta )}/{\beta }\ge h_{\Xi}((1-\alpha(1-\beta )/\beta )-).$ Denote by $h_{\Xi}^-(x) =\lim_{y\uparrow x}h_{\Xi}(y)$. We have $h_{\Xi}^-(\alpha)=0$ and
	\begin{align}\label{eq-EScase3} \frac{h_{\Xi}^-(\beta )-h_{\Xi}^-(\alpha)}{\beta -\alpha}\ge \frac{h_{\Xi}^-(\alpha+(1-\alpha/\beta )) -h_{\Xi}^-(\alpha)}{1-\alpha/\beta },~~\beta >\alpha.
	\end{align}
	Letting $\beta _0=\beta $ and $\beta _{n+1}=1+\alpha-\alpha/\beta _n$ for $n\ge0$, we have $\beta _n=\frac{\alpha(1-\beta )+(\beta -\alpha)
		\alpha^{-n+1}}{(1-\beta )+(\beta -\alpha)\alpha^{-n+1}}\uparrow 1$ as $n\to\infty$. Combining with \eqref{eq-EScase3}, we obtain
	$$
	\frac{h_\Xi^-(\beta )-h_\Xi^-(\alpha)}{\beta -\alpha}\ge \frac{h_\Xi^-(\beta _n)-h_\Xi^-(\alpha)}{\beta _n-\alpha}\to
	\frac{h_\Xi^-(1)
		-h_\Xi^-(\alpha)}{1-\alpha}=\frac{1}{1-\alpha}~~{\rm as}~~ n\to\infty
	$$
	for all $\beta \in(\alpha,1]$. It follows  that $h_\Xi(\beta )\ge h_\Xi^-(\beta )\ge (\beta -\alpha)_+/(1-\alpha)$ for $\beta \in (\alpha,1]$.
	We next show $h_\Xi(\beta )\le (\beta -\alpha)_+/(1-\alpha)$ for $\beta \in (\alpha,1]$ by contradiction. Suppose that there exists $\beta ^*\in(\alpha,1)$ such that $h_\Xi(\beta ^*)> (\beta ^*-\alpha)_+/(1-\alpha)$. Noting that $h_{\Xi} (\beta ) = \max_{\xi\in\Xi}\min_{h\in\mathcal H_{\xi}} h(\beta )$, there exists $\xi_0\in \Xi$ such that
	\begin{align}\label{eq-220108-4}
		\min_{h\in\mathcal H_{\xi_0}} h(\beta ^*)>\frac{(\beta ^*-\alpha)_+}{1-\alpha}.
	\end{align}
	Meanwhile, by $h_\Xi(\beta )=0$ for $\beta <\alpha$, we have $\min_{h\in\mathcal H_{\xi}} h(\beta )=0$ for any $\beta <\alpha$ and any $\xi\in\Xi$, and thus, $\min_{h\in\mathcal H_{\xi_0}} h(\beta )=0$ for $\beta <\alpha$. This implies  that there exists $h_0\in \mathcal H_{\xi_0}$ such that $h_0(\beta )=0$ for $\beta <\alpha$. By \eqref{eq-220108-4}, we have  $ h_0(\beta ^*)>(\beta ^*-\alpha)_+/(1-\alpha)$, which,  combined with $h_0(1)=1$ and $h_0(\beta )=0$ for $\beta <\alpha$, yields a contradiction to that $h_0\in\mathcal H$ is convex.  Hence,  $h_{\Xi}(\beta )=(\beta -\alpha)_+/(1-\alpha)$, and this completes the proof of (b)$\Rightarrow$(c).
	
	\underline{(c)$\Rightarrow$(a)}:   One can check that under the condition of (c), for $F=\beta \delta_0 + (1-\beta )\delta_1$, we have
	$
	\rho(F) = 1- h_{\Xi}(\beta )= {\rm ES}_\alpha(F)
	$. By positive homogeneity and translation invariance of $\rho$,   for any $F=\beta \delta_a + (1-\beta )\delta_b$,
	$\rho(F) =  {\rm ES}_\alpha(F)$.
	For $F\in\mathcal M_\infty$, define $G= \alpha \delta_{{\rm VaR}_0(F)} + (1-\alpha) \delta_{ {\rm ES}_\alpha (F)}$. By computing the $\pi$ function, one can check $G \preceq_2 F$ and thus, $\rho(F)\ge \rho(G) = {\rm ES}_\alpha(G)=  {\rm ES}_\alpha(F)$.
	On the other hand,
	define $H = \beta  \delta_{{\rm VaR}_\alpha  (F)} + (1-\beta ) \delta_{{\rm VaR}_{1}(F)}$, where $\beta >\alpha$ satisfies $(\beta -\alpha){\rm VaR}_\alpha  (F) + (1-\beta ){\rm VaR}_{1}(F) =(1-\alpha) {\rm ES}_\alpha(F)$, that is, $ {\rm ES}_\alpha(H)=  {\rm ES}_\alpha(F)$.
	By computing the ES$_s$, $s\in [0,1]$, one can check $F \preceq_2 H$ and thus, $\rho(F)\le \rho(H) = {\rm ES}_\alpha(H)=  {\rm ES}_\alpha(F)$.
	We therefore have $\rho(F)={\rm ES}_\alpha(F)$, which completes the proof.
 \end{proof}

\noindent{\bf Proof of Theorem \ref{th-ESM-G}.}~
Translation invariance, positive
homogeneity and lower semicontinuity of $\ES_{\alpha}$, $\alpha\in(0,1)$, are well-known, and the property $\preceq_2$-cEMA of
$\ES$ follows from Theorem \ref{prop-cxES}. Conversely, note that $\preceq_2$-cEMA implies $\preceq_2$-consistency, and hence,
by Lemma \ref{prop-EMAM2}, we have $\rho=\ES_\alpha$ for some $\alpha\in[0,1)$ on $\mathcal M_\infty$. Thus, it remains to show that this representation can be extended to $\mathcal M_p$ for $p\in[1,\infty)$. To see this, for $F\in\mathcal M_p$, let $X$ be a random variable with cdf $F$.
Since the probability space is nonatomic, there exists a uniform random variable $U$ on $[0,1]$ such that $X=F^{-1}(U)$ $\p$-a.s. (see, e.g., Lemma A.28 of \cite{FS16}). Define
$$
U_n=\sum_{i=0}^{n-1}\frac{\alpha i}{n}\id_{\left\{\frac{\alpha i}{n}\le U<\frac{\alpha(i+1)}{n}\right\}}+\sum_{i=0}^{n-1}
\left(\alpha+\frac{(1-\alpha)i}{n}\right)
\id_{\left\{\alpha+\frac{(1-\alpha)i}{n}\le U<\alpha+\frac{(1-\alpha)(i+1)}{n}\right\}},~~n\ge 1.
$$
and denote by $X_n=\E[X|U_n]$. One can obatin $F_{X_n}\in\mathcal M_\infty$, and $\rho(F_{X_n})=\ES_\alpha(F_{X_n})=\ES_{\alpha}(F)$. On one hand, since $F_{X_n}\preceq_2 F$ for all $n\ge 1$, and note that $\preceq_2$-cEMA implies $\preceq_2$-consistency, we have $\rho(F_{X_n})\le \rho(F)$. Hence, we have $\ES_{\alpha}(F)=\limsup_{n\to \infty}\rho(F_{X_n})\le\rho(F)$. On the other hand, noting that $X_n\stackrel{\rm d}\to X$, it follows from the lower semicontinuitycof $\rho$ that $\ES_{\alpha}(F)=\liminf_{n\to \infty}\rho(F_{X_n})\ge\rho(F)$. Hence, we conclude that $\rho(F)=\ES_{\alpha}(F)$. Since $\ES_0=\E$ is not lower semicontinuous, we obtain $\rho=\ES_{\alpha}$ for some $\alpha\in(0,1)$.
\qed

\begin{remark}
	\label{rem:10}
	The characterization results in Theorems \ref{th-VaRM-G} and \ref{th-ESM-G} are obtained for spaces
	$\M_p$, $p\in [1,\infty)$, i.e., cdfs with finite $p$th moment.
	On the space $\M_\infty$ of compactly supported cdfs, the situation is more delicate.
	In particular,  for $\alpha \in (0,1)$ and $\lambda \in (0,1)$, we find that the mappings $\lambda \VaR_\alpha+(1-\lambda)\VaR_1$
	and $\lambda \ES_\alpha+(1-\lambda)\VaR_1$  on $\M_\infty$ satisfy the conditions in  (a) and (b) of Theorem \ref{th-VaRM}, respectively.
	These mappings are not real-valued on $\M_p$ for $p\in [1,\infty)$. A full characterization   on $\mathcal M_\infty$ seems  beyond current techniques and requires future study. This   hints at the level of technical sophistication of the theory.
\end{remark}

\section{Supplementary numerical results in Section \ref{sec:Numerical}}
\label{app:portfolio}

We present the summary statistics of the return rates of 20 stocks in Section \ref{sec:Numerical} in Table \ref{tab:basic}.

\subsection{Portfolio selection under mean-variance uncertainty}
\label{app:portfolio1}
We follow the portfolio selection setting discussed in Section \ref{sec:72} to assume that only the mean and the covariance matrix are available  to the investor. This appendix complements the study in Section \ref{sec:72} with Wasserstein uncertainty.

For a given portfolio weight $\mathbf{w}$, the uncertainty set is $\mathcal F_{\mathbf{w}^\top{\bm \mu},\sqrt{\mathbf{w}^\top\Sigma \mathbf{w}}}$,
where $\bm\mu$ is the mean vector and $\Sigma$ is the covariance matrix of   losses from the stocks as reported in or computed from Table \ref{tab:basic}.
By the results in Section \ref{sec:MV}, the optimization problem of the portfolio selection under the MA approach with $\preceq_1$ and $\preceq_2$ are
\begin{align}\label{eq-MAOP1}
	\min_{{\mathbf w}\in\Delta_{20}}:~~ \rho^{\rm MA_1}\left( \mathcal F_{{\mathbf w}^\top\bm\mu,\sqrt{{\mathbf w}^\top\Sigma {\mathbf w}}}\right)=   {\mathbf w}^\top\bm\mu+\beta_k\sqrt{{\mathbf w}^\top\Sigma {\mathbf w}}~~~~{\rm s.t.}~{\mathbf w}^\top\bm \mu\le -r_0/m,
\end{align}
and
\begin{align}\label{eq-MAOP}
	\min_{{\mathbf w}\in\Delta_{20}}:~~ \rho^{\rm MA_2}\left( \mathcal F_{{\mathbf w}^\top\bm\mu,\sqrt{{\mathbf w}^\top\Sigma {\mathbf w}}}\right)= {\mathbf w}^\top\bm\mu+\gamma_k\sqrt{{\mathbf w}^\top\Sigma {\mathbf w}}~~~~{\rm s.t.}~{\mathbf w}^\top\bm\mu\le -r_0/m,
\end{align}
respectively, where $\beta_k={(\sqrt{\pi}\Gamma\left(k+1/2\right))}/{\Gamma(k)}$ and
$\gamma_k={(\sqrt{\pi} (k-1)\Gamma\left(k+1/2\right))}/{((2k-1)\Gamma(k))},$
as in Table \ref{tab-R}, and $r_0$ is the expected annualized return and $m=250$.
Using the results of \cite{L18}, the WR portfolio optimization problem is
\begin{align}\label{eq-WCOP}
	&\min_{{\mathbf w}\in\Delta_{20}}: ~~\rho^{\rm WR} \left( \mathcal F_{{\mathbf w}^\top\bm\mu,\sqrt{{\mathbf w}^\top\Sigma {\mathbf w}}}\right)= {\mathbf w}^\top\bm\mu+\eta_k\sqrt{{\mathbf w}^\top\Sigma {\mathbf w}}~~~~{\rm s.t.}~{\mathbf w}^\top\bm\mu\le -r_0/m,
\end{align}
where $\eta_k= ({k-1})/{\sqrt{2k-1}}.$ 
Figure \ref{Robust_MV_PW_k} presents the optimal values of the optimization problem under mean-variance uncertainty with the SAA, WR and MA approaches for different values of $k$ and $r_0$ using the whole-period data. We can see that the robust value computed by the MA approach with $\preceq_1$ is always the largest one and that of  SAA  is always the smallest one; this is consistent with our intuition as MA with $\preceq_1$ is the most robust approach among them, and SAA is not conservative.

\begin{figure}[pht]
	\caption{The optimized values of ${\rm PD}_k$ under mean-variance uncertainty. Left: $r_0=0.2$ and $k\in [1,20]$; Right: 
		$k=10$ and $r_0\in  [0.15,0.5]$.}
	\begin{center}
		\includegraphics[width=7.5cm]{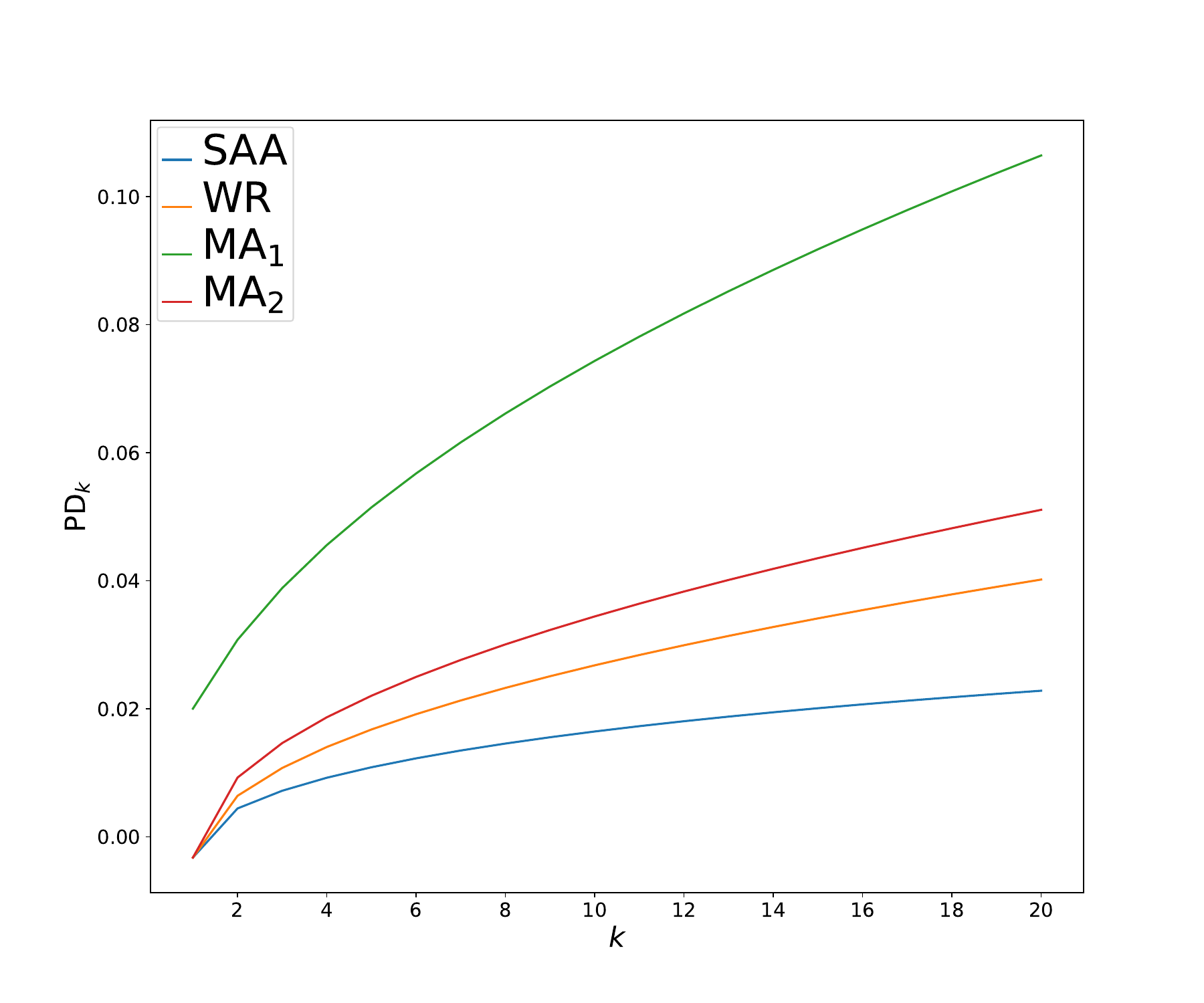}
		\label{Robust_MV_PW_k}
		\includegraphics[width=7.5cm]{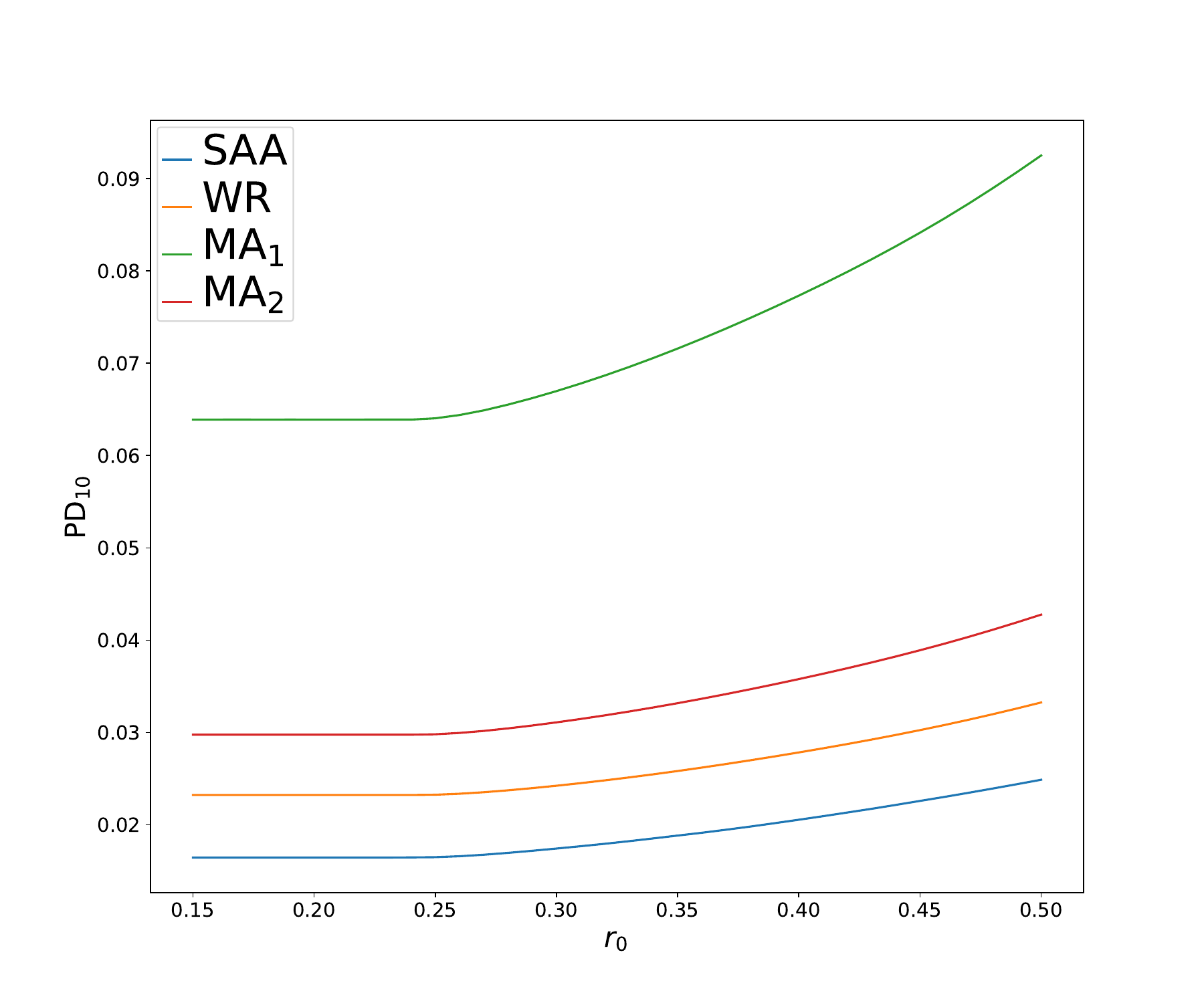}
		\label{Robust_MV_PW_mu}
		\medskip
	\end{center}
\end{figure}

Similarly to Section \ref{sec:72}, we choose 350 trading days  for the initial training,
and compute the optimal portfolio weights in each day with a rolling window. Figure \ref{rollingwindow} presents the SAA approach, Markowitz model, and WR and MA approaches under Wasserstein uncertainty set with t benchmark distribution and mean-variance uncertainty over the remaining 300  trading days with $r_0=0.2$ and $k=2$ (left) or $k=20$ (right).
In the case $k=2$, the WR and ${\rm MA}_2$ approaches exhibit comparable performance but underperform the other approaches. In the first 150 trading days, the SAA approach outperform the others, and its performance aligns closely with the ${\rm MA}_1$ approach thereafter, with the Markowitz model delivering the best performance.
In the case $k=20$, the Markowitz model stands out with superior performance, particularly in the final 100 trading days. Both the MA and WR approaches perform very similarly, trailing slightly behind the SAA approach during the first 150 trading days.


{\color{black}
Tables \ref{tab:SR} and \ref{tab:TC} present the Sharpe ratios and the
nominal transaction cost, respectively, for SAA, Markowitz, WR and MA approaches under mean-variance uncertainty as well as Wasserstein uncertainty, which has been shown in Section \ref{sec:72}.
With $r_0=0.2$ and $0.3$, the MA and WR approaches based on Wasserstein uncertainty have smaller transaction costs than the other approaches. Additionally, these approaches also exhibit larger Sharpe ratios. The 
$\rm MA_1$ approach with mean-variance uncertainty has the smallest transaction cost for $r_0=0.1$.
}

\begin{figure}[pht]
	\caption{Wealth evolution for different portfolio strategies from May 2020 to Aug 2021 with $r_0=0.2$. Left: $k=2$; Right: $k=20$.}
	\begin{center}
		\includegraphics[width=7.5cm]{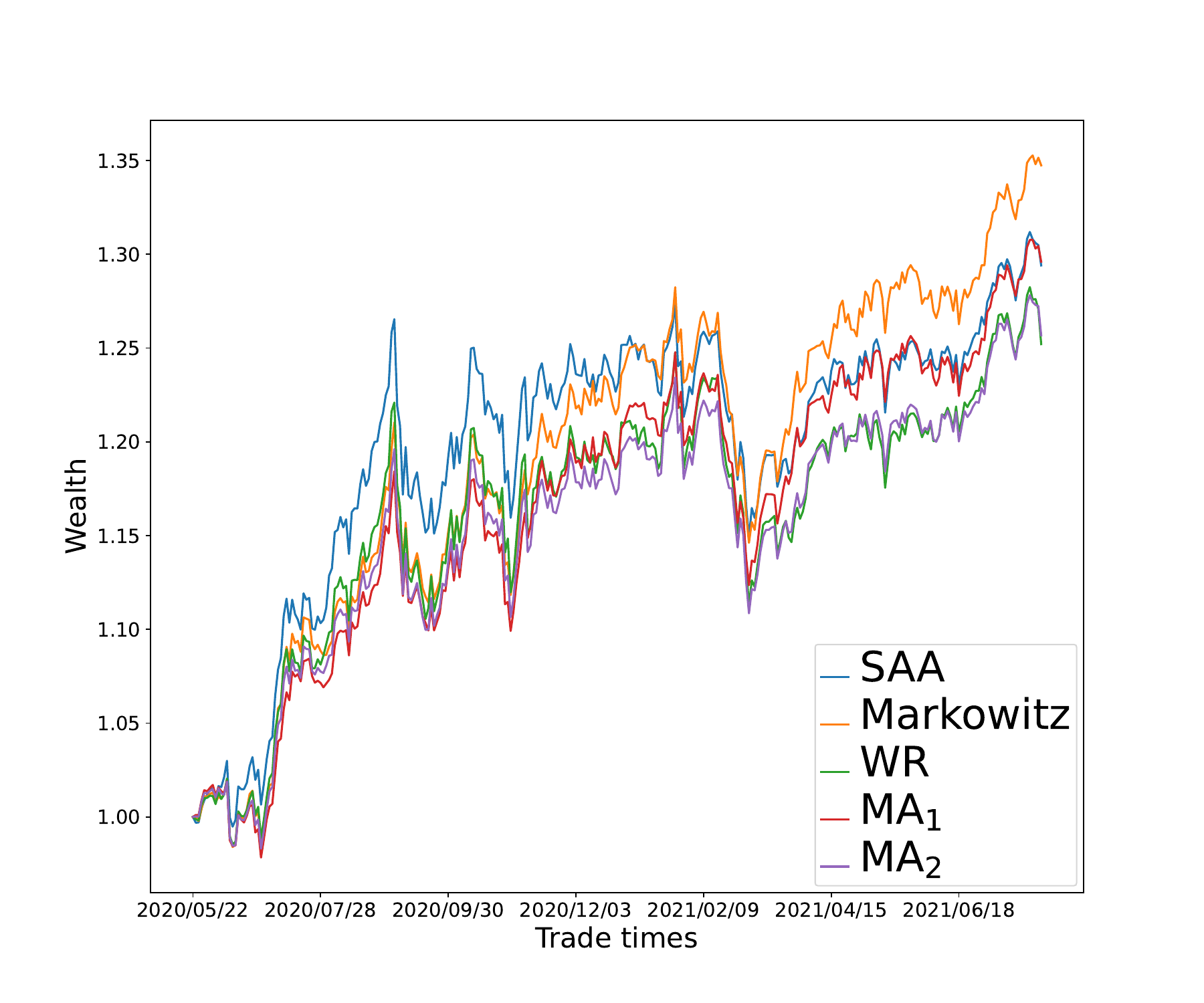}
		\includegraphics[width=7.5cm]{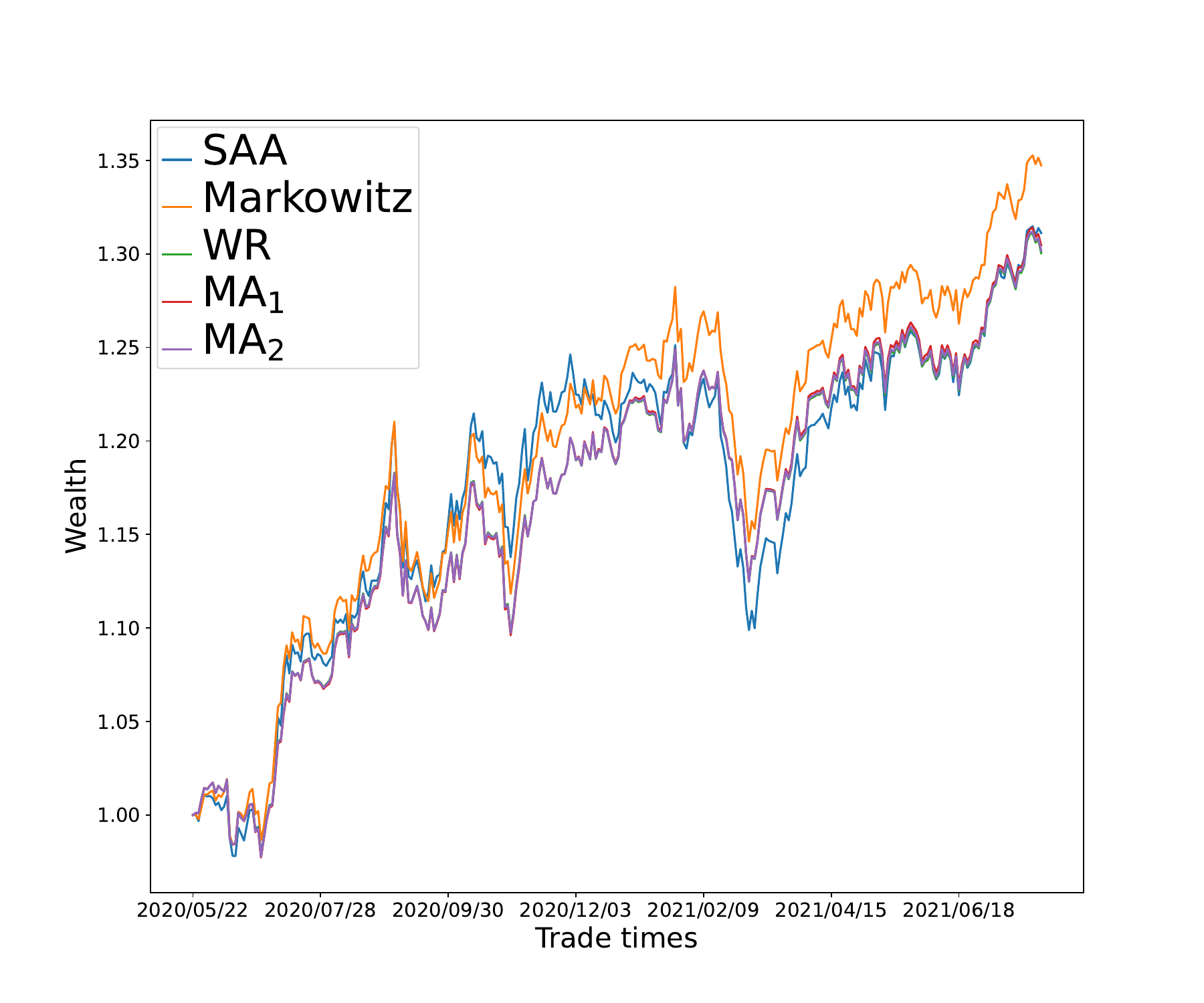}
		\label{rollingwindow}
		\medskip
	\end{center}
\end{figure}

\begin{table}[htb]
\color{black}
	\caption{Annualized return (AR), annualized volatility (AV) and Sharpe ratio (SR) for different portfolio strategies from May 2020 to Aug 2021 with $r_0=0.2$}
	\label{tab:SR}
 \renewcommand{\arraystretch}{1.2}
\begin{center}
	\begin{tabular}{lcccccc}
		\toprule
		& \multicolumn{2}{c}{AR ($\%$)} & \multicolumn{2}{c}{AV ($\%$)} & \multicolumn{2}{c}{SR ($\%$)} \\
		\cmidrule(lr){2-3} \cmidrule(lr){4-5} \cmidrule(lr){6-7}
		Approach & $ k = 2 $ & $ k = 20 $ & $ k = 2$ & $ k = 20 $ & $ k = 2$ & $ k = 20 $ \\
		\midrule
		SAA & 25.42 & 26.72 & 14.82 & 14.35 & 170.4 & 185.0 \\
	Markowitz  & 29.50 & 29.50 & 13.54 & 13.54 & 216.6 & 216.6 \\
	W-WR   & 32.75 & 30.79 & 14.25 & 13.30 & 228.7 & 230.3 \\
	W-${\rm MA}_2$ & \textbf{33.77} & \textbf{32.01} & 14.48 & 13.56 & \textbf{232.0} & \textbf{234.9} \\
MV-WR & 22.10 & 25.60 & 15.30 & 12.86 & 143.4 & 197.8 \\
	MV-${\rm MA}_1$  & 25.26 & 25.93 & \textbf{12.92} & \textbf{12.81} & 194.2 & 201.1 \\
	MV-${\rm MA}_2$ & 22.24 & 25.71 & 13.90 & 12.84 & 158.8  & 198.9 \\
		\bottomrule
\end{tabular}
\end{center}
~\\[5pt]
 \footnotesize{Note: W-WR: WR with Wasserstein uncertainty; W-${\rm MA}_2$: ${\rm MA}_2$ with Wasserstein uncertainty; MV-WR: WR with mean-variance uncertainty; MV-${\rm MA}_1$: ${\rm MA}_1$ with mean-variance uncertainty; MV-${\rm MA}_2$: ${\rm MA}_2$ with mean-variance uncertainty. }
\end{table}

\begin{table}[htb]
\centering
{\color{black}
	\caption{Nominal transaction cost $\sum_{t=1}^{T} \|\mathbf w_{t+1}-\mathbf w_{t}\|_1/T$ with $\epsilon=0.01$ and $T=299$}
	\label{tab:TC}
 \renewcommand{\arraystretch}{1.2}
	\begin{tabular}{lcccccc}
		\toprule
		& \multicolumn{2}{c}{$ r_0 = 0.1 $} & \multicolumn{2}{c}{$ r_0 = 0.2 $} & \multicolumn{2}{c}{$ r_0 = 0.3 $} \\
		\cmidrule(lr){2-3} \cmidrule(lr){4-5} \cmidrule(lr){6-7}
		Approach & $ k = 2 $ & $ k = 20 $ & $ k = 2$ & $ k = 20 $ & $ k = 2$ & $ k = 20 $ \\
		\midrule
		SAA & 0.0549 & 0.0071 & 0.0871 & 0.0121 & 0.0969 & 0.0833 \\
	Markowitz  & 0.0102 & 0.0102 & 0.0110 & 0.0110 & 0.0746 & 0.0746 \\
	W-WR   & 0.0127 & 0.0032 & 0.0127 & \textbf{0.0033} & 0.0271 & 0.0382 \\
	W-${\rm MA}_2$ & 0.0114 & 0.0035 & \textbf{0.0105} & 0.0035 & \textbf{0.0239} & \textbf{0.0348} \\
MV-WR & 0.0321 & 0.0062 & 0.0766 & 0.0074 & 0.0690 & 0.0714 \\
	MV-${\rm MA}_1$  & \textbf{0.0080} & \textbf{0.0024} & 0.0108 & 0.0047 & 0.0714 & 0.0714 \\
	MV-${\rm MA}_2$ & 0.0236 & 0.0049 & 0.0496 & 0.0065 & 0.0714  & 0.0714 \\
		\bottomrule
\end{tabular}
}
\end{table}

\begin{remark}
	\label{rem:omit-emp}
	The similar performance of MA and WR approaches for large $k$ is not a coincidence.
	In the setting of mean-variance uncertainty, as $k$ grows,   the weights $\beta_k$, $\gamma_k$ and $\eta_k$ in the optimization problems \eqref{eq-MAOP1}, \eqref{eq-MAOP} and \eqref{eq-WCOP} also grow. If these weights are large enough, then the terms involving  $\sqrt{{\mathbf w}^\top\Sigma {\mathbf w}}$  in those  problems become dominant in the optimization. For this reason,  Problems \eqref{eq-MAOP1}, \eqref{eq-MAOP} and \eqref{eq-WCOP}  are
 similar to the case that
 the mixture of mean and standard variance as the objective risk measure with a large weight on the standard variance.
\end{remark}

\subsection{Wasserstein uncertainty with normal benchmark distribution}
We follow the portfolio selection setting discussed in Section \ref{sec:72} under  Wassserstein uncertainty with a fitted normal benchmark distribution. The considered optimization problems have the same form of \eqref{eq-WWCOP} and \eqref{eq-WMAOP} with the unit variance t-distribution   replaced by the standard norm distribution.
Figure \ref{Robust_W_PW_k} presents the robust risk values of the optimization problem  with the SAA, WR and MA approaches for different values of $\epsilon$, $r_0$ and $k$ using the whole-period data. In the left panel, we see that SSA may be the largest if $\epsilon\le 0.002$, because the empirical cdf of the data may be outside the Wasserstein uncertainty set if $\epsilon$ is too small.
As seen from Theorem \ref{prop-WU}, although the multivariate normal distribution of $\mathbf{X}$ leads to a light-tailed benchmark distribution of $\mathbf w^\top \mathbf X$, the robust model we use in the MA approach is heavy-tailed.
Figure \ref{rollingwindow_W} reports the wealth process under SAA approach, and WR and MA approaches under Wasserstein uncertainty with a normal benchmark distribution as well as a t-benchmark distribution which has been shown in Section \ref{sec:72}. All robust approaches perform similarly, and they generally outperform SAA approach, especially after the first 150 trading days.
\begin{figure}[pht]
	\caption{The optimized values of ${\rm PD}_k$ under Wasserstein uncertainty. Left: $r_0=0.2$, $k=10$ and $\epsilon\in[0,0.1]$; Middle: $r_0=0.2$, $\epsilon=0.01$ and $k\in [1,20]$; Right: $k=10$, $\epsilon=0.01$ and $r_0\in  [0.15,0.5]$.}
	\begin{center}
		\includegraphics[width=5.4 cm]{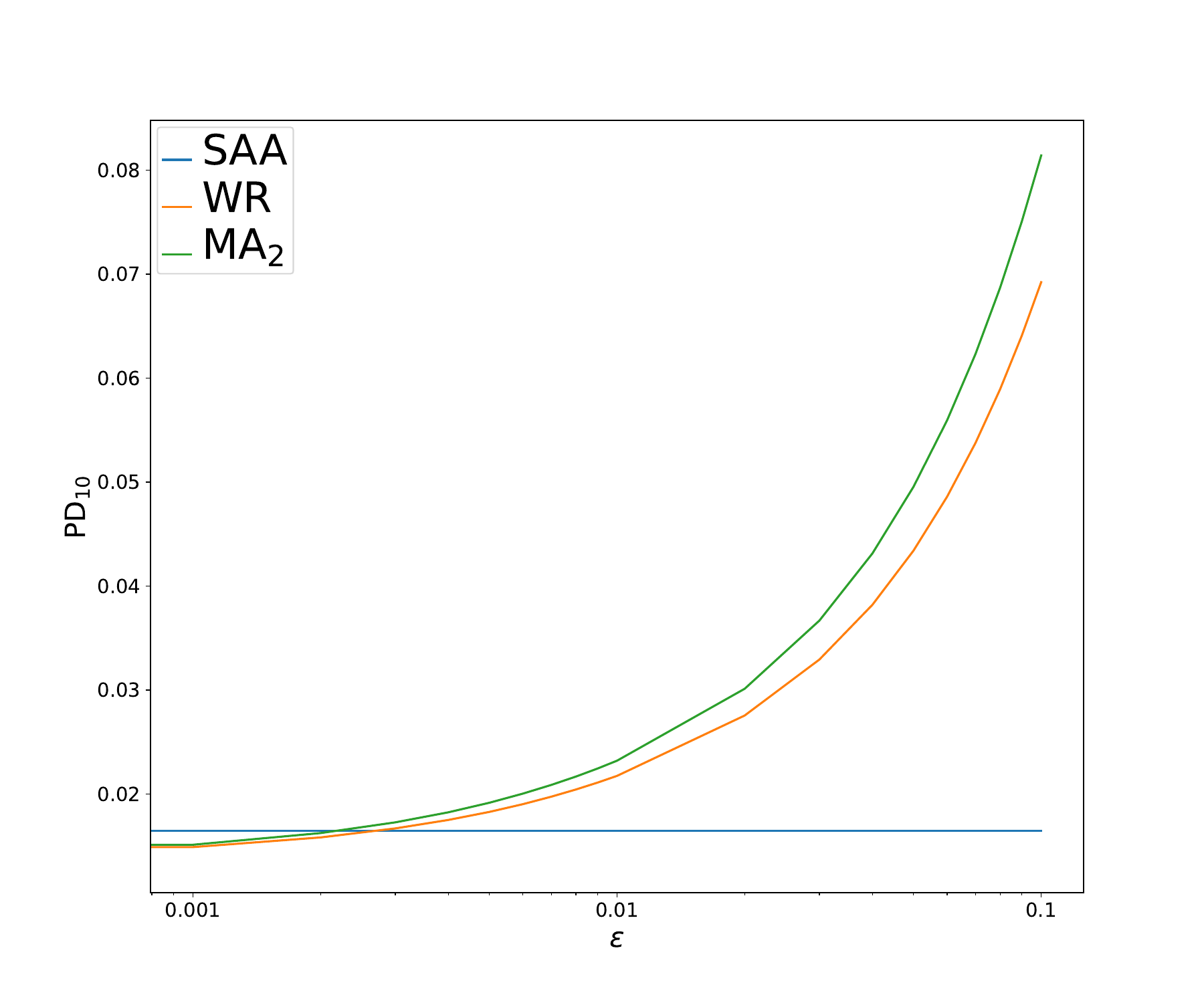}
		\label{Robust_W_PW_epsilon}
		\includegraphics[width=5.4 cm]{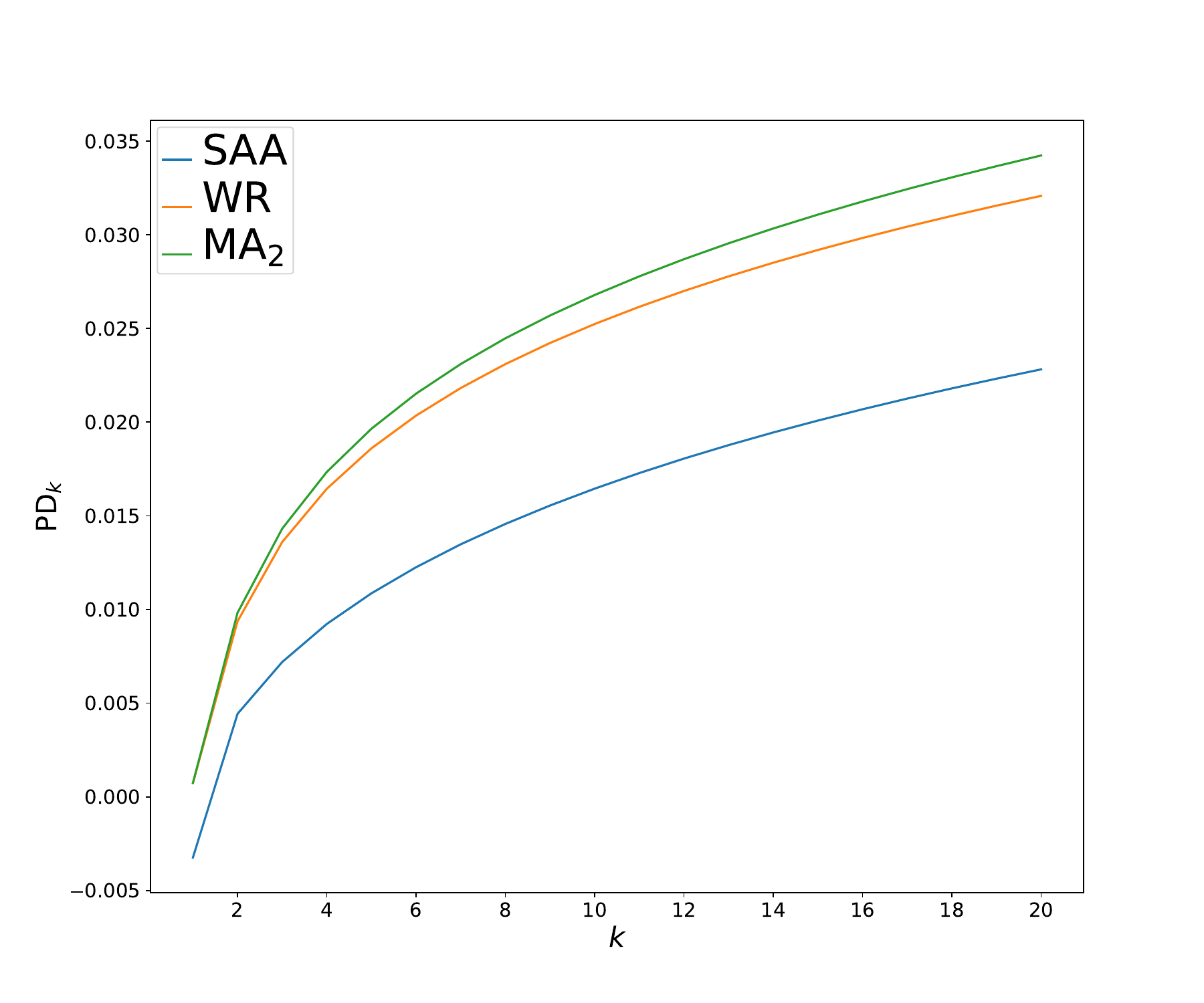}
		\label{Robust_W_PW_k}
		\includegraphics[width=5.4 cm]{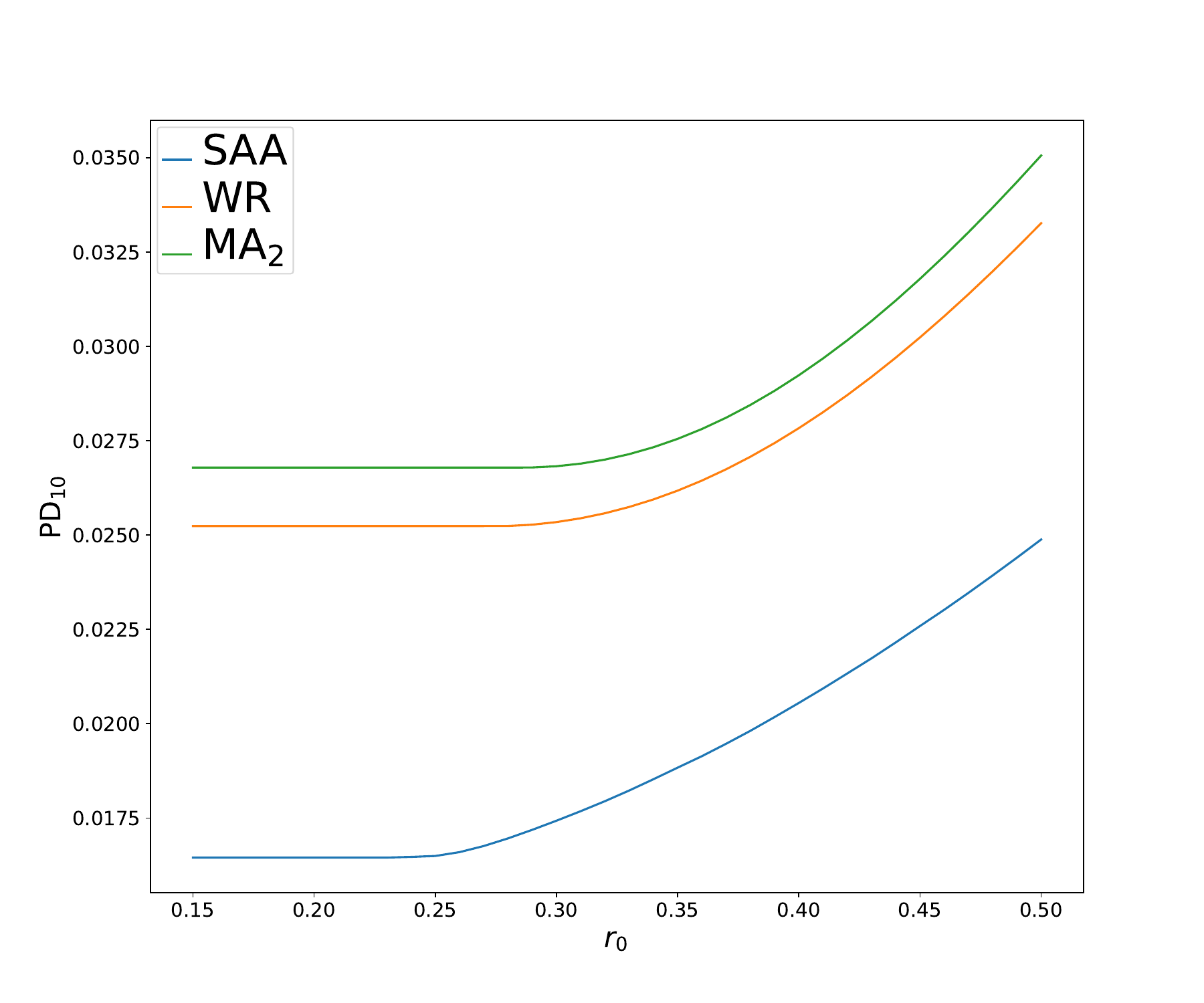}
		\label{Robust_W_PW_mu}
		\medskip
	\end{center}
\end{figure}

\begin{figure}[pht]
	\caption{Wealth evolution for different portfolio strategies from May 2020 to Aug 2021 ($\epsilon=0.01$, $r_0=0.2$). Left: $k=2$; Right: $k=20$.}
	\begin{center}
		\includegraphics[width=7.5cm]{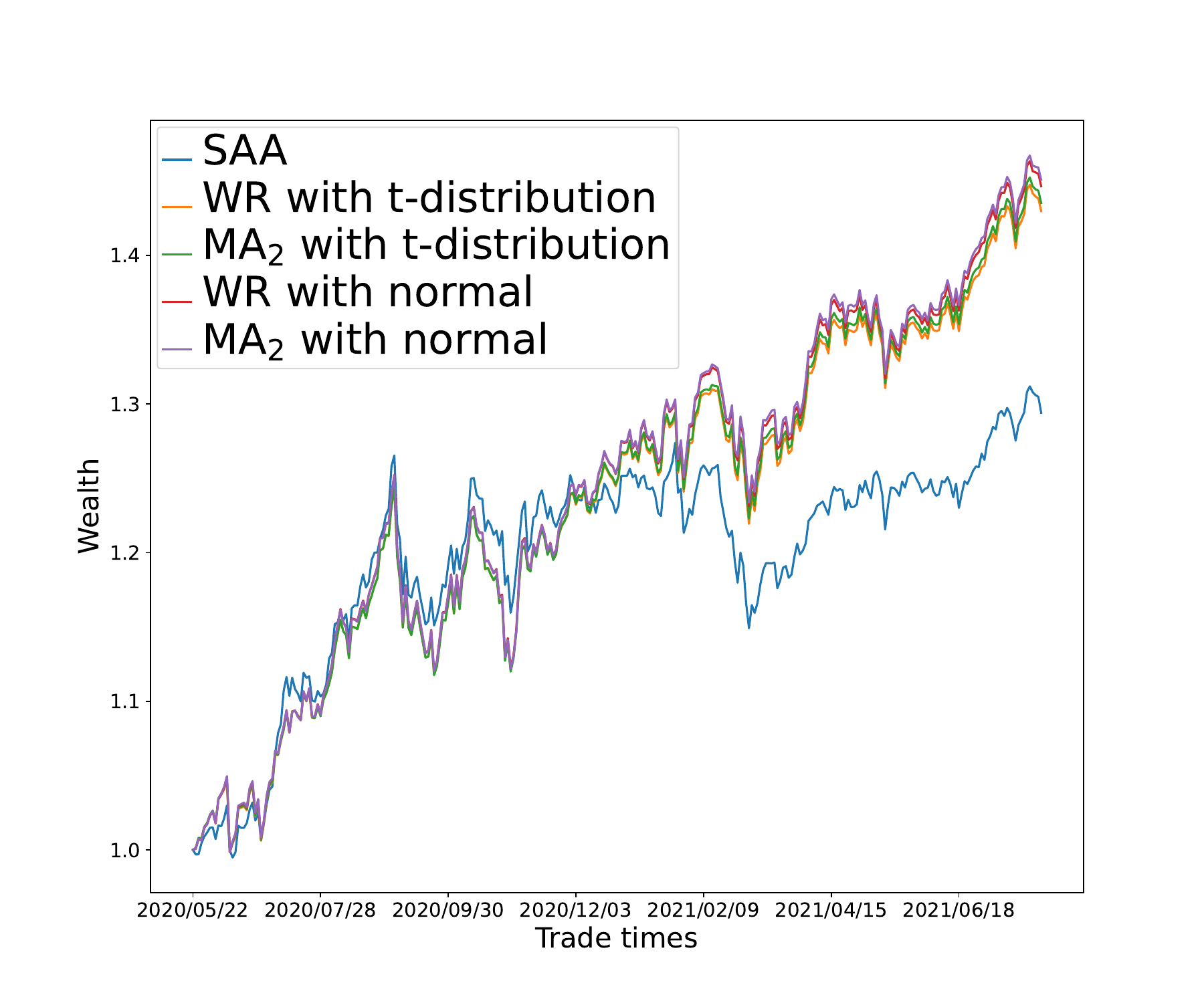}
		\includegraphics[width=7.5cm]{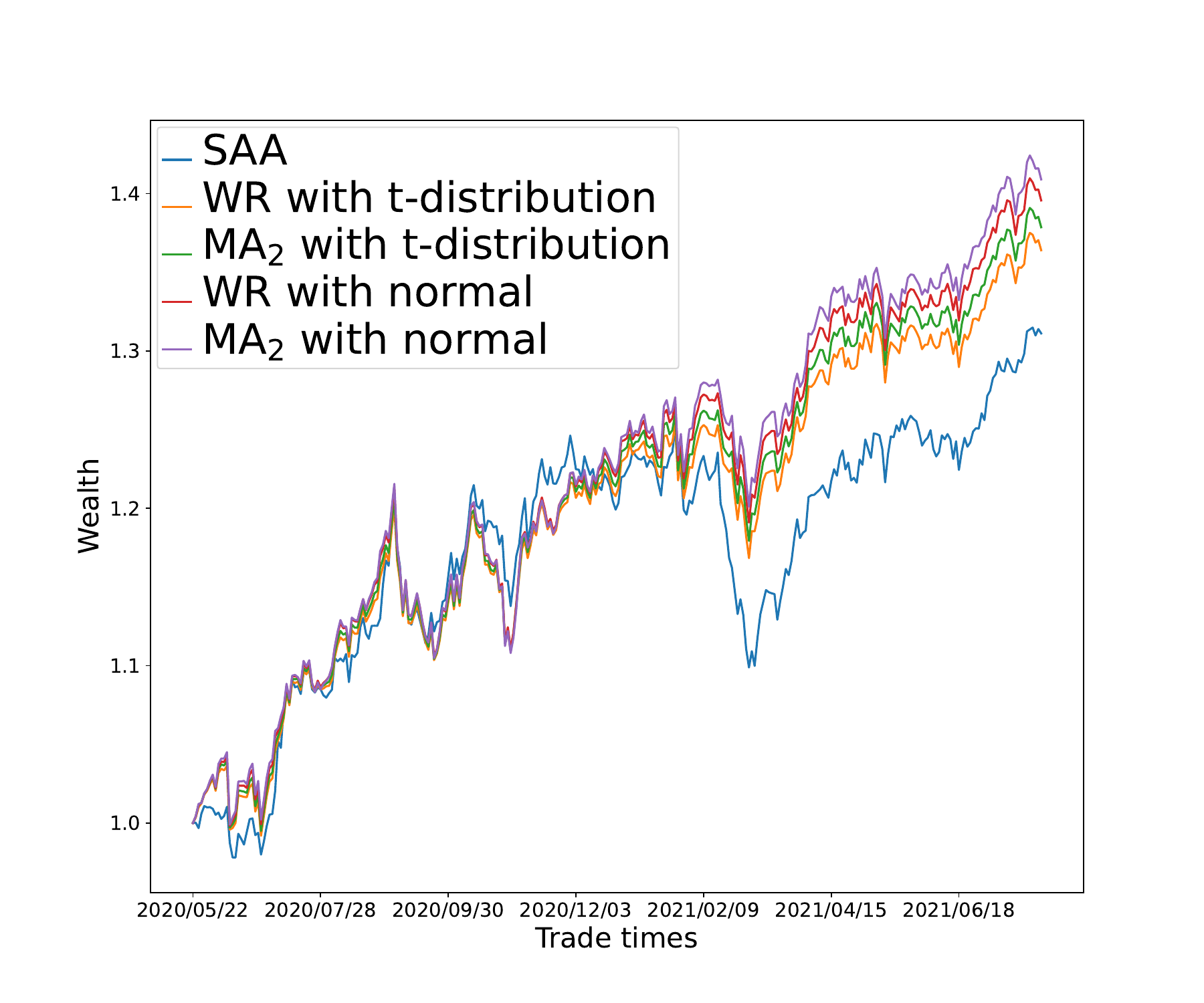}
		\label{rollingwindow_W}
		\medskip
	\end{center}
\end{figure}

\begin{center}
\begin{table}
\rotatebox{90}
{\begin{minipage}{\textheight}
\vspace{3.8cm}
		\captionof{table}{Summary statistics of daily returns of the 20 stocks, including sample mean ($\times 10^{-3}$), sample standard deviation (SD, $\times10^{-4}$), and sample correlations}
		\label{tab:basic}
		\begin{center}
			\footnotesize
			\begin{tabular}{c|cc|*{20}{c}}  
				\hline
				 & Mean & SD  & 1 & 2 & 3 & 4 & 5 & 6 & 7 & 8 & 9 & 10 & 11 & 12 & 13 & 14 & 15 & 16 & 17 & 18 & 19 & 20\\
				\hline
				1 &2.3 &5.1 & 1 & & & & & & & & & & & & & & & & & & & \\
				2 &1.8 &4.0 &0.784 & 1 & & & & & & & & & & & & & & & & & & \\
				3 &1.6 &3.8 &0.679 &0.786 & 1 & & & & & & & & & & & & & & & & & \\
				4 &1.4 &3.7 &0.660 &0.714 &0.644  & 1 & & & & & & & & & & & & & & & & \\
				5 &1.8 &5.0 & 0.707 &0.838 &0.716 &0.704 &1 & & & & & & & & & & & & & & & \\
				6 &1.3 &6.3 & 0.496 &0.546 &0.505 &0.614 &0.612 &1 & & & & & & & & & & & & & & \\
				7 &3.2 &11.6 & 0.583 &0.600 &0.519 &0.560 &0.574 &0.445 &1 & & & & & & & & & & & & & \\
				8 &1.6 &3.7 & 0.637 &0.737 &0.685 &0.443 &0.653 &0.330 &0.473 &1 & & & & & & & & & & & & \\
				9 &0.7 &2.0 & 0.485 &0.563 &0.478 &0.332 &0.429 &0.251 &0.323 &0.568 &1 & & & & & & & & & & & \\
				10 &1.3 &2.0 & 0.588 &0.647 &0.531 &0.529 &0.600 &0.405 &0.447 &0.513 &0.534 &1 & & & & & & & & & & \\
				11 &0.8 &2.1 & 0.438 &0.520 &0.403 &0.378 &0.449 &0.330 &0.334 &0.377 &0.479 &0.680 &1 & & & & & & & & & \\
				12 &0.9 &2.3 & 0.491 &0.581 &0.478 &0.371 &0.494 &0.271 &0.338 &0.538 &0.656 &0.610 &0.610 &1 & & & & & & & & \\
				13 &1.4 &5.0 & 0.626 &0.715 & 0.662&0.457 &0.628 &0.308 &0.472 &0.922 &0.534 &0.470 &0.349 &0.513 &1 & & & & & & & \\
				14 &1.1 &4.8 & 0.500 &0.580 &0.534 &0.357 &0.482 &0.292 &0.379 &0.611 &0.565 &0.483 &0.363 &0.491 &0.567 &1 & & & & & & \\
				15 &1.0 &5.1 & 0.437 &0.493 &0.498 &0.308 &0.402 &0.198 &0.307 &0.649 &0.410 &0.310 &0.274 &0.392 &0.670 &0.415 &1 & & & & & \\
				16 &1.3 &3.8 & 0.606 &0.681 &0.584 &0.424 &0.598 &0.372 &0.471 &0.654 &0.494 &0.569 &0.469 &0.548 &0.613 &0.593 &0.539 &1 & & & & \\
				17 &0.6 &6.6 & 0.584 &0.646 &0.552 &0.461 &0.587 &0.397 &0.434 &0.562 &0.442 &0.514 &0.437 &0.444 &0.548 &0.457 &0.409 &0.560 &1 & & & \\
				18 &2.1 &6.5 & 0.659 &0.742 &0.636 &0.595 &0.742 &0.445 &0.522 &0.653 &0.395 &0.513 &0.374 &0.425 &0.640 &0.428 &0.428 &0.595 &0.540 &1 & & \\
				19 &1.6 &5.7 & 0.525 &0.544 &0.535 &0.321 &0.440 &0.241 &0.364 &0.649 &0.452 &0.398 &0.321 &0.426 &0.633 &0.556 &0.631 &0.621 &0.521 &0.441 &1 & \\
				20 &0.7 &3.7 &0.479 &0.546 &0.513 &0.354 &0.450 &0.227 &0.361 &0.630 &0.563 &0.455 &0.390 &0.543 &0.623 &0.521 &0522 &0.595 &0.555 &0.391 &0.630 & 1 \\
				\hline
			\end{tabular}
		\end{center} 
	\end{minipage}
 }
	\end{table}
\end{center}

\end{appendix}

\end{document}